\documentclass[11pt]{article}
\usepackage{fullpage}
\usepackage{amsmath,amssymb,amsthm}

\usepackage[utf8]{inputenc}
\usepackage[T1]{fontenc}
\usepackage{lmodern,microtype}

\begingroup
\fontencoding{T1}\fontfamily{lmr}\selectfont
\endgroup

\DeclareFontShape{T1}{lmr}{b}{sc}{<-> ssub * lmr/m/sc}{}
\DeclareFontShape{T1}{lmr}{bx}{sc}{<-> ssub * lmr/m/sc}{}

\DeclareFontShape{T1}{lmr}{m}{scit}{<-> ssub * lmr/m/sc}{}

\DeclareFontShape{T1}{lmr}{b}{scit}{<-> ssub * lmr/m/sc}{}
\DeclareFontShape{T1}{lmr}{bx}{scit}{<-> ssub * lmr/m/sc}{}

\usepackage{mathtools}
\usepackage{thm-restate}
\usepackage[sort,nocompress]{cite}
\usepackage{hyperref}
\hypersetup{
    colorlinks=true,       
    linkcolor=blue!50!black,        
    citecolor=magenta,         
    filecolor=magenta,     
    urlcolor=cyan,         
    linktoc=all
}

\usepackage{float}
\usepackage[linesnumbered,ruled,vlined]{algorithm2e}
\usepackage{enumitem}

\usepackage[flushleft]{threeparttable}

\usepackage{subcaption}
\usepackage{tikz}
\usetikzlibrary{decorations.pathmorphing,positioning,arrows,arrows.meta}
\usepackage[table]{xcolor}
\usepackage{array}
\usepackage{colortbl,arydshln,booktabs,multirow}

\definecolor{greenish}{RGB}{27,158,119}
\definecolor{MyOrange}{RGB}{217,95,2}
\definecolor{MyPurple}{RGB}{117,112,179}

\theoremstyle{plain}
\newtheorem{thm}{Theorem}
\numberwithin{thm}{section}

\newtheorem{lemma}[thm]{Lemma}
\newtheorem{cor}[thm]{Corollary}

\newtheorem{prop}[thm]{Proposition}

\theoremstyle{definition}

\newtheorem{remark}[thm]{Remark}

\makeatletter
\renewcommand\operator@font{\sf}
\makeatother

\DeclareMathOperator{\ei}{\ell_{\infty}}

\newcommand{\NP}{\mathsf{NP}}

\newcommand{\cl}{\operatorname{cl}}
\newcommand{\ind}{\operatorname{index}}

\DeclareMathOperator*{\argmin}{arg\,min}

\newcommand{\bQ}{\mathbb{Q}}
\newcommand{\bR}{\mathbb{R}}
\newcommand{\bZ}{\mathbb{Z}}

\newcommand{\cB}{\mathcal{B}}

\newcommand{\cF}{\mathcal{F}}

\newcommand{\cI}{\mathcal{I}}

\usepackage{comment}

\newcommand{\OPT}{\delta^*}

\newcommand{\problemdef}[3]{%
\vspace{4pt}
    \begin{center}
	\fbox { \parbox[c] {0.9\textwidth} {
    \textsc{\large {#1}}
    
    \smallskip
    
    \textbf{Input:} #2 

    \smallskip
    
    \textbf{Goal:} #3
    }}
    \end{center}
\vspace{4pt}
}

\newlength{\bibitemsep}
\newlength{\bibparskip}
\let\oldthebibliography\thebibliography
\renewcommand\thebibliography[1]{%
  \oldthebibliography{#1}%
  \setlength{\parskip}{\bibitemsep}%
  \setlength{\itemsep}{\bibparskip}%
}

\makeatletter
\renewcommand{\paragraph}{%
  \@startsection{paragraph}{4}%
  {\z@}{1.6ex \@plus 1ex \@minus .2ex}{-0.5em}%
  {\normalfont\normalsize\bfseries}%
}
\makeatother

\def\final{1}  
\def\iflong{\iffalse}
\ifnum\final=1  
\newcommand{\knote}[1]{{\color{blue}[{\tiny \textbf{Krist{\'o}f:} \bf #1}]\marginpar{\color{blue}*}}}
\newcommand{\mnote}[1]{{\color{purple}[{\tiny \textbf{Mirabel:} \bf #1}]\marginpar{\color{purple}*}}}
\newcommand{\jnote}[1]{{\color{cyan}[{\tiny \textbf{José:} \bf #1}]\marginpar{\color{cyan}*}}}
\else 
\newcommand{\knote}[1]{}
\newcommand{\mnote}[1]{}
\newcommand{\jnote}[1]{}
\fi
\usepackage{titling}
\usepackage{orcidlink}
\title{Subset-Constrained Inverse Matroid Optimization}

\author{
Kristóf Bérczi \orcidlink{0000-0003-0457-4573} \thanks{MTA-ELTE Matroid Optimization Research Group and HUN-REN–ELTE Egerváry Research Group, Department of Operations Research, Eötvös Loránd University, and HUN-REN Alfréd Rényi Institute of Mathematics, Budapest, Hungary. Email: \texttt{kristof.berczi@ttk.elte.hu}.}
\and
Mirabel Mendoza-Cadena \orcidlink{0000-0002-4805-0029} \thanks{Instituto de Ciencias de la Ingeniería, Universidad de O'Higgins, Rancagua, Chile. Email: \texttt{mirabel.mendoza@uoh.cl}.}
\and
José A. Soto \orcidlink{0000-0003-2219-8401} \thanks{Department of  Mathematical Engineering and Center for Mathematical Modeling (CNRS IRL2807), Universidad de Chile, Santiago, Chile.  Email: \texttt{jsoto@dim.uchile.cl}.}
}
\date{}

\begin{document}
\pagenumbering{roman}
\maketitle
\thispagestyle{empty}

\begin{abstract} 
In inverse optimization, the goal is to find a minimum perturbation of weights that makes a prescribed feasible solution optimal. For matroids, the classical inverse problem fixes a target basis. We replace this fixed target by a subset constraint: given a matroid $M=(S,\cI)$, weights $w$, and a subset $S_0\subseteq S$, we specify how the family of maximum-weight bases relates to the bases contained in $S_0$. We study six natural variants. The positive variants require, respectively, that at least one basis contained in $S_0$ be optimal, that every basis contained in $S_0$ be optimal, or that the optimal bases be exactly the bases contained in $S_0$; we also study the three corresponding negated requirements. This framework captures partial inverse requirements such as forced or forbidden elements, as well as settings where undesirable optimal bases should be excluded.

We give a complete classification of these subset-constrained inverse matroid problems under the $\ell_\infty$- and $\ell_1$-norms. Under the $\ell_\infty$-norm, all six variants admit polynomial-time combinatorial algorithms (interpreting the variants with strict inequalities in the natural integral-weight setting). The algorithms are based on matroid exchange, uniform perturbations, and the connected-component structure of the restriction $M|S_0$. Under the $\ell_1$-norm, the picture changes sharply: the variant requiring at least one optimal basis contained in $S_0$ is strongly $\NP$-hard even for graphic matroids, whereas the remaining variants considered here admit polynomial-time algorithms. Thus the subset-constrained setting separates the two norms already for matroids, and even for spanning trees.

\bigskip


\end{abstract}
 \newpage
\tableofcontents
\newpage
\pagenumbering{arabic}
\newpage
\section{Introduction}
\label{sec:intro}

In inverse optimization, the objective is to adjust the parameters of an optimization problem so that a prescribed feasible solution becomes optimal. These problems are motivated by applications where one seeks to explain or to enforce a given solution through limited changes to the input, and they have been widely studied across different models; see, for example, the surveys~\cite{richter2016inverse,heuberger2004inverse,demange2014introduction,chan2023inverse} and the recent monograph~\cite{Guan2025InverseBook}.

Let $S$ be a ground set and let $\cF\subseteq 2^S$ be the family of feasible sets. Given a weight function $w\in\bR^S$, let $(S,\cF,w)$ denote the problem of finding a maximum-weight set in $\cF$. Deviations between weight functions are measured using a function $\|\cdot\|$, typically a norm. A classical inverse optimization problem can then be stated as follows.

\problemdef{Inverse Optimization (\textsc{IO})}
    {A finite ground set $S$, a collection $\cF \subseteq 2^S$ of feasible sets, 
    an input set $F^* \in \cF$, a weight function $w \in \bR^S$, an objective function $\| \cdot \|$ defined on $\bR^S$, and an oracle $\mathcal{O}$ for the maximization problem $(S, \cF, w')$ for any weight function $w'\in\bR^S$.}
    {Find a new weight function $w^* \in \bR^S$ such that 
    \begin{enumerate}[label=(\alph*),topsep=1px,partopsep=2px]\itemsep0em
        \item \label{it:IOP_a} $F^*$ is a maximum-weight member of $\cF$ with respect to $w^*$, and
        \item \label{it:IOP_c} $\|w-w^*\|$ is minimized.
    \end{enumerate}}
A weight function satisfying property~\ref{it:IOP_a} is called \emph{feasible}, and it is called \emph{optimal} if it also satisfies~\ref{it:IOP_c}. The \emph{optimal value} of the instance is $\delta^* \coloneqq \|w-w^*\|$. 

In the classical formulation, the target is a single feasible solution. In many situations, however, the desired outcome is specified only implicitly. One may require that at least one feasible solution satisfying a given property become optimal, that all feasible solutions satisfying the property become optimal, or that these be exactly the optimal solutions. Conversely, one may wish to rule out undesirable optimal solutions. Such requirements arise naturally in partial inverse optimization and dynamic pricing, among other applications.  Concrete instances of these conditions, arising in network design and market-pricing applications, are discussed in Section~\ref{sec:motiv}.

In this paper, we study these questions for the \textsc{Inverse Matroid} (\textsc{IM}) problem, where the feasible sets are the bases of a matroid $M=(S,\cI)$. Instead of prescribing a single target basis, we specify a subset $S_0\subseteq S$ and impose conditions on the family of bases contained in $S_0$. This leads to three natural positive variants: requiring that some basis contained in $S_0$ be maximum-weight (\textsc{IM-Exists}), that every basis contained in $S_0$ be maximum-weight (\textsc{IM-All}), or that the maximum-weight bases be exactly the bases contained in $S_0$ (\textsc{IM-Only}). We also study the corresponding negated variants, as well as an auxiliary problem, (\textsc{IM-Outside}).

Although this generalization is simple to state, it changes the structure of the inverse problem substantially. In the classical setting, the target is a single basis, whereas in \textsc{IM-Exists} the perturbation may choose which basis contained in $S_0$ becomes optimal. The stronger variants impose conditions on the entire family of bases contained in $S_0$, requiring many bases to become simultaneously optimal or excluding other bases from being optimal. From the polyhedral point of view, this means that the target is no longer a single normal cone of the matroid base polytope, but rather a union of normal cones together with additional global constraints. As we show, this change has significant algorithmic consequences.

We study these problems under both the $\ell_\infty$- and $\ell_1$-norms. Some of these variants naturally require an integral formulation: when strict separation between basis weights is needed, the real-valued version may have an infimum that is not attained. This setting has received little attention in the inverse optimization literature despite the fact that integrality is known to introduce additional computational challenges.

Our main result is a complete complexity classification of these subset-constrained inverse matroid problems, including the integrality-constrained versions. Under the $\ell_\infty$-norm, all variants admit polynomial-time combinatorial algorithms. Under the $\ell_1$-norm, in contrast, \textsc{IM-Exists} becomes strongly NP-hard even for graphic matroids, while the remaining variants remain polynomial-time solvable. The algorithms rely on several structural properties that appear to be specific to the subset-constrained setting, including reductions to classical inverse matroid optimization, characterizations based on the connected components of restricted matroids, and minimum-cut formulations.

\subsection{Our results and techniques} 
\label{sec:our_results}

We classify the subset-constrained inverse matroid variants according to the norm used to measure the perturbation; see Table~\ref{tab:results}. The outcome is a sharp contrast between the two norms. Under the $\ell_\infty$-norm, all six variants admit polynomial-time combinatorial algorithms. Under the $\ell_1$-norm, the seemingly mild variant \textsc{IM-Exists} becomes strongly $\NP$-hard even for graphic matroids, while the remaining variants considered here remain polynomial-time solvable. It is worth noting that, in each tractable case, our algorithms implicitly lead to min--max results for the optimum value in question.

We first revisit the classical \textsc{IM} problem in Section~\ref{sec:offline}. This serves two purposes. First, it makes the paper self-contained. Second, the algorithm for \textsc{IM-Exists} under the $\ell_\infty$-norm uses the classical problem as a subroutine. For the $\ell_\infty$-norm, we show that the optimal value is determined by a single exchange between the target basis and a maximum-weight basis (Theorem~\ref{thm:im_infty}). This gives a direct combinatorial algorithm and a refinement of the min--max formula from~\cite{berczi2023infinity} in the matroid setting (Corollary~\ref{cor:IM-offline-infty_symmDiff_equals_2}). For the $\ell_1$-norm, we recall the standard monotonicity structure of optimal perturbations (Lemma~\ref{lem:IM-L1-nice-structure}) and give a self-contained reduction to a minimum-weight vertex-cover problem in a bipartite exchange graph (Theorem~\ref{thm:IM_L1}).

For the subset-constrained variants, the $\ell_\infty$-norm leads to particularly clean structural descriptions. In \textsc{IM-Exists}, an optimal perturbation can be chosen by increasing all elements of $S_0$ and decreasing all elements of $S\setminus S_0$ by the same amount (Lemma~\ref{lem:IM-Exists_special_weight}). We show that the problem reduces to a classical \textsc{IM} instance obtained by choosing a maximum-weight basis of the restriction $M|S_0$ (Theorem~\ref{thm:IM-Exists-combAlgo}).

Before specializing to matroids, we also show that the corresponding \textsc{IO-All} problem over an arbitrary feasible family admits a weakly polynomial-time ellipsoid algorithm for any convex deviation objective, assuming an optimization oracle for the underlying problem (Lemma~\ref{lem:IO-All-convexFunc}). For \textsc{IM-All}, the first requirement is that all bases contained in $S_0$ have the same weight. This is a matroidal condition: it is equivalent to making the weight function constant on each connected component of the restriction $M|S_0$ (Lemma~\ref{lem:hm}).

The behavior under the $\ell_1$-norm is different, and the hardness of \textsc{IM-Exists} is the main negative result of the paper. We prove that \textsc{IM-Exists} under the $\ell_1$-norm is strongly $\NP$-hard, even for graphic matroids, by a reduction from the Steiner tree problem (Theorem~\ref{thm:IM-Exists_L1_hardness}). Thus the obstruction is not a consequence of allowing arbitrary matroids: it is already present for spanning trees. This is the only hardness result in Table~\ref{tab:results}, but it has broader consequences. Through the reductions in Section~\ref{sec:implications}, it also yields hardness for standard partial inverse formulations in which one wants some optimal basis to avoid a prescribed forbidden set, and, after passing to the dual matroid, for the corresponding forced-element formulations (Corollary~\ref{cor:IM-Exists-four-hardness}). In contrast, the remaining $\ell_1$ variants considered here are polynomial-time solvable. For \textsc{IM-All} and \textsc{IM-Only}, the connected-component formulation for $M|S_0$ leads to a minimum-weight closure formulation, and hence to a minimum-cut algorithm (Theorem~\ref{thm:algo-IM-All-Only-L1}). For \textsc{IM-Not-Exists}, we prove that, in the nontrivial case, there is an optimal solution that increases the weight of at most one element outside $S_0$ (Theorem~\ref{thm:algorithmINMOTEXISTS}). Finally, \textsc{IM-Not-All} and \textsc{IM-Not-Only} are solved by combining the preceding structural results (Theorem~\ref{thm:NOT-IM-All-algorithm} and Corollary~\ref{cor:NOT-IM-Only-L1}).

For \textsc{IM-Only}, \textsc{IM-Not-Exists}, \textsc{IM-Not-All} and \textsc{IM-Not-Only}, we state the problem with integer input weights and require the new weight function to be integer-valued. These variants require a strict inequality between two basis weights. With real weights, the minimum need not be achieved: a tie can be broken by an arbitrarily small perturbation, but no best perturbation may exist. The integer formulation removes this ambiguity. 
The other variants are stated over real weights; when the input weights are integral, our algorithms also return integral solutions in the cases indicated in Table~\ref{tab:results}. 

It is helpful to view the classification through three recurring mechanisms. The first appears in the classical \textsc{IM} problem and in \textsc{IM-Exists}: in both cases, the goal is to make one suitable basis optimal. This keeps the problem close to standard matroid exchange; under the $\ell_\infty$-norm it leads to uniform shifts and reductions to classical \textsc{IM}, whereas under the $\ell_1$-norm the same flexibility is enough to encode the Steiner tree problem. The second mechanism appears in \textsc{IM-All} and \textsc{IM-Only}, where one has to control all bases contained in $S_0$ simultaneously. This is where the connected-component structure of the restriction $M|S_0$ becomes essential: the relevant step is to homogenize the weights so that all bases of $M|S_0$ have the same weight. The negated variants form the third group. Their algorithms combine the preceding exchange and homogenization ideas with a separation step, implemented through the auxiliary problem \textsc{IM-Outside}, which forces at least one maximum-weight basis to lie outside $S_0$.

\begin{table}[ht!]
    \caption{%
    Overview of the results. The columns labelled ``Integral'' refer to the version in which the input weights are integral and the modified weight function is required to be integral. Shaded cells indicate variants for which this integral formulation is the natural one, because strict separation is required for the optimum to be attained. The symbol $\star$ marks overlapping results in earlier work:~\cite{li2016algorithm,zhang2016partialmatroid} examined partial inverse problems allowing only weight increases or decreases, and ~\cite{ahmadian2018algorithms} solved the problem using an LP.
    }
    \label{tab:results}
    \medskip
    \centering
    \renewcommand{\arraystretch}{0.8}
    \arrayrulewidth 1pt
    \footnotesize{
    \begin{tabular}[t]{|>{\centering\arraybackslash}m{7.5em}|>{\centering\arraybackslash}m{9.5em}|>{\centering\arraybackslash}m{9.5em}|>{\centering\arraybackslash}m{8.5em}|>{\centering\arraybackslash}m{7.5em}|}
        \hline
        \textbf{Problem} 
            & \multicolumn{2}{c|}{\textbf{Result under the $\ell_\infty$-norm}} 
            & \multicolumn{2}{c|}{\textbf{Result under the $\ell_1$-norm}}\\ \hline
        & \textbf{General} & \textbf{Integral} & \textbf{General} & \textbf{Integral}\\ \hline
        \arrayrulecolor{gray!30}
        \textsc{IM} 
            & $\mathsf{P}$ (Thm.~\ref{thm:im_infty}) 
            & $\mathsf{P}$ (Rem.~\ref{rem:integral-IM-infty}) 
            & $\mathsf{P}$ (Thm.~\ref{thm:IM_L1}) 
            & $\mathsf{P}$ (Rem.~\ref{rem:integral-IM-L1})\\ \hdashline[1pt/1pt]
        \textsc{IM-Exists} 
            & $\mathsf{P}$ (Thm.~\ref{thm:IM-Exists-combAlgo}, $\star$\cite{li2016algorithm,zhang2016partialmatroid}) 
            & $\mathsf{P}$ (Rem.~\ref{rem:integral-IM-Exists-infty}) 
            & $\NP$-hard (Thm.~\ref{thm:IM-Exists_L1_hardness}) 
            & ---\\ \hdashline[1pt/1pt]
        \textsc{IM-All} 
            & $\mathsf{P}$ (Thm.~\ref{thm:IM-All_algorithm}) 
            & $\mathsf{P}$ (Rem.~\ref{rem:IM-All_integral}) 
            & $\mathsf{P}$ (Thm.~\ref{thm:algo-IM-All-Only-L1}) 
            & $\mathsf{P}$ (Thm.~\ref{thm:algo-IM-All-Only-L1})\\ \hdashline[1pt/1pt]
        \textsc{IM-Only} 
            & \cellcolor{greenish!20} ---
            & \cellcolor{greenish!20} $\mathsf{P}$ (Thm.~\ref{thm:IM-Only_algorithm}, $\star$\cite{ahmadian2018algorithms}) 
            & \cellcolor{greenish!20} ---
            & \cellcolor{greenish!20} $\mathsf{P}$  (Thm.~\ref{thm:algo-IM-All-Only-L1}) \\ \hdashline[1pt/1pt]
        \textsc{IM-Not-Exists} 
            & \cellcolor{greenish!20} ---
            & \cellcolor{greenish!20} $\mathsf{P}$ (Thm.~\ref{thm:NOT-IM-Exists-algorithm}) 
            & \cellcolor{greenish!20} ---
            & \cellcolor{greenish!20} $\mathsf{P}$ (Thm.~\ref{thm:algorithmINMOTEXISTS}) \\ \hdashline[1pt/1pt]
        \textsc{IM-Not-All} 
            & \cellcolor{greenish!20} ---
            & \cellcolor{greenish!20} $\mathsf{P}$ (Thm.~\ref{thm:NOT-IM-All-algorithm}) 
            & \cellcolor{greenish!20} ---
            & \cellcolor{greenish!20} $\mathsf{P}$ (Thm.~\ref{thm:NOT-IM-All-algorithm}) \\ \hdashline[1pt/1pt]
        \textsc{IM-Not-Only} 
            & \cellcolor{greenish!20} --- 
            & \cellcolor{greenish!20} $\mathsf{P}$ (Cor.~\ref{cor:NOT-IM-Only-infty}) 
            & \cellcolor{greenish!20} --- 
            & \cellcolor{greenish!20} $\mathsf{P}$ (Cor.~\ref{cor:NOT-IM-Only-L1})  \\ \hdashline[1pt/1pt] 
        \textsc{IM-Outside} 
            & $\mathsf{P}$ (Thm.~\ref{thm:algo-IM-Outside-infty}) 
            & $\mathsf{P}$ (Cor.~\ref{cor:rine}) 
            & $\mathsf{P}$ (Thm.~\ref{thm:algo-IM-Outside-L1}) 
            & $\mathsf{P}$ (Thm.~\ref{thm:algo-IM-Outside-L1})\\
        \arrayrulecolor{black}\hline
    \end{tabular}
    }
\end{table}

\subsection{Further implications}
\label{sec:implications}

The six variants above are formulated with respect to the family
\[
    \{B\in\cB \colon B\subseteq S_0\}
\]
of bases contained in a prescribed set $S_0$. 
Our algorithms extend verbatim
to the following three families, which correspond, respectively, to
\emph{forcing} the elements of $S_0$ into the basis, \emph{forbidding} the
elements of $S_0$ from the basis, and forcing the elements outside $S_0$ into
the basis:
\[
    \{B\in\cB \colon S_0\subseteq B\},\qquad
    \{B\in\cB \colon B\cap S_0=\emptyset\},\qquad
    \{B\in\cB \colon S\setminus B\subseteq S_0\}.
\]
We explain the reductions. First, the condition
$S_0\subseteq B$ can be handled by matroid duality. If $B$ is a basis of $M$,
then $S\setminus B$ is a basis of $M^*$. Moreover, $B$ is maximum-weight with
respect to $w$ if and only if $S\setminus B$ is maximum-weight in $M^*$ with
respect to $-w$. Hence
\[
    S_0\subseteq B
    \quad\Longleftrightarrow\quad
    S\setminus B\subseteq S\setminus S_0 ,
\]
so this case reduces to the subset-containment variants for $(M^*,-w)$ with
target set $S\setminus S_0$. Second, the condition $B\cap S_0=\emptyset$ is equivalent to
\[
    B\subseteq S\setminus S_0 .
\]
Thus this case is obtained directly from the original subset-containment
variants by replacing $S_0$ with $S\setminus S_0$. Third, the condition $S\setminus B\subseteq S_0$ is already a subset-containment condition for the dual basis $S\setminus B$. Therefore this case reduces to the subset-containment variants for $(M^*,-w)$ with target set
$S_0$. Consequently, after possibly replacing $S_0$ by $S\setminus S_0$, and
replacing $(M,w)$ by $(M^*,-w)$ when complements of bases are used, these
variants reduce to the subset-containment variants studied in the paper.

For \textsc{IM-Exists}, this gives the standard partial inverse
interpretations. If $P\subseteq S$ is a prescribed set of forced elements and
the goal is to make some maximum-weight basis contain $P$, then we apply
\textsc{IM-Exists} to $(M^*,-w)$ with target set $S\setminus P$. Indeed, a
basis $B$ of $M$ contains $P$ if and only if the dual basis $S\setminus B$ is
contained in $S\setminus P$. Similarly, if $F\subseteq S$ is a prescribed set
of forbidden elements and the goal is to make some maximum-weight basis avoid
$F$, then we apply \textsc{IM-Exists} directly to $M$ with target set
$S\setminus F$.

These reductions transfer both polynomial-time solvability and hardness
results. For algorithmic results, the stated running times may change if the
independence oracle of $M^*$ is simulated from the independence oracle of
$M$. We return to these partial inverse interpretations in
Section~\ref{sec:motiv}.

\subsection{Motivation and related work}
\label{sec:motiv}

Inverse combinatorial optimization is a well-established field with a vast
literature; see the comprehensive surveys by Heuberger~\cite{heuberger2004inverse}, Demange and Monnot~\cite{demange2014introduction}, and Chan et al.~\cite{chan2023inverse}, and the recent book on inverse combinatorial optimization by Guan et al.~\cite{Guan2025InverseBook}. Inverse matroid problems under subset
constraints are closely related to several areas in combinatorial optimization.
While previous techniques are well suited for fixed-target inverse problems
and for some partial inverse problems, they do not readily extend to the
subset-constrained variants introduced in this paper.

\paragraph{Motivation and examples.}
Applications of \textsc{IM-Exists}, \textsc{IM-All}, and \textsc{IM-Only}
arise naturally in inverse problems with forced or forbidden elements. In
network design problems modeled by the graphic matroid, for instance, one may
wish to perturb edge weights as little as possible so that a preferred
spanning tree, or a collection of spanning trees, becomes optimal. Such
settings capture leader-follower scenarios in which a decision-maker selects
an optimal solution while another agent seeks to influence this choice through
limited modifications of the input.

Similarly, \textsc{IM-Not-Exists} and \textsc{IM-Not-All} model situations
where undesirable optimal solutions should be ruled out. For example, in
procurement or hiring, one may want to prevent a single supplier or group from
capturing all selected slots. Such problems can be modeled by matroid
constraints, for instance by laminar matroids imposing upper bounds on nested
categories. In such settings, one may seek to modify item costs so that no
optimal solution is contained in a prescribed undesired subset.

A further motivation stems from dynamic pricing schemes. A combinatorial
market consists of a set of indivisible items and a set of agents, where each
agent has a valuation function assigning a value to each subset of items.
Cohen-Addad, Eden, Feldman, and Fiat~\cite{cohen2016invisible} and
independently Hsu, Morgenstern, Rogers, Roth, and
Vohra~\cite{hsu2016prices} observed that Walrasian prices alone may not be
enough to achieve optimal social welfare based solely on buyers' decisions. To
address this,~\cite{cohen2016invisible} introduced dynamic pricing schemes,
where prices can be updated between buyer arrivals. A central open question
posed in~\cite{cohen2016invisible} asked whether every market with gross
substitutes valuations admits a dynamic pricing scheme achieving optimal
social welfare.

Since matroid rank functions form a subclass of gross substitutes valuations,
a natural first step is to investigate whether a pricing scheme exists for
markets with matroid rank valuations~\cite{berczi2021market}. In this
setting, the problem reduces to an inverse-type formulation: given $k$
matroids $M_1,\dots,M_k$ over the same ground set, find a weight function
$w\in\mathbb{R}^S$ such that, for any maximum $w$-weight basis $B$ of $M_i$,
the complement $S\setminus B$ is a basis of the sum of the remaining
matroids. This resembles the structural requirements of \textsc{IM-Only},
where one prescribes a family of bases that should be exactly the family of
maximum-weight bases.

\paragraph{Related work: inverse matroid optimization.}
The classical \textsc{IM} problem was first introduced by
Cai~\cite{mao1999inverse}, who studied matroid intersection under the
weighted $\ell_1$-norm objective. Cai derived an LP formulation and
transformed the problem into a circulation model. For \textsc{IM}, this model
implies an LP with $\frac{|S|^2}{2}+2|S|$ constraints and $2|S|$ variables.
Dell'Amico, Maffioli, and Malucelli~\cite{dell2003base} later considered
\textsc{IM} under the $\ell_1$-norm and proposed an auxiliary matroid
construction that runs in $O(|S|r+r^3+|S|\varphi)$ time, where $\varphi$ is
the time to find the unique circuit in $M$ formed by adding to the given basis
$B$ an element not in $B$.

For the $\ell_\infty$-norm, classical \textsc{IM} is covered by several more
general settings. For example, both Li, Zhang, and
Lai~\cite{li2016algorithm} and Zhang, Li, Lai, and
Du~\cite{zhang2016partialmatroid} solve partial inverse problems for matroids
and their algorithms imply solutions to \textsc{IM} in one-sided settings. In
inverse optimization for $0$-$1$ linear programs, the work of Ahuja and
Orlin~\cite{Ahuja2001Inverse} implies an LP description for \textsc{IM} via
the ellipsoid method, while the results of Zhang and
Liu~\cite{zhang1999further,zhang2002general} imply an LP with one variable
and constraints corresponding to fundamental circuits with respect to the
target basis. Over general set systems, Bérczi, Mendoza-Cadena, and
Varga~\cite{berczi2023infinity,berczi2024spanLATIN} studied inverse problems
under the $\ell_\infty$-norm and span objectives, implying an algorithm that
makes $O(|S|)$ calls to an oracle for finding a maximum-weight feasible set.
In the matroid setting, our direct algorithm makes $O(r\log r)$ calls to the
independence oracle once a maximum-weight basis is known. Variants of the problem have also been explored under other distance
measures, including weighted Hamming distance~\cite{aman2016inverse} and the
$\ell_\infty$-norm~\cite{tayybi2021inverse}.

Frank and Murota~\cite{frank2022discrete} developed a general min--max
framework for minimizing an integer-valued separable discrete convex function
over the integral points of a box-TDI polyhedron, and showed that this
framework covers a broad class of inverse combinatorial optimization
problems. In particular, it captures fixed-solution inverse problems with
separable deviation objectives such as the $\ell_1$ deviation. However, this
framework does not subsume the results of the present paper. First, the
$\ell_\infty$ objective is not separable since it couples all coordinates through a maximum. Second, our subset-constrained
variants are not fixed-target inverse problems. For example,
\textsc{IM-Exists} asks that some basis contained in $S_0$ be optimal, so the
feasible set of admissible weight functions is a union of normal cones, one
for each such basis, rather than the normal cone of a prescribed solution.
The \textsc{IM-Only} and negated variants impose further exclusion or
strict-separation requirements. Finally, the results
of~\cite{frank2022discrete} are primarily min--max characterizations and do
not, by themselves, provide algorithms for computing an optimal modified
weight function.

\paragraph{Related work: forced and forbidden elements.}
Partial inverse optimization problems aim to adjust weights so that a desired
subset of elements is either included in or excluded from an optimal solution.
In the matroid setting, Li, Zhang, and Lai~\cite{li2016algorithm} and Zhang,
Li, Lai, and Du~\cite{zhang2016partialmatroid} considered the problem of
modifying weights, only decreasing or only increasing them, respectively, so
that a maximum-weight basis contains a given independent set. The partial
inverse problem where only a set of elements must be part of a basis was
studied by Lai and Orlin~\cite{lai2003complexity} under the weighted
$\ell_\infty$-norm.

Similar problems have been studied for spanning trees, cuts, and matchings~\cite{dong2022partial,li2019capacited,li2020approximation,li2022partial,li2021capacited}, with different complexities depending on the structure of the constraints and the norm used as an objective~\cite{gassner2010,lai2003complexity,gassner2009}. In the exclusion variant, where certain elements must be avoided in any optimal solution, many well-known problems become $\NP$-hard even under simple norms~\cite{lai2003complexity,gassner2009}. Ahmadian, Bhaskar, Sanit{\`a}, and Swamy~\cite{ahmadian2018algorithms} considered integral inverse optimization problems in which a prescribed collection of feasible solutions is required to be exactly the set of optimal solutions. For matroid bases, their results include an LP-based treatment of the corresponding inverse min-cost and max-cost basis problems, as well as extensions to minimum-deviation variants such as $\ell_\infty$ distance minimization. In our terminology, this coincides with the integral $\ell_\infty$ version of \textsc{IM-Only}. Our results recover the overlapping case by a direct combinatorial algorithm based on the connected components of $M|S_0$; more broadly, we treat all six subset-constrained variants and compare the $\ell_\infty$- and $\ell_1$-norms.

As detailed in Section~\ref{sec:implications}, the unconstrained versions of
these standard partial inverse problems map directly to instances of our
subset-constrained variants through dual matroid constructions.

\section{Preliminaries}
\label{sec:preliminaries}
Here we provide a brief summary of basic definitions and notation. For further details on matroid theory, we refer the reader to~\cite{oxley2011matroid,frank2011book}.

We denote the sets of \emph{real} and \emph{integer} numbers by $\bR$ and $\bZ$, respectively, and add $+$ as a subscript whenever non-negativity is assumed. For a positive integer $k$, we use $[k]\coloneqq \{1,\dots,k\}$. Let $S$ be a ground set of size $n$. Given subsets $X,Y\subseteq S$, the \emph{symmetric difference} of $X$ and $Y$ is denoted by $X\triangle Y\coloneqq (X\setminus Y)\cup(Y\setminus X)$. If $Y$ consists of a single element $y$, then $X\setminus\{y\}$ and $X\cup \{y\}$ are abbreviated as $X-y$ and $X+y$, respectively. Given a weight function $w\in \bR^S$, the weight of a set $F \in \cF$ is $w(F) \coloneqq \sum_{s \in F} w(s)$. By convention, we define $\min \emptyset = +\infty$ and $\max \emptyset = -\infty$. Given a vector $p\in\bR^S$, its \emph{$\ell_1$-norm} is $\|p\|_1\coloneqq \sum_{s\in S}|p(s)|$, and its \emph{$\ell_\infty$-norm} is $\|p\|_\infty\coloneqq\max_{s\in S}|p(s)|$. The \emph{characteristic vector} of a set $Z\subseteq S$ is $\chi_Z(s)=1$ if $s\in Z$ and $0$ otherwise.

\medskip

A \emph{matroid} $M=(S,\cI)$ consists of a finite \emph{ground set} $S$ and a family $\cI\subseteq 2^S$ of \emph{independent sets} satisfying the independence axioms: (I1) $\emptyset\in\cI$, (I2) $X\subseteq Y,\ Y\in\cI\Rightarrow X\in\cI$, and (I3) if $X,Y\in\cI$ and $|X|<|Y|$, then there exists $e\in Y\setminus X$ such that $X+e\in\cI$. We set $n\coloneqq |S|$. The \emph{rank} $r(X)$ of a set $X\subseteq S$ is the maximum size of an independent set contained in $X$, and the rank of the matroid is $r(S)$, denoted simply by $r$ when the matroid is clear from the context. The maximal independent sets are called \emph{bases}, and their family is denoted by $\cB$. If a matroid is specified by its rank function or by its family of bases, then we write $M=(S,r)$ or $M=(S,\cB)$, respectively.

A set $X\subseteq S$ is said to \emph{span} or \emph{generate} a set $Y\subseteq S$ if $r(X\cup Y)=r(X)$. The \emph{closure} of $X$ is defined by $\cl(X)\coloneqq \{s\in S\colon r(X+s)=r(X)\}$. Equivalently, $\cl(X)$ is the set of elements spanned by $X$.

We use the following standard matroid operations. For $S'\subseteq S$, the \emph{restriction} of $M=(S,\cI)$ to $S'$ is the matroid $M|S'=(S',\cI')$, where $\cI'\coloneqq \{I\in\cI\colon I\subseteq S'\}$. Given a matroid $M=(S,\cB)$, its \emph{dual} is the matroid $M^*=(S,\cB^*)$, where $\cB^*\coloneqq \{X\subseteq S\colon S\setminus X\in\cB\}$. Thus, the bases of $M^*$ are precisely the complements of the bases of $M$. If $M_1=(S_1,\cI_1),\dots,M_k=(S_k,\cI_k)$ are matroids on pairwise disjoint ground sets and $S=S_1\cup\dots\cup S_k$, then their \emph{direct sum} is the matroid $M_1\oplus\dots\oplus M_k=(S,\cI)$, where $\cI\coloneqq \{I\subseteq S\colon I\cap S_i\in\cI_i\text{ for every }i\in\left[k\right]\}$.

We will also use graphic matroids. Given a graph $G=(V,E)$, the \emph{graphic matroid} $M(G)$ has ground set $E$, and a set of edges is independent if it is a forest. Thus, the bases of $M(G)$ are the spanning forests of $G$, and if $G$ is connected, they are the spanning trees of $G$.

An inclusionwise minimal dependent set is called a \emph{circuit}. If $X$ is independent and $X+e$ is dependent, then the unique circuit contained in $X+e$ is called the \emph{fundamental circuit} of $e$ with respect to $X$, and it is denoted by $C(X,e)$.
\begin{remark}\label{rem:findingFundCircuit}
Given a basis $B$ and an element $f \notin B$, we can find $C(B,f)$ using $O(r)$ independence-oracle calls. To do this, observe that for each $e \in B$, we query the independence oracle to test whether $B - e + f$ is independent or dependent. For each $e\in B$, we have
$e\in C(B,f)\setminus\{f\}$ if and only if $B-e+f$ is independent.
Thus, $C(B,f)$ can be found using $r=|B|$ independence-oracle calls.
\end{remark}

The \emph{fundamental circuit exchange graph} of a basis $B$ is the bipartite graph $G_B=(B,S\setminus B;E_B)$ in which $ef\in E_B$, for $e\in B$ and $f\in S\setminus B$, if and only if $e\in C(B,f)$, or equivalently, $B-e+f$ is a basis.

The connected components of a matroid are defined through circuits. More precisely, define a relation on $S$ by declaring $x\sim y$ if $x=y$ or if there is a circuit containing both $x$ and $y$. By the circuit elimination property, this is an equivalence relation; equivalently, its equivalence classes are the connected components of the circuit hypergraph of $M$. Let these classes be $S_1,\dots,S_k$. The matroids $M_i=M|S_i$ are called the \emph{connected components} of $M$, and $M=M_1\oplus\dots\oplus M_k$. A matroid is \emph{connected} if it has a single connected component. The connected components can also be obtained from the fundamental circuits with respect to any fixed basis $B$: the connected components of the hypergraph with vertex set $S$ and hyperedge set $\{C(B,f)\colon f\in S\setminus B\}$ are exactly the connected components of $M$. In particular, the connected components of a matroid can be computed using $O(nr)$ independence-oracle calls by the results of Krogdahl~\cite{krogdahl1977dependence}. In the special case of a graphic matroid $M(G)$, the connected components are the loops, the bridges, and the edge sets of the maximal $2$-vertex-connected blocks of $G$.

We use the following form of the basis exchange property; see, e.g.,~\cite[Theorem 5.3.4]{frank2011book}.

\begin{prop}\label{prop:exchange}
For any two bases $B_1$ and $B_2$ of a matroid, there exists a bijection $\varphi\colon B_1\setminus B_2\to B_2\setminus B_1$ such that $B_1-e+\varphi(e)$ is a basis for every $e\in B_1\setminus B_2$.
\end{prop}

We also use the symmetric basis-exchange property; see, e.g.,~\cite[Theorem 5.3.3]{frank2011book}.

\begin{prop}[Symmetric basis exchange]\label{prop:symmetric-exchange}
Let $B_1$ and $B_2$ be bases and let $f\in B_2\setminus B_1$. Then there exists $e\in B_1\setminus B_2$ such that both $B_1-e+f$ and $B_2-f+e$ are bases.
\end{prop}

The following characterization of maximum-weight bases will be used repeatedly; see, e.g.,~\cite[Theorem 5.5.3]{frank2011book}.

\begin{prop}\label{prop:charact_optimum_basis}
A basis $B$ is maximum-weight with respect to a weight function $w\in\bR^S$ if and only if $w(f)\leq w(e)$ for every $f\in S\setminus B$ and every $e\in C(B,f)\setminus\{f\}$.
\end{prop}

We will often use an equivalent closure formulation. For a set $S'\subseteq S$ and an element $x\in S$, define
\[
    S'\!\left[w\geq w(x)\right]\coloneqq \{s\in S'\colon w(s)\geq w(x)\}.
\]
If $B$ is a basis and $f\in S\setminus B$, then $C(B,f)-f\subseteq X$ for $X\subseteq B$ if and only if $f\in\cl(X)$. Applying this to $X=B\!\left[w\geq w(f)\right]$ gives the following corollary.

\begin{cor}[Closure condition]\label{cor:charact_optimum_basis_Closure}
A basis $B$ is maximum-weight with respect to $w$ if and only if
\[
    f\in\cl\bigl(B\!\left[w\geq w(f)\right]\bigr)
\]
for every $f\in S\setminus B$.
\end{cor}

As usual in matroid algorithms, we assume that the matroid is given by an independence oracle, and we measure the running time by the number of oracle calls and elementary operations. A polynomial number of oracle calls means a number polynomial in $n=|S|$.

\section{\textsc{Inverse Matroid}}
\label{sec:offline}

As a warm-up, we consider \textsc{IM}, where the underlying optimization problem is to find a maximum-weight basis of a matroid. We include this case mainly for completeness and to keep the presentation self-contained. Moreover, the algorithm for \textsc{IM-Exists} under the $\ell_\infty$-norm relies on this problem as a subroutine, and the algorithms for \textsc{IM-All} in Section~\ref{sec:IM-All} also solve this problem by taking $S_0 = B_0$ for both the $\ell_\infty$- and $\ell_1$-norms. The results in this section should be viewed as a simple starting point rather than as the main technical contribution: their role is to isolate, in the classical setting, the exchange and perturbation ideas that will be used later. While, to the best of our knowledge, \textsc{IM} has not been studied explicitly under the $\ell_\infty$-norm, closely related versions have already been addressed in the literature, including for set systems, more general matroid settings, and spanning trees. Readers already familiar with classical inverse matroid optimization may read this section mainly as preparation for the subset-constrained variants. Section~\ref{sec:IM_infty} presents our algorithm for the $\ell_\infty$-norm, exploiting the matroid structure to refine the min--max theorem from~\cite{berczi2023infinity} and derive additional structural consequences. Section~\ref{sec:IM-L1} treats the $\ell_1$-norm.

Formally, \textsc{IM} is defined as follows.

\problemdef{Inverse Matroid (\textsc{IM})}
    {A matroid $M=(S,\mathcal{I})$, a fixed basis $B^*$, a weight function $w\in\bR^S$, and an objective function $\|\cdot\|$ defined on $\bR^S$.}
    {Find a weight function $w^* \in \bR^S$ such that 
    \begin{enumerate}[label=(\alph*),topsep=1px,partopsep=2px]\itemsep0em
        \item \label{it:IM_a} $B^*$ is a maximum-weight basis of $M$ with respect to $w^*$, and
        \item \label{it:IM_c} $\|w-w^*\|$ is minimized.
    \end{enumerate}}

A weight function is called \textit{feasible} if condition~\ref{it:IM_a} holds. Note that \textsc{IM} is always feasible: set $w^*(s)=0$ for all $s\in S$.

\begin{remark}
    Given a weight function $w'$, we can verify in polynomial time whether $w'$ is feasible for \textsc{IM}. To do so, compute a maximum $w'$-weight basis of $M$ and compare its weight to that of $B^*$.
\end{remark}

\subsection{IM under \texorpdfstring{$\ell_\infty$}{l-infinity}-norm}
\label{sec:IM_infty}
The goal of this section is to design an algorithm for \textsc{IM} under the $\ell_\infty$-norm.
As mentioned in Section~\ref{sec:motiv}, several works address this problem in more general settings~\cite{berczi2023infinity,Ahuja2001Inverse,zhang2002general,zhang1999further,li2016algorithm,zhang2016partialmatroid}, yet none of them targets matroids specifically.
We exploit the structure of matroids and prove a refined version of~\cite{berczi2023infinity} for \textsc{IM}.
 
A common technique for \textsc{IO} under the $\ei$-norm is to show that an optimal weight function $w^*$ of a special form exists; see, e.g.,~\cite{mao1999inverse,berczi2023infinity}. We use this result here.

\begin{prop}\label{prop:ell_infty_devVector_of_special_form}     
 Let $\delta^*$ be the optimum value of an \textsc{IO} instance $( S, \cF, F^*, w, \|\cdot\|_\infty )$. Then, for any $\delta\geq \delta^*$, 
    $w_\delta= w +\delta\cdot \chi_{F^*} - \delta\cdot \chi_{S\setminus F^*}$
    is a feasible weight function. In particular, $w_{\delta^*}$ is an optimal solution.
\end{prop}

We proceed to our main theorem, which has an algorithmic proof. An illustration of the algorithm is shown in Figure~\ref{fig:example_IM_infty}.

\begin{thm}\label{thm:im_infty}
Let $B_{\max}$ be a maximum $w$-weight basis of $M$, and let $B^*$ be a fixed basis. If $B^*$ is already a maximum $w$-weight basis, then the optimum value of the \textsc{IM} instance $(M=(S,\cI),B^*,w,\|\cdot\|_\infty)$ is $0$. Otherwise, there exist elements $f\in B_{\max}\setminus B^*$ and $e\in B^*\setminus B_{\max}$ such that $B^*-e+f$ is a basis, $w(f)>w(e)$, and the optimum value is $\delta^*=(w(f)-w(e))/2$. Moreover,
$\delta^*$ can be found using $O(n\log n)$ elementary operations and $O(n + r\log r)$ independence-oracle calls.
\end{thm}
\begin{proof}
If $B^*$ is already a maximum $w$-weight basis, then $w$ is feasible and the
optimum value is $0$. We assume from now on that $B^*$ is not maximum-weight.
Label $B^*=\{e_1,\dots,e_r\}$ so that $w(e_1)\ge\dots\ge w(e_r)$.  For each $f\in B_{\max}\setminus B^*$, we define $\ind(f)=\min\{i\colon f\in\cl(\{e_1,\dots,e_i\})\}$ and $\gamma_f=\max\{w(f)-w(e_{\ind(f)}),0\}$. Finally, let $\delta^*=\max\{\gamma_f/2\colon f\in B_{\max}\setminus B^*\}$.

We claim that $\delta^*$ is the optimum value of the \textsc{IM} instance considered. To prove this, define a modified weight function $w^*$ as in Proposition~\ref{prop:ell_infty_devVector_of_special_form}, namely
\[
w^*(s)=
\begin{cases}
w(s)+\delta^* & \text{if $s\in B^*$},\\
w(s)-\delta^* & \text{if $s\notin B^*$}.
\end{cases}
\]
We define a total order $s_1, \dots, s_n$ on the ground set $S$, sorted in decreasing order of $w^*$-value. To break ties, elements in $B^*$ come before those in $B_{\max} \setminus B^*$, which in turn come before elements in $S \setminus (B^* \cup B_{\max})$. That is, if $w^*(s_i) = w^*(s_j)$, then $s_i$ precedes $s_j$ whenever $s_i \in B^*$ and $s_j \notin B^*$, or when $s_i \in B_{\max}$ and $s_j \notin B^* \cup B_{\max}$. Ties not resolved by these rules are broken using any fixed total order on $S$.

We claim that the greedy algorithm applied to $S$ under this order returns $B_G = B^*$. To see this, note that since $B_{\max}$ is a maximum $w$-weight basis, any element $g \in S \setminus (B_{\max} \cup B^*)$ satisfies
\[
g \in \cl\bigl(B_{\max}[w \geq w(g)]\bigr) \subseteq \operatorname{cl}\bigl(\{s_j \in S \colon j < i\}\bigr),
\]
where $s_i = g$ is its position in the ordering. Hence, $g\notin B_G$. Furthermore, for every $f\in B_{\max}\setminus B^*$, 
$f\in \cl\bigl(\{e_1,\dots,e_{\ind(f)}\}\bigr) \subseteq \cl\bigl(B^*[w \geq w(f)-\gamma_f]\bigr) \subseteq \cl\bigl(\{s_j \in S \colon j < i\}\bigr)$,
 where $s_i = f$ is its position in the ordering. Hence, $f \notin B_G$. Therefore, no element outside $B^*$ is selected by the greedy algorithm, which implies that $B_G \subseteq B^*$. Since $B_G$ and $B^*$ are both bases, it follows that $B_G = B^*$.

It remains to show that no smaller value is feasible. Let $\delta<\delta^*$. Then there exists $f\in B_{\max}\setminus B^*$ with $2\delta<\gamma_f$. Let $e=e_{\ind(f)}$. By the minimality in the definition of $\ind(f)$, we have $e\in C(B^*,f)$, and hence $B^*-e+f$ is a basis. Let $w'$ be any weight function with $\|w'-w\|_\infty\le\delta$. Then $w'(f)\geq w(f)-\delta$ and $w'(e)\leq w(e)+\delta$, thus 
\[
w'(f)-w'(e)\geq w(f)-w(e)-2\delta>0.
\]
Hence $B^*-e+f$ has larger weight than $B^*$ with respect to $w'$, and so $w'$ is not feasible.

This gives an algorithm. If $B^*$ is already a maximum $w$-weight basis, then the optimum value is
$0$. Otherwise, sort the elements of $B^*$ in nonincreasing order of
their $w$-weights. For each $f\in B_{\max}\setminus B^*$, compute
$\ind(f)$, determine the corresponding value $\gamma_f$, and finally
compute $\delta^*$ and $w^*$ as stated above.

\textit{Running time.}
Obtaining $B_{\max}$ requires $O(|S| \log |S|)$ time for sorting the elements in $S$, and $O(|S|)$ independence-oracle calls for running the Greedy algorithm. 
Sorting the elements of $B^*$ requires
$O(r\log r)$ elementary operations. For each
$f\in B_{\max}\setminus B^*$, the value $\ind(f)$ can be computed using
binary search over the ordered basis $B^*$, requiring
$O(\log r)$ independence-oracle calls and elementary operations.
Since $|B_{\max}\setminus B^*|\le r$, computing all values $\ind(f)$,
and hence all values $\gamma_f$, requires
$O(r\log r)$ elementary operations and
$O(r\log r)$ independence-oracle calls. The remaining computations take
linear time. 
\end{proof}

\begin{figure}[ht!]
\centering
    \begin{subfigure}[b]{0.45\textwidth}
        \centering
        \begin{tikzpicture}[scale=0.9,
            myEdge/.style={line width = 1.5pt},     
            state/.style={circle,  minimum size=1em, draw, line width = 1.2pt}]   
            \node (a) at (0,2) [state] {$a$};
            \node (b) at (4,2) [state] {$b$};
            \node (c) at (1,1) [state] {$c$};
            \node (d) at (3,1) [state] {$d$};
            \node (e) at (0,0) [state] {$e$};
            \node (f) at (4,0) [state] {$f$}; 
            \draw [myEdge] (a) to node[midway, above] {7} (b) ; 
            \draw [myEdge] (e) to node[midway, below] {6} (f) ;
            \draw [myEdge] (a) to node[midway, left] {0} (e);
            \draw [myEdge] (d) to node[midway, left] {3} (f);
            \foreach \u \v \w \pos in {%
            a/c/{\rlap{0}\phantom{3.5..}}/right, c/d/6/above, d/b/6/left,  e/c/1/right, b/f/{\rlap{1}\phantom{3.5..}}/right%
            }{%
                \draw [MyOrange!30, line width = 5pt] (\u) to (\v) ; 
                \draw [myEdge] (\u) to node[midway,\pos] {\w} (\v) ; 
            }     
        \end{tikzpicture}
        \caption{Original instance.}
        \label{fig:example_IM_infty_original}
    \end{subfigure}%
    \begin{subfigure}[b]{0.45\textwidth}
        \centering
        \begin{tikzpicture}[scale=0.9,
            myEdge/.style={line width = 1.5pt},     
            state/.style={circle,  minimum size=1em, draw, line width = 1.2pt}]   
            \node (a) at (0,2) [state] {$a$};
            \node (b) at (4,2) [state] {$b$};
            \node (c) at (1,1) [state] {$c$};
            \node (d) at (3,1) [state] {$d$};
            \node (e) at (0,0) [state] {$e$};
            \node (f) at (4,0) [state] {$f$}; 
            \draw [myEdge] (a) to node[midway, above] {3.5} (b) ; 
            \draw [myEdge] (e) to node[midway, below] {2.5} (f) ;
            \draw [myEdge] (a) to node[midway, left] {-3.5} (e);
            \draw [myEdge] (d) to node[midway, left] {0.5} (f);
            \foreach \u \v \w \pos in {%
            a/c/{\rlap{3.5}\phantom{3.5..}}/right, c/d/9.5/above, d/b/9.5/left,  e/c/4.5/right, b/f/{\rlap{4.5}\phantom{3.5..}}/right%
            }{%
                \draw [MyOrange!30, line width = 5pt] (\u) to (\v) ; 
                \draw [myEdge] (\u) to node[midway,\pos] {\w} (\v) ; 
            }     
        \end{tikzpicture}
        \caption{Optimal $w^*$ under the $\ell_\infty$-norm.}
        \label{fig:example_IM_infty_final}
    \end{subfigure}%
    \caption{Illustration of algorithms for \textsc{IM} on the graphic matroid in Figure~\ref{fig:example_IM_infty_original} under the $\ell_\infty$-norm (the algorithm from Theorem~\ref{thm:im_infty}). The basis $B^* = \{bd,cd,bf,ce,ac\}$, ordered by nonincreasing weight, is highlighted in orange; $B_{\max}=\{ab,bd,cd,ef,df\}$. 
    We have that $\ind(ab)=5$ and $\gamma_{ab}=\max\{w(ab)-w(ac),0\}=7-0=7$, $\ind(ef)=3$ and $\gamma_{ef}=\max\{w(ef)-w(bf),0\}=6-1=5$, $\ind(df)=3$ and $\gamma_{df}=\max\{w(df)-w(bf),0\}=3-1=2$, and thus $\delta^*=\max\{7/2,5/2,2/2\}=3.5$. The optimal weight function is shown in Figure~\ref{fig:example_IM_infty_final}.
    } 
    \label{fig:example_IM_infty}
\end{figure}

The proof of Theorem~\ref{thm:im_infty} also yields the following min--max characterization: the maximum of
\[
    \frac{w(B)-w(B^*)}{|B\triangle B^*|}
\]
over all $B\in\cB\setminus\{B^*\}$ is attained by a basis of the form $B'=B^*-e+f$. Thus the maximum is attained by a basis with $|B'\triangle B^*|=2$, refining the matroidal specialization of the min--max result in~\cite{berczi2023infinity}. We record a closely related statement too. In Figure~\ref{fig:example_IM_infty}, the basis is $B^*-ac+ab$.

\begin{cor} \label{cor:IM-offline-infty_symmDiff_equals_2}
    Let $M=(S,\cB)$ be a matroid, $B^*\in\cB$ be a fixed basis, and $w\in\bR^S$ be a weight function.
    \begin{enumerate}[label=(\alph*)]\itemsep0em
        \item If $B^*$ is not of maximum $w$-weight, then there exists $B'\in\cB$ with $|B'\triangle B^*|=2$ that maximizes $(w(B)-w(B^*))/|B\triangle B^*|$ over all $B\in\cB\setminus\{B^*\}$. \label{it:a}
        \item If $B^*$ is not of minimum $w$-weight, then there exists $B'\in\cB$ with $|B^*\triangle B'|=2$ that maximizes $(w(B^*)-w(B))/|B^*\triangle B|$ over all $B\in\cB\setminus\{B^*\}$. \label{it:b}
    \end{enumerate}
\end{cor}
\begin{proof}
Part~\ref{it:a} is a direct consequence of Theorem~\ref{thm:im_infty}, and~\ref{it:b} follows from~\ref{it:a} by multiplying $w$ by $-1$.
  \end{proof}

\begin{remark}\label{rem:integral-IM-infty}
For the \textsc{Integral-IM} problem under the $\ell_\infty$-norm, where $w\in\bZ^S$ and the modified weight function $w^*$ is required to be integer-valued, Corollary~\ref{cor:IM-offline-infty_symmDiff_equals_2} implies that the optimum value is $\lceil\delta^*\rceil$, where $\delta^*$ is the optimum value of the corresponding real-valued instance. Indeed, the algorithm of Theorem~\ref{thm:im_infty} computes $\delta^*$, and applying the same perturbation with $\lceil\delta^*\rceil$ yields an integer-valued feasible weight function.
\end{remark}

\subsection{IM under \texorpdfstring{$\ell_1$}{l-1}-norm}
\label{sec:IM-L1}

As in the $\ell_\infty$ case, our approach relies on the following known structural theorem for \textsc{IM}~\cite{mao1999inverse,dell2003base,zhang1999further}.

\begin{lemma}\label{lem:IM-L1-nice-structure}
Let $(M=(S,\cI),B^*,w,\|\cdot\|_1)$ be an IM instance.
Then there exists an optimal solution $w^*$ such that $w^*(s)\geq w(s)$ for every $s\in B^*$ and $w^*(s)\leq w(s)$ for every $s\in S\setminus B^*$.
\end{lemma}

The $\ell_1$ version of inverse matroid optimization is closely related to earlier work on inverse matroid intersection. Cai~\cite{mao1999inverse} studied inverse matroid intersection under a weighted $\ell_1$ objective and gave an LP formulation that can be transformed into a circulation problem. For the single-matroid basis problem considered here, Dell'Amico, Maffioli, and Malucelli~\cite{dell2003base} also derive a matching-based formulation before giving a faster algorithm via an auxiliary matroid construction. We include the short vertex-labeling proof below to keep the paper self-contained and to fix the notation used later; an example is shown in
Figure~\ref{fig:example_IM_L1}.

\begin{thm}\label{thm:IM_L1}
\textsc{IM} under the $\ell_1$-norm can be solved in polynomial time by making $O(rn)$ independence-oracle calls, plus the time needed to solve the minimum-cost feasible vertex-labeling problem on a bipartite graph with at most $O(rn)$ edges.
\end{thm}
\begin{proof}
If the original weight function $w$ is feasible, then it is already optimal. We may therefore assume that $B^*$ is not maximum-weight with respect to $w$. Consider an \textsc{IM} instance $(M=(S,\cI),B^*,w,\|\cdot\|_1)$. By Lemma~\ref{lem:IM-L1-nice-structure}, there is an optimal solution $w'$ such that $w'(e)\geq w(e)$ for every $e\in B^*$ and $w'(f)\leq w(f)$ for every $f\in S\setminus B^*$. Thus it is enough to search for an optimal solution of this form. Write
\[
    \pi(e)=w'(e)-w(e)\quad\text{for }e\in B^*
\]
and
\[
    \pi(f)=w(f)-w'(f)\quad\text{for }f\in S\setminus B^* .
\]
Then $\pi(s)\geq 0$ for every $s\in S$, and $\|w'-w\|_1=\sum_{s\in S}\pi(s)$.

We construct the fundamental circuit exchange graph of $B^*$, namely the bipartite graph $G=(B^*,S\setminus B^*;E)$ in which $ef\in E$, for $e\in B^*$ and $f\in S\setminus B^*$, if and only if $e\in C(B^*,f)$. By Proposition~\ref{prop:charact_optimum_basis}, the basis $B^*$ is maximum-weight with respect to $w'$ if and only if $w'(e)\geq w'(f)$ for every such pair $e,f$. Substituting $w'(e)=w(e)+\pi(e)$ and $w'(f)=w(f)-\pi(f)$, this condition becomes
\[
    \pi(e)+\pi(f)\geq w(f)-w(e)\qquad\text{for every }ef\in E .
\]
Edges with $w(e)\geq w(f)$ give redundant inequalities, since their right-hand side is nonpositive. We therefore keep only the edges with $w(f)>w(e)$, and for each remaining edge $ef$ we set
\[
    \alpha(ef)\coloneqq w(f)-w(e)>0 .
\]

The above shows that the optimum value of the \textsc{IM} instance is equal to the optimum value of
\[
\begin{aligned}
    \min \quad & \sum_{s\in S}\pi(s) \\
    \text{s.t.}\quad & \pi(e)+\pi(f)\geq \alpha(ef) \qquad &&\text{for every }ef\in E,\\
    & \pi(s)\geq 0 &&\text{for every }s\in S .
\end{aligned}
\]
Indeed, every feasible weight function $w'$ of the monotone form above gives a feasible vector $\pi$ of the same $\ell_1$-cost. Conversely, every feasible vector $\pi$ defines a weight function by setting $w'(e)=w(e)+\pi(e)$ for $e\in B^*$ and $w'(f)=w(f)-\pi(f)$ for $f\in S\setminus B^*$. For this weight function, the inequalities imply
\[
    w'(e)-w'(f)=\pi(e)+\pi(f)-\alpha(ef)\geq 0
\]
for every $f\in S\setminus B^*$ and every $e\in C(B^*,f)\setminus\{f\}$ with $w(f)>w(e)$. The remaining exchange inequalities satisfy $w(e)\geq w(f)$ and are automatic by the monotonicity of $w'$. Hence $B^*$ is maximum-weight by Proposition~\ref{prop:charact_optimum_basis}. Note that this is the minimum-cost feasible vertex-labeling problem on the bipartite graph $G$, equivalently the dual of the maximum-weight matching problem; see~\cite[Section~3.3.1]{frank2011book}. Hence it can be solved in polynomial time.

This gives an algorithm. We construct $G$ by testing all pairs $e\in B^*$ and $f\in S\setminus B^*$. For each pair, we check whether $B^*-e+f$ is a basis. If it is a basis and $w(f)>w(e)$, we add the edge $ef$ with demand $\alpha(ef)=w(f)-w(e)$. We then solve the minimum-cost feasible vertex-labeling problem, equivalently the dual of maximum-weight bipartite matching, and output the weight function $w'$ defined above.

\textit{Running time.} 
Constructing $G$ takes $r(n-r)=O(rn)$ independence-oracle calls and produces a graph with at most $r(n-r)=O(rn)$ edges. Therefore,
 overall we require $O(rn)$ independence-oracle calls, plus the time required to solve the
minimum-cost feasible vertex-labeling problem on the resulting
bipartite graph.
\end{proof}

\begin{remark}
The proof shows that \textsc{IM} under the $\ell_1$-norm reduces to a minimum-cost vertex-labeling problem on a bipartite graph, or equivalently, by Egerváry's theorem, to maximum-weight matching in $G$. This is in contrast with the $\ell_\infty$ case, where the optimum is determined by a single exchange.
\end{remark}

\begin{figure}[ht!]
\centering
    \begin{subfigure}[b]{0.25\textwidth}
        \centering
        \resizebox{\textwidth}{!}{
        \begin{tikzpicture}[
            myEdge/.style={line width = 1.5pt},     
            state/.style={circle,  minimum size=2em, draw, line width = 1.2pt}]   
            \node (a) at (0,2) [state] {$a$};
            \node (b) at (4,2) [state] {$b$};
            \node (c) at (1,1) [state] {$c$};
            \node (d) at (3,1) [state] {$d$};
            \node (e) at (0,0) [state] {$e$};
            \node (f) at (4,0) [state] {$f$}; 
            \draw [myEdge] (a) to node[midway, above] {7} (b) ; 
            \draw [myEdge] (e) to node[midway, below] {6} (f) ;
            \draw [myEdge] (a) to node[midway, left] {0} (e);
            \draw [myEdge] (d) to node[midway, left] {3} (f);
            \foreach \u \v \w \pos in {%
            a/c/{\rlap{0}\phantom{3.5..}}/right, c/d/6/above, d/b/6/left,  e/c/1/right, b/f/{\rlap{1}\phantom{3.5..}}/right%
            }{%
                \draw [MyOrange!30, line width = 5pt] (\u) to (\v) ; 
                \draw [myEdge] (\u) to node[midway,\pos] {\w} (\v) ; 
            }     
        \end{tikzpicture}}
        \caption{Original instance.}
        \label{fig:example_IM_L1_original}
    \end{subfigure}%
    \hfill
    \begin{subfigure}[b]{0.3\textwidth}
        \centering
        \resizebox{\textwidth}{!}{
        \begin{tikzpicture}[
            myEdge/.style={line width=1pt},
            dashEdge/.style={line width=1pt, dashed, gray},
            state/.style={rectangle, rounded corners=4pt,
                          minimum width=5em, minimum height=1em,
                          draw, line width=1pt},
            every node/.style={font=\small}
        ]
        
        \node (ac) at (0,4) [state] {$ac,\ w=0$};
        \node (ce) at (0,3) [state] {$ce,\ w=1$};
        \node (cd) at (0,2) [state] {$cd,\ w=6$};
        \node (db) at (0,1) [state] {$db,\ w=6$};
        \node (bf) at (0,0) [state] {$bf,\ w=1$};
        
        \node (ab) at (6,3.5) [state] {$ab,\ w=7$};
        \node (ae) at (6,2.5) [state] {$ae,\ w=0$};
        \node (ef) at (6,1.5) [state] {$ef,\ w=6$};
        \node (df) at (6,0.5) [state] {$df,\ w=3$};
        
        \node at (0,4.8) {$B^*$};
        \node at (6,4.8) {$S\setminus B^*$};
        
        \draw [myEdge] (ac) to node[pos=0.4, above, sloped] {$7$} (ab);
        \draw [myEdge] (cd) to node[pos=0.1, below, sloped] {$1$} (ab);
        \draw [myEdge] (db) to node[pos=0.3, above, sloped] {$1$} (ab);
        
        \draw [dashEdge] (ac) to node[pos=0.15, below, sloped] {$0$} (ae);
        \draw [dashEdge] (ce) to node[pos=0.2, above, sloped] {$-1$} (ae);
        
        \draw [myEdge] (ce) to node[pos=0.2, above, sloped] {$5$} (ef);
        \draw [dashEdge] (cd) to node[pos=0.25, below, sloped] {$0$} (ef);
        \draw [myEdge] (bf) to node[pos=0.15, above, sloped] {$5$} (ef);
        
        \draw [dashEdge] (db) to node[pos=0.2, above, sloped] {$-3$} (df);
        \draw [myEdge] (bf) to node[pos=0.2, below, sloped] {$2$} (df);
        
        \end{tikzpicture}}
        \caption{Exchange graph.}
        \label{fig:example_exchangeGraph}
    \end{subfigure}
    \hfill
    \begin{subfigure}[b]{0.25\textwidth}
        \centering
        \resizebox{\textwidth}{!}{
        \begin{tikzpicture}[
            myEdge/.style={line width = 1.5pt},     
            state/.style={circle,  minimum size=2em, draw, line width = 1.2pt}]   
            \node (a) at (0,2) [state] {$a$};
            \node (b) at (4,2) [state] {$b$};
            \node (c) at (1,1) [state] {$c$};
            \node (d) at (3,1) [state] {$d$};
            \node (e) at (0,0) [state] {$e$};
            \node (f) at (4,0) [state] {$f$}; 
            \draw [myEdge] (a) to node[midway, above] {6} (b) ; 
            \draw [myEdge] (e) to node[midway, below] {1} (f) ;
            \draw [myEdge] (a) to node[midway, left] {0} (e);
            \draw [myEdge] (d) to node[midway, left] {1} (f);
            \foreach \u \v \w \pos in {%
            a/c/{\rlap{6}\phantom{3.5..}}/right, c/d/6/above, d/b/6/left,  e/c/1/right, b/f/{\rlap{1}\phantom{3.5..}}/right%
            }{%
                \draw [MyOrange!30, line width = 5pt] (\u) to (\v) ; 
                \draw [myEdge] (\u) to node[midway,\pos] {\w} (\v) ; 
            }     
        \end{tikzpicture}}
        \caption{Optimal $w^*$ under the $\ell_1$-norm.}
        \label{fig:example_IM_L1_final}
    \end{subfigure}%
    \caption{Illustration of algorithms for \textsc{IM} on the graphic matroid in Figure~\ref{fig:example_IM_L1_original} under the $\ell_1$-norm (algorithm from Thm.~\ref{thm:IM_L1}). We obtain the exchange graph $G=(B^*, S\setminus B^*)$ as shown in Figure~\ref{fig:example_exchangeGraph}, where dashed edges correspond to redundant constraints and are omitted from the reduced exchange graph. The associated minimum-cost vertex-labeling problem has optimum value $14$, attained, for example, by $\pi=(6,0,0,0,0,1,0,5,2)$, indexed by $(ac,ce,cd,db,bf,ab,ae,ef,df)$, yielding the optimal weight function $w'=(6,1,6,6,1,6,0,1,1)$.
    } 
    \label{fig:example_IM_L1}
\end{figure}

\begin{remark}\label{rem:integral-IM-L1}
For the \textsc{Integral-IM} problem under the $\ell_1$-norm, where $w\in\bZ^S$ and the modified weight function $w^*$ is required to be integer-valued, the algorithm of Theorem~\ref{thm:IM_L1} also returns an optimal integer-valued solution. Integrality follows from the total dual integrality of the feasible vertex-labeling LP on bipartite graphs.
\end{remark}

\section{\textsc{Inverse Matroid Exists}}\label{sec:IM-Exists}

We now introduce several extensions of \textsc{IM} in which, instead of a fixed basis $B^*$, a subset $S_0$ is given, and various constraints are imposed on the set of optimal solutions depending on $S_0$. Our first goal is to find a weight function $w^*$ close to $w$ that ensures that $S_0$ contains a maximum $w^*$-weight basis.

\problemdef{Inverse Matroid Exists (IM-Exists)}
    {A matroid $M=(S,\mathcal{I})$, a subset $S_0 \subseteq S$ that contains a basis, a weight function $w\in\bR^S$, and an objective function $\|\cdot\|$ defined on $\bR^S$.}
    {Find a weight function $w^* \in \bR^S$ such that 
    \begin{enumerate}[label=(\alph*),topsep=1px,partopsep=2px]\itemsep0em
        \item \label{it:feas_cond_IM-Exists} there exists a basis contained in $S_0$ that has maximum $w^*$-weight, and
        \item \label{it:IM_c_IM-Exists} $\|w-w^*\|$ is minimized.
    \end{enumerate}}

As in \textsc{IM}, we say that a weight function is \emph{feasible} if condition~\ref{it:feas_cond_IM-Exists} holds. Note that \textsc{IM-Exists} is always feasible: set $w^*(s)=0$ for all $s \in S$.

\begin{remark}\label{rem:IM-Exists_feas}
    Given a weight function $w'$, we can verify in polynomial time whether $w'$ is feasible for \textsc{IM-Exists}. To do so, it is enough to find a maximum $w'$-weight basis $B_0$ of $M|S_0$ and a maximum $w'$-weight basis $A$ of $M$. Then, $w'$ is feasible if and only if $w'(B_0) = w'(A)$.
\end{remark}

\subsection{IM-Exists under \texorpdfstring{$\ell_\infty$}{l-infinity}-norm}

The $\ell_\infty$-norm is particularly well suited to \textsc{IM-Exists}. We first establish structural lemmas and then reduce \textsc{IM-Exists} under the $\ell_\infty$-norm to classical \textsc{IM}.


\subsubsection{Structural Properties}
\label{sec:prepexists}

First, we give a characterization of feasible weight functions. Recall from Section~\ref{sec:preliminaries} that, for any subset $S'\subseteq S$ and $x\in S$, we define $S'[w\geq w(x)]\coloneqq \{s\in S'\colon w(s)\geq w(x)\}$.

\begin{lemma}\label{lem:IM-Exists_closure_for_feas_weight}
    A weight function $w'$ is feasible for an \textsc{IM-Exists} instance $( M=(S, \cI),  S_0, w, \|\cdot\| )$ if and only if 
    $e \in \cl\bigl(S_0[w' \geq w'(e)]\bigr)$
    for all $e \in S\setminus S_0$.
\end{lemma}
\begin{proof}
    Suppose that $w'$ is feasible and that $B_0 \subseteq S_0$ is a basis of maximum $w'$-weight. Run the greedy algorithm with weight function $w'$, breaking ties in favor of elements in $S_0$. No element $e \in S \setminus S_0$ is selected, since the algorithm outputs $B_0$. This means that each $e \in S \setminus S_0$ is spanned by the set $B_0[w' \geq w'(e)]$. Therefore,
    \[
        e \in \cl\bigl(S_0[w' \geq w'(e)]\bigr),
    \]
    as claimed.

    Conversely, assume that
    \[
        e \in \cl\bigl(S_0[w' \geq w'(e)]\bigr)
    \]
    holds for all $e \in S \setminus S_0$. Then the greedy algorithm for weight function $w'$, where ties are broken by giving priority to elements in $S_0$ (and arbitrarily otherwise), returns a maximum $w'$-weight basis contained in $S_0$. Indeed, for every $e \in S \setminus S_0$, the set of elements considered before $e$---that is, those with weight at least $w'(e)$---already spans $e$. This concludes the proof.
  \end{proof}

Similarly to \textsc{IM}, \textsc{IM-Exists} admits an optimal weight function $w^*$ that has a special structure. 

\begin{lemma}\label{lem:IM-Exists_special_weight}
    Let $\delta^*$ be the optimum value of an \textsc{IM-Exists} instance $( M=(S, \cI),  S_0, w, \|\cdot\|_\infty )$. Then  
    \begin{equation*}
        w^*(s) = \begin{cases}
            w(s) +\delta^* &  \text{if $s \in S_0$}, \\
            w(s) -\delta^* & \text{if $s \in S \setminus S_0$}
        \end{cases}
    \end{equation*}
     is an optimal weight function.
\end{lemma}
\begin{proof}
Since $\|w-w^*\|_\infty = \delta^*$, it suffices to show that $w^*$ is feasible. Let $w^{\text{opt}}$ be an optimal weight function. By Lemma~\ref{lem:IM-Exists_closure_for_feas_weight}, we know that
\[
    e \in \cl\bigl(S_0[w^{\text{opt}} \geq w^{\text{opt}}(e)]\bigr)
\]
for all $e \in S\setminus S_0$. Moreover,
\[
    \cl\bigl(S_0[w^{\text{opt}} \geq w^{\text{opt}}(e)]\bigr) \subseteq \cl\bigl(S_0[w^* \geq w^*(e)]\bigr),
\]
since $w^*$ increases the weight of each element in $S_0$ and decreases the weight of each element in $S\setminus S_0$ at least as much as $w^{\text{opt}}$. Therefore, $w^*$ is also feasible by Lemma~\ref{lem:IM-Exists_closure_for_feas_weight}.
  \end{proof}

\begin{remark}
    This special structure for the optimal weight function does not always extend to other optimization problems. For example, consider perfect matchings on bipartite graphs. Take the graph with edge set $\{ac, ad, bc, bd\}$ and weights $w(ac)=w(ad)=w(bd)=0$ and $w(bc)=1$.
    Let $S_0=\{ac, ad, bd\}$. The optimal value is  $\delta^*=1/4$.
    However, increasing the weights of edges in $S_0$ and decreasing the weights of
    edges outside $S_0$ does not yield an optimal solution supported on $S_0$.
    An optimal weight function is $w^*(ac)=w^*(bd)=1/4$, $w^*(ad)=-1/4$ and $w^*(bc)=3/4$.
\end{remark}

\subsubsection{Reduction to \textsc{IM}}
\label{sec:algred}

The following reduction gives a clean conceptual explanation of the optimum value. It shows that, under the $\ell_\infty$-norm, the freedom to choose any basis contained in $S_0$ can be resolved before solving the inverse problem: it suffices to choose a maximum $w$-weight basis of the restriction $M|S_0$ and make this basis optimal in the original matroid. This procedure is implemented by Algorithm~\ref{algo:problem4_opt_devVector}. An illustration is shown in Figure~\ref{fig:example_IM-EXISTS}.

\begin{algorithm}[ht!]
\caption{Algorithm for \textsc{IM-Exists} via reduction to \textsc{IM} under the $\ell_\infty$-norm}\label{algo:problem4_opt_devVector}
\DontPrintSemicolon
\KwIn{An instance $(M=(S,\mathcal{I}),S_0,w,\|\cdot\|_\infty)$ of \textsc{IM-Exists} where $S_0$ contains at least one basis.}
\KwOut{An optimal weight function $w^*$.}
    Let $B_0$ be a maximum $w$-weight basis of $M|S_0$. \;
    Let $\delta^*$ be the optimum value of the \textsc{IM} instance $(M=(S,\cI),B_0,w,\|\cdot\|_\infty)$.\;
    Set $w^* \coloneqq w +\delta^* (\chi_{S_0} - \chi_{S\setminus S_0}) $.\;
    \Return{$w^*$}\;
\end{algorithm}

\begin{figure}[ht!]
\centering
    \begin{subfigure}[b]{0.33\textwidth}
        \centering
        \begin{tikzpicture}[scale=0.9,
            myEdge/.style={line width = 1.5pt},     
            myEdgeZigzag/.style={line width = 1.5pt, decoration = {zigzag, segment length = 4pt, amplitude = 1pt},decorate},
            state/.style={circle,  minimum size=1em, draw, line width = 1.2pt}]   
            \node (a) at (0,2) [state] {$a$};
            \node (b) at (4,2) [state] {$b$};
            \node (c) at (1,1) [state] {$c$};
            \node (d) at (3,1) [state] {$d$};
            \node (e) at (0,0) [state] {$e$};
            \node (f) at (4,0) [state] {$f$}; 
            \draw [myEdge] (a) to node[midway, above] {7} (b) ; 
            \draw [myEdge] (e) to node[midway, below] {6} (f) ;
            \foreach \u \v \w \pos in { a/c/0/right, c/d/6/above, d/b/6/left,  e/c/1/right, d/f/3/left}{
                \draw [MyOrange!30, line width = 5pt] (\u) to (\v) ; 
                \draw [myEdgeZigzag] (\u) to node[midway,\pos] {\w} (\v) ; 
            } 
            \foreach \u \v \w \pos in {a/e/0/left,b/f/1/right}{
                \draw [MyOrange!30, line width = 5pt] (\u) to (\v) ; 
                \draw [myEdge] (\u) to node[midway,\pos] {\w} (\v) ; 
            }
        \end{tikzpicture} 
        \caption{Original instance.}
        \label{fig:example_IM-EXISTS-original}
    \end{subfigure}%
    \begin{subfigure}[b]{0.33\textwidth}
        \centering
        \begin{tikzpicture}[scale=0.9,
            myEdge/.style={line width = 1.5pt},     
            state/.style={circle,  minimum size=1em, draw, line width = 1.2pt}]   
            \node (a) at (0,2) [state] {$a$};
            \node (b) at (4,2) [state] {$b$};
            \node (c) at (1,1) [state] {$c$};
            \node (d) at (3,1) [state] {$d$};
            \node (e) at (0,0) [state] {$e$};
            \node (f) at (4,0) [state] {$f$}; 
            \draw [myEdge] (a) to node[midway, above] {3.5} (b) ; 
            \draw [myEdge] (e) to node[midway, below] {2.5} (f) ;
            \foreach \u \v \w \pos in { a/c/3.5/right, c/d/9.5/above, d/b/9.5/left,  e/c/4.5/right, d/f/6.5/left}{
                \draw [MyOrange!30, line width = 5pt] (\u) to (\v) ; 
                \draw [myEdge] (\u) to node[midway,\pos] {\w} (\v) ; 
            } 
            \foreach \u \v \w \pos in {a/e/3.5/left,b/f/4.5/right}{
                \draw [MyOrange!30, line width = 5pt] (\u) to (\v) ; 
                \draw [myEdge] (\u) to node[midway,\pos] {\w} (\v) ; 
            }
        \end{tikzpicture}
        \caption{Optimal weight.}
        \label{fig:example_IM-EXISTS-FINAL}
    \end{subfigure}%
    \caption{Illustration of Algorithm~\ref{algo:problem4_opt_devVector} on the graphic matroid in Figure~\ref{fig:example_IM-EXISTS-original}. The set $S_0 = \{ac, ae, ce, cd, db, df, bf\}$ is highlighted in orange. 
    A maximum-weight basis $B_0\subseteq S_0$ is shown in zigzagged edges. Running the IM algorithm with target basis $B_0$ yields $\delta^*=3.5$. The optimal weight function $w^*$ is shown in Figure~\ref{fig:example_IM-EXISTS-FINAL}.
    } 
    \label{fig:example_IM-EXISTS}
\end{figure}

\begin{thm}\label{thm:IM-Exists-combAlgo}
    Algorithm~\ref{algo:problem4_opt_devVector} determines an optimal weight function $w^*$ for the \textsc{IM-Exists} instance $(M=(S, \cI),  S_0, w, \|\cdot\|_\infty )$ 
    using $O(n\log n)$ elementary operations and $O(n + r\log r)$ independence-oracle calls.
\end{thm}
\begin{proof}
    Consider an optimal weight function $w^*$ provided by Lemma~\ref{lem:IM-Exists_special_weight} with $\|w-w^*\|_\infty=\delta^*$. Let $B$ and $B_0$ be maximum $w^*$- and $w$-weight bases contained in $S_0$, respectively. Then 
\begin{equation*}
    w^*(B) = w(B) + \delta^* \cdot |B| \leq w(B_0) + \delta^* \cdot |B_0| =  w^*(B_0),
\end{equation*}
hence $B_0$ is also a maximum $w^*$-weight basis contained in $S_0$. Let $w'$ denote an optimal solution for the \textsc{IM} instance $( M=(S, \cI),  B_0, w, \|\cdot\|_\infty )$ and let $\delta'=\|w-w'\|_\infty$; clearly, $\delta' \leq \delta^*$. On the other hand, $\delta^*$ can be seen as the minimum value among all $\delta$'s that make at least one basis $B \subseteq S_0$ maximum-weight, and thus $\delta^* \leq  \delta'$. Thus the optimum value $\delta'$ of the fixed-basis \textsc{IM} instance equals the optimum value $\delta^*$ of the original \textsc{IM-Exists} instance. By Lemma~\ref{lem:IM-Exists_special_weight}, the uniform perturbation $w+\delta^*\bigl(\chi_{S_0}-\chi_{S\setminus S_0}\bigr)$, which is exactly the weight function returned by the algorithm, is feasible and has distance $\delta^*$. Hence it is optimal.

\textit{Running time.}
 Line 1 can be done using  $O(|S_0|\log |S_0|)$ elementary operations and $O(|S_0|)$ independence-oracle calls, and Line 2 requires $O(|S|\log |S|)$ elementary operations and $O(|S| + r\log r)$ independence-oracle calls by Theorem~\ref{thm:im_infty}.
  \end{proof}

Note that \textsc{IM-Exists} admits an algorithm with the same time complexity as \textsc{IM}. This algorithm extends naturally to handle integrality constraints, as the following remark shows.
\begin{remark}\label{rem:integral-IM-Exists-infty}
    For the \textsc{Integral-IM-Exists} problem under the $\ell_\infty$-norm, where $w\in\bZ^S$ and $w^*$ is required to be integer-valued, the optimum value is $\lceil\delta^*\rceil$, since every integral feasible solution has integral distance and the proof of Lemma~\ref{lem:IM-Exists_special_weight} shows that $w+\lceil\delta^*\rceil(\chi_{S_0}-\chi_{S\setminus S_0})$ is feasible.
\end{remark}

\subsection{IM-Exists under \texorpdfstring{$\ell_1$}{l-1}-norm}
\label{sec:hardness}

We now turn to the $\ell_1$-norm. This is the point where the behavior of \textsc{IM-Exists} changes sharply. In Section~\ref{sec:IM_infty}, we showed that the $\ell_\infty$ version admits a direct combinatorial algorithm, and Section~\ref{sec:IM-L1} showed that the fixed-basis problem under the $\ell_1$-norm reduces to a bipartite matching problem. In contrast, once the target basis is replaced by the requirement that some basis contained in $S_0$ become optimal, the additive nature of the $\ell_1$-norm is strong enough to encode the Steiner tree problem. We prove that the decision version of \textsc{IM-Exists} under the $\ell_1$-norm is strongly $\NP$-complete, even for graphic matroids. Thus the hardness is already present in one of the most basic matroid classes.

Let us recall that, in a graphic matroid, the ground set $S$ is the edge set of an undirected graph $G=(V,S)$, and the independent sets are the forests. Thus, if $G$ is connected, the bases are precisely the spanning trees of $G$. The decision problem under rational data is the following. 

\problemdef{Decision Inverse Graphic Matroid Exists (D-IM-Exists) }
    {A graphic matroid $M=(S,\mathcal{I})$, a subset $S_0 \subseteq S$ that contains a basis, a weight function $w\in\bQ^S$, a function $\|\cdot\|$ defined on $\bQ^S$, and $k\in \bQ_+$.}
    {Decide whether there exists a weight function $w^* \in \bQ^S$ such that 
    \begin{enumerate}[label=(\alph*),topsep=1px,partopsep=2px]\itemsep0em
        \item \label{it:feas_cond_D-IM-Exists} there exists a basis contained in $S_0$ that has maximum $w^*$-weight, and
        \item \label{it:IM_c_D-IM-Exists} $\|w-w^*\| \leq k$.
    \end{enumerate}}
We say that a weight function $w' \in \bQ^S$ is \emph{valid} for an instance of \textsc{D-IM-Exists} if it satisfies conditions \ref{it:feas_cond_D-IM-Exists} and \ref{it:IM_c_D-IM-Exists}.

\begin{thm}\label{thm:IM-Exists_L1_hardness}
    \textsc{D-IM-Exists} under the $\ell_1$-norm is strongly $\NP$-complete. 
\end{thm}
\begin{proof}
The problem belongs to $\NP$. A certificate is a basis $B\subseteq S_0$.
Given $B$, we first verify that it is a basis. Then Theorem~\ref{thm:IM_L1}
computes, in polynomial time, an optimal rational solution $w^*$ to the
corresponding \textsc{IM} instance under the $\ell_1$-norm; moreover, the
construction in its proof yields a solution of polynomial encoding length.
If the original instance is a yes-instance, then some feasible weight
function makes $B$ maximum-weight, and hence $\|w-w^*\|_1\le k$.
Conversely, if $\|w-w^*\|_1\le k$, then $B$ is maximum-weight under $w^*$
and Remark~\ref{rem:IM-Exists_feas} implies that $w^*$ is feasible for
\textsc{IM-Exists}.

    We reduce from the following problem, which is $\NP$-complete even in the unweighted case (i.e., when the objective is to minimize the number of edges) and when there are at least two terminals~\cite{Garey1979computers}.
    
    \problemdef{Steiner Tree}
    {An undirected connected graph $H$, a subset of terminals $X\subseteq V(H)$, and $k\in \mathbb{Z}_+$.}
    {Decide whether there exists a Steiner tree $F\subseteq E(H)$ with $|F|\leq k$, that is, a tree that connects all vertices in $X$.}

    Let $(H,X,k)$ be an instance of \textsc{Steiner Tree}. We may assume that $1\le k<|V(H)|$. Indeed, if $k=0$, the instance is a No-instance because $X$ contains at least two terminals. If $k\ge |V(H)|-1$, any spanning tree of $H$ is a Steiner tree with at most $k$ edges.
    
    We construct an instance of \textsc{D-IM-Exists} as follows: We define a graph $G$ with vertex set $V(G)=V(H)$ and edge set $E(G) = E(H) \cup P$, where for every two distinct terminals $u,v\in X$ we add $k|V(H)|$ new parallel copies of the edge $uv$ to $P$.
    Let $S_0=E(H)$. Let $S=E(G)$ and therefore $S_0=E(G)\setminus P$. The weight function $w$ is defined as $w(s)=0$ if $s\in S_0$ and $w(s)=1$ if $s\in P$. The budget of the \textsc{D-IM-Exists} instance is $k$.
    This concludes the description of the \textsc{D-IM-Exists} instance, which can be obtained in polynomial time. Below we show that this is a reduction.

Suppose that $(M,S_0,w,\|\cdot\|_1,k)$ is a Yes-instance of
\textsc{D-IM-Exists}. Let $\widehat w$ be an optimal solution to the
corresponding \textsc{IM-Exists} instance. Then
$\|w-\widehat w\|_1\le k$, and there exists a maximum-weight spanning
tree $T$ with respect to $\widehat w$ such that $T\subseteq S_0$.

Consider the \textsc{IM} instance with prescribed basis $T$. The weight
function $\widehat w$ is feasible for this instance. Conversely, every feasible solution to this \textsc{IM} instance makes
$T\subseteq S_0$ maximum-weight, and is therefore feasible for the
original \textsc{IM-Exists} instance. Therefore, the optimal deviation of the \textsc{IM} instance with
prescribed basis $T$ equals the optimal deviation of the original
\textsc{IM-Exists} instance.

By Theorem~\ref{thm:IM_L1} and
Remark~\ref{rem:integral-IM-L1}, the \textsc{IM} instance has an optimal
integer-valued solution $w'$. This solution is also optimal for
\textsc{IM-Exists}, and hence
\[
    \|w-w'\|_1=\|w-\widehat w\|_1\le k.
\]
Since $w$ and $w'$ are integer-valued, the weights of at most $k$
edges are changed.

Fix $u,v \in X$. Since there are $k|V(H)|>k$ parallel copies of the edge
$uv$ in $G$, at least one copy must have its weight unchanged in $w'$. Hence, there
exists an edge $e_{uv}\in P$ such that $w'(e_{uv}) = w(e_{uv}) = 1$.
Since $T$ is maximum-weight, every edge on the unique $u$-$v$ path in $T$
must have $w'$-weight at least $w'(e_{uv})=1$; otherwise, one could replace a
lighter edge on that path by the copy $e_{uv}$ and strictly increase the total
$w'$-weight.

Define $F_G \coloneqq  \{e \in T : w'(e) \ge 1 \}$. By the previous paragraph, one connected component of $F_G$ contains
all terminals. Since $F_G\subseteq T\subseteq S_0=E(H)$, this component
is a Steiner tree in $H$. Moreover, since $w$ is zero on $S_0$, each
edge of $F_G$ contributes at least $1$ to $\|w'-w\|_1$. Therefore, this
Steiner tree has at most
\[
    |F_G|\le \|w'-w\|_1\le k
\]
edges. Hence $H$ contains a Steiner tree with at most
$k$ edges, so $(H,X,k)$ is a Yes-instance of \textsc{Steiner Tree}.

Conversely, suppose that $(H,X,k)$ is a Yes-instance of \textsc{Steiner Tree},
and let $F$ be a Steiner tree for $X$ in $H$ with $|F|\le k$.
Define $w'$ by setting $w'(s)=1$ for $s\in F$ and $w'(s)=w(s)$ otherwise.
Then $\|w-w'\|_1=|F|\le k$.

Extend $F$ to a spanning tree $T$ of $H$ (so $T\subseteq S_0$) and note that
$w'(T)=|F|$. Let $E_1\coloneqq P\cup F$. For every $e=uv\in P$, $u$ and $v$ are
connected in $F$, hence $e\in\cl(F)$ and thus $E_1\subseteq \cl(F)$.
Therefore $r(E_1)\le r(F)=|F|$, and for any spanning tree $B$ of $G$ we have
$w'(B)=|B\cap E_1|\le r(E_1)\le |F|=w'(T)$. Hence $T$ has maximum $w'$-weight,
so $w'$ is valid for \textsc{D-IM-Exists}.

Since $k<|V(H)|$ and all weights in the constructed instance belong to
$\{0,1\}$, the reduction establishes strong $\NP$-hardness.
\end{proof}

Theorem~\ref{thm:IM-Exists_L1_hardness} also transfers to the other set
constraints discussed in Section~\ref{sec:implications}. In the following
variants, the input and budget condition are the same as in
\textsc{D-IM-Exists}; only condition~\ref{it:feas_cond_D-IM-Exists} is replaced.

\begin{cor}\label{cor:IM-Exists-four-hardness}
Together with Theorem~\ref{thm:IM-Exists_L1_hardness}, the following three
decision problems are strongly $\NP$-hard under the $\ell_1$-norm.
\begin{enumerate}[label=(\roman*)]\itemsep0em
    \item The partial inverse matroid problem with forced elements: decide
    whether there exists a weight function $w^*$ with $\|w-w^*\|_1\le k$
    such that some maximum $w^*$-weight basis $B$ satisfies $S_0\subseteq B$.
    This problem is strongly $\NP$-hard even for cographic matroids.

    \item The partial inverse matroid problem with forbidden elements: decide
    whether there exists a weight function $w^*$ with $\|w-w^*\|_1\le k$
    such that some maximum $w^*$-weight basis $B$ satisfies
    $B\cap S_0=\emptyset$. This problem is strongly $\NP$-hard even for
    graphic matroids.

    \item The complement-containment variant: decide whether there exists a
    weight function $w^*$ with $\|w-w^*\|_1\le k$ such that some maximum
    $w^*$-weight basis $B$ satisfies $S\setminus B\subseteq S_0$. This
    problem is strongly $\NP$-hard even for cographic matroids.
\end{enumerate}
\end{cor}

\begin{proof}
Theorem~\ref{thm:IM-Exists_L1_hardness} gives strong $\NP$-hardness for the
condition $B\subseteq S_0$ on graphic matroids.

The forbidden-element variant follows by replacing $S_0$ with
$S\setminus S_0$. Indeed, $B\subseteq S_0$ if and only if
$B\cap (S\setminus S_0)=\emptyset$. This reduction preserves the matroid,
the weight function, and the budget.

For the forced-element and complement-containment variants, we pass to the
dual matroid and negate the weights. Let $N=M^*$ and use the initial weight
function $-w$. A basis $B$ of $M$ is maximum-weight under $w'$ if and only if
the dual basis $C=S\setminus B$ is maximum-weight in $N$ under $-w'$. Moreover,
$\|w-w'\|_1=\|(-w)-(-w')\|_1$.

The condition $B\subseteq S_0$ is equivalent to
$S\setminus S_0\subseteq C$. This gives the forced-element variant with
prescribed set $S\setminus S_0$. Since $B=S\setminus C$, the same condition
is also $S\setminus C\subseteq S_0$, which gives the
complement-containment variant with prescribed set $S_0$.

The dual of a graphic matroid is cographic, so both variants are strongly
$\NP$-hard even for cographic matroids.

\end{proof}

\section{\textsc{IM-All} and \textsc{IM-Only}}
\label{sec:IM-All}
In this section, we consider two variants in which the whole family of bases contained in $S_0$ has to be controlled. In \textsc{IM-All}, every basis contained in $S_0$ must become maximum-weight; in \textsc{IM-Only}, these bases must be exactly the maximum-weight bases. The key new issue is that one no longer chooses a single target basis. Instead, all bases of the restriction $M|S_0$ have to be made indistinguishable by weight. For matroids, this requirement has a clean structural form: it is equivalent to making the weight function constant on each connected component of $M|S_0$. This homogenization step is the main reason the problems remain tractable.

\problemdef{Inverse Matroid All (\textsc{IM-All})}
    {A matroid $M=(S,\mathcal{I})$, a subset $S_0 \subseteq S$, a weight function $w\in\bR^S$, and an objective function $\|\cdot\|$ defined on $\bR^S$.}
    {Find a weight function $w^* \in \bR^S$ such that 
    \begin{enumerate}[label=(\alph*),topsep=1px,partopsep=2px]\itemsep0em
        \item{\label{it:feas_cond_IM-All}} all bases contained in $S_0$ have maximum $w^*$-weight, and
        \item \label{it:IM_c_IM-All} $\|w- w^*\|$ is minimized.
    \end{enumerate}}

We say that a weight function is \emph{feasible} if condition~\ref{it:feas_cond_IM-All} holds. If $S_0$ does not contain a basis of $M$, then the condition is vacuous and the optimal solution is $w^*=w$. Hence, in the rest of the section, we assume that $S_0$ contains at least one basis of $M$.

\begin{remark}\label{rem:feasibility_check_IM-All}
    Given a weight function $w'$, we can verify in polynomial time whether $w'$ is feasible for \textsc{IM-All}. To do so, it is enough to find a minimum $w'$-weight basis $B^1_0$ of $M|S_0$,  a maximum $w'$-weight basis $B^2_0$ of $M|S_0$, and a maximum $w'$-weight basis $A$ of $M$. Then, $w'$ is feasible if and only if $w'(B^2_0)=w'(B^1_0)\geq w'(A)$.
\end{remark}

In \textsc{IM-Only}, the bases contained in $S_0$ must be exactly the maximum-weight bases. If real weights are allowed, an optimal solution need not exist. Indeed, a weight function may make all bases in $S_0$ maximum-weight while tying with some basis outside $S_0$. An arbitrarily small perturbation can break this tie. Hence the infimum may fail to be attained. We avoid this by requiring $w^*$ to be integer-valued: then every strict separation has gap at least one.

Ahmadian, Bhaskar, Sanit{\`a}, and Swamy~\cite{ahmadian2018algorithms} considered an integral inverse optimization problem in which a prescribed collection of feasible solutions is required to be exactly the set of optimal solutions. For matroid bases, this overlaps with the integral $\ell_\infty$ version of \textsc{IM-Only}. Our contribution here is a direct matroidal treatment based on the connected components of $M|S_0$. The same viewpoint also yields the algorithms for \textsc{IM-All}, the $\ell_1$ versions, and the negated variants studied later.

Let us state the problem.

\problemdef{Inverse Matroid Only (\textsc{IM-Only})}
    {A matroid $M=(S,\mathcal{I})$, a subset $S_0 \subseteq S$ that contains at least one basis, a weight function $w\in\bZ^S$, and an objective function $\|\cdot\|$ defined on $\bR^S$.}
    {Find a weight function $w^* \in \bZ^S$ such that 
    \begin{enumerate}[label=(\alph*),topsep=1px,partopsep=2px]\itemsep0em
        \item{\label{it:feas_cond_IM-Only}} the bases contained in $S_0$ are exactly the ones of maximum $w^*$-weight, and
        \item \label{it:IM_c_IM-Only} $\|w- w^*\|$ is minimized.
    \end{enumerate}}
    
We say that a weight function is \emph{feasible} if condition~\ref{it:feas_cond_IM-Only} holds. Note that \textsc{IM-Only} is always feasible: set $w^*(s)=1$ for all $s\in S_0$ and $w^*(s)=0$ for all $s\in S\setminus S_0$.

\begin{remark}\label{rem:feas_IM-ONLY}
    Given a weight function $w'$, we can verify in polynomial time whether $w'$ is feasible for \textsc{IM-Only}. By assumption, $S_0$ contains at least one basis of $M$. Thus it is enough to find a minimum $w'$-weight basis $B_0^1$ of $M|S_0$ and a maximum $w'$-weight basis $B_0^2$ of $M|S_0$, and check if $w'(B_0^2)=w'(B_0^1)$. If not, then $w'$ is not feasible. Otherwise, run the greedy algorithm while breaking ties in favor of elements in $S\setminus S_0$. If the resulting basis is not contained in $S_0$, then $w'$ is not feasible; otherwise it is feasible.
\end{remark}

\subsection{A general convex formulation}

Before exploiting the matroid structure, we record a general convex formulation for the analogue of \textsc{IM-All} over an arbitrary family of feasible sets. This shows that the requirement that all feasible sets contained in $S_0$ be optimal is not inherently matroidal: if the underlying optimization problem can be solved, then the corresponding \textsc{IO-All} problem can be solved in weakly polynomial time for any convex deviation objective with a suitable separation oracle. The formulation extends the standard inverse-optimization approach of Ahuja and Orlin~\cite{Ahuja2001Inverse} by adding constraints that force all feasible solutions contained in $S_0$ to have the same weight. What is special about matroids is that this condition admits the much more explicit connected-component formulation developed in the next subsection.

\problemdef{Inverse Optimization All (\textsc{IO-All})}
    {A ground set $S$, a set of feasible solutions $\mathcal{F}$, a subset $S_0 \subseteq S$ that contains at least one feasible solution, a weight function $w\in\bR^S$, and an objective function $\phi$ defined on $\bR^S$.}
    {Find a weight function $w^* \in \bR^S$ such that
    \begin{enumerate}[label=(\alph*),topsep=1px,partopsep=2px]\itemsep0em
        \item \label{it:feas_cond_all} all feasible solutions contained in $S_0$ have maximum $w^*$-weight, and
        \item \label{it:opt_cond_all} $\phi(w-w^*)$ is minimized.
    \end{enumerate}}

\begin{lemma}\label{lem:IO-All-convexFunc}
Suppose the underlying optimization problem $(S,\mathcal{F})$ can be solved for maximization with respect to every weight function $c$ using an oracle $\mathcal{O}$. Then the feasible region of the \textsc{IO-All} instance $(S,\mathcal{F},S_0,w,\phi)$ admits a polynomial-time separation oracle. Consequently, under the standard oracle and boundedness assumptions for convex minimization, an optimal weight function for \textsc{IO-All} can be computed in weakly polynomial time.
\end{lemma}
\begin{proof}

Let $\mathcal F_0$ be the set of feasible solutions contained in $S_0$.
Since $\mathcal F_0\neq\emptyset$, one call to $\mathcal O$ with weights
$0$ on $S_0$ and $-1$ on $S\setminus S_0$ returns a set
$F_0\in\mathcal F_0$.

    Consider the following convex program over variables $\Delta, w^* \in \mathbb{R}^S$:
 \begin{alignat}{3}
        &\min\quad & \phi(\Delta) & \notag \\
        &\text{\rm s.t.}\quad & \Delta(e) & = w(e) - w^*(e) &\qquad& \forall e \in S, \label{eq:IO-ALL-delta} \\
        & & w^*(F_0) & \leq w^*(F) && \forall F \in \mathcal{F}_0, \label{eq:LP-IO-ALL-HOMOGEN}\\
        & & w^*(F_0) & \geq w^*(H) && \forall H \in \mathcal{F}. \label{eq:LP-IO-ALL-MAX}
    \end{alignat}
    Observe that the two families of inequalities imply that $w^*(F_0)=w^*(F)$ for all $F\in\mathcal{F}_0$. Hence all feasible solutions contained in $S_0$ have the same $w^*$-weight, and by~\eqref{eq:LP-IO-ALL-MAX} this common weight is maximum.
    
    We next describe a separation oracle. Let $(\Delta', w')$ be a candidate assignment. Checking the equality constraint~\eqref{eq:IO-ALL-delta} takes linear time. To check~\eqref{eq:LP-IO-ALL-MAX}, compute a maximum-weight feasible set $F^{\max}$ with respect to $w'$ using the oracle $\mathcal{O}$. If $w'(F^{\max})> w'(F_0)$, then~\eqref{eq:LP-IO-ALL-MAX} is violated by $H=F^{\max}$; otherwise it is satisfied. 
    To check~\eqref{eq:LP-IO-ALL-HOMOGEN}, compute a minimum-weight feasible set 
    $F_0^{\min}\in\mathcal{F}_0$ with respect to $w'$. This can be done with one call to $\mathcal{O}$ by maximizing $-w'$ after assigning weight $-L$ to every element in $S\setminus S_0$, where $L$ is sufficiently large, for example, $L=2|S|\|w'\|_\infty+1$.  If $w'(F_0)> w'(F_0^{\min})$, then~\eqref{eq:LP-IO-ALL-HOMOGEN} is violated by $F=F_0^{\min}$; otherwise it is satisfied.
    
    Thus the separation problem can be solved using two calls to $\mathcal{O}$. Together with the standard oracle assumptions for the convex objective $\phi$, this separation oracle allows the convex program to be solved in weakly polynomial time by the ellipsoid method.
\end{proof}

\subsection{Matroid reformulation using connected components}

For matroids, the general formulation above can be replaced by a much smaller and more transparent one. The key structural fact is that all bases of a matroid have the same weight if and only if the weight function is constant on each connected component. Applied to the restriction $M|S_0$, this means that the first part of \textsc{IM-All} is exactly a homogenization problem on the connected components of $M|S_0$. The same formulation will also be used for \textsc{IM-Only}, where one additionally has to separate bases contained in $S_0$ from bases not contained in $S_0$. We include the proof of the structural fact for completeness.

\begin{lemma}\label{lem:connected_components_decomp}
    Let $M=(S,\cI)$ be a matroid and let $w\in\mathbb{R}^S$ be a weight function. Then, every basis has the same $w$-weight if and only if $w$ is constant on the connected components of $M$.
\end{lemma} 
\begin{proof}
    Let $S=S_1\cup\dots\cup S_k$ be the partition of $S$ into the connected components of $M$, and let $M_i=M|S_i$ for $i\in[k]$. Since $M=M_1\oplus\dots\oplus M_k$,  a set $B\subseteq S$ is a basis of $M$ if and only if $B_i=B\cap S_i$ is a basis of $M_i$ for all $i\in[k]$.

    Suppose first that $w$ is constant on each $S_i$. Then all bases of $M_i$ have the same $w$-weight, namely $w(e_i)r(M_i)$ for any $e_i\in S_i$. Since every basis of $M$ is the disjoint union of bases of $M_1,\dots,M_k$, all bases of $M$ have the same $w$-weight.

    Conversely, suppose that $w$ is not constant on some component. Then there are $i\in[k]$ and $e,f\in S_i$ with $w(e)\neq w(f)$. Since $e$ and $f$ are in the same component of $M$, there exists a circuit $C$ containing both of them. Extend the independent set $C-f$ to a basis $B$ of $M$. Then $f\not\in B$ and $C=C(B,f)$ is the fundamental circuit of $f$ with respect to $B$. As $e\in C(B,f)$, the set $B'=B-e+f$ is a basis with $w(B')-w(B)=w(f)-w(e)\neq 0$, showing that $M$ has bases with different $w$-weights.
  \end{proof}

\begin{remark}
The connected-component formulation is specific to matroids. For example, in matching problems, two matchings may have the same total weight without imposing any comparable equality condition on individual edge weights.
\end{remark}
  
A weight function is feasible for \textsc{IM-All} precisely when the bases contained in $S_0$ have a common weight and this common weight is maximum among all bases of $M$. We analyze the two requirements separately. First, apply Lemma~\ref{lem:connected_components_decomp} to the restriction $M|S_0$: The bases contained in $S_0$ have the same weight if and only if the weight function is constant on each connected component of $M|S_0$. We denote these components by $K_1,\dots,K_\ell$.
  
    Second, the requirement that this common weight be maximum means that no element can strictly improve a basis contained in $S_0$. To make this precise, fix an arbitrary basis $B_0\subseteq S_0$. By Proposition~\ref{prop:charact_optimum_basis}, $B_0$ is a maximum-weight basis of $M$ if and only if $w^*(e)\ge w^*(f)$ for every $f\in S\setminus B_0$ and every $e\in C(B_0,f)\setminus\{f\}$. 
    
    Since $w^*$ is constant on the connected components of $M|S_0$, these inequalities hold with equality for all $f\in S_0\setminus B_0$. Indeed, in that case $C(B_0,f)$ is a circuit of $M|S_0$, so all its elements lie in the same component of $M|S_0$. Thus, it remains only to impose the inequalities for $f\in S\setminus S_0$. For such an element $f$, we have $C(B_0,f)\setminus\{f\}\subseteq B_0\subseteq S_0$, and therefore every element of $C(B_0,f)\setminus\{f\}$ belongs to a unique component $K_i$. Define $K(f)\subseteq[\ell]$ as the set of indices $i$ such that $(C(B_0,f)\setminus\{f\})\cap K_i\neq\emptyset$. We now check that this definition does not depend on the choice of $B_0$.
    
Indeed, fix $i\in[\ell]$. Since $B_0\cap(S_0\setminus K_i)$ is a basis of $M|(S_0\setminus K_i)$, the fundamental-circuit characterization of closure gives \[ (C(B_0,f)\setminus\{f\})\cap K_i=\emptyset \quad\iff\quad f\in\cl_M(S_0\setminus K_i). \] Equivalently, this holds if and only if $f$ is a loop in the contraction $M/(S_0\setminus K_i)$. This condition does not involve $B_0$, so $K(f)$ is well defined. Hence $w^*$ is feasible for \textsc{IM-All} if and only if it assigns a constant value $x_i$ to each component $K_i$, assigns a value $y_f$ to each element $f\in S\setminus S_0$, and satisfies $x_i\ge y_f$ whenever $i\in K(f)$.  This characterization reduces \textsc{IM-All} to a polynomial-size separable convex minimization program with $O(|S|)$ variables and $O(|S|^2)$ constraints.
\begin{lemma}\label{lem:separable_convex_prog}
    Let $(M=(S,\cI),S_0,w,\phi)$ be an instance of \textsc{IM-All}, and let $K_1,\dots,K_\ell$ denote the connected components of the restriction $M|_{S_0}$. Then the following hold: 
    \begin{enumerate}[label=(\alph*)]\itemsep0em 
    \item the problem is equivalent to the convex minimization program over variables $\Delta\in\mathbb{R}^S$, $x\in\mathbb{R}^\ell$, and $y\in\mathbb{R}^{S\setminus S_0}$ given by
    \begin{alignat}{3}
        &\min\quad & \phi(\Delta) & \notag \\
        &\text{\rm s.t.}\quad & \Delta(e) & = w(e) - x_i &\qquad& \forall i\in [\ell],\ \forall e\in K_i, \label{eq:separable_delta_e} \\
        & & \Delta(f) & = w(f) - y_f && \forall f\in S\setminus S_0, \label{eq:separable_delta_f} \\
        & & x_i & \ge y_f && \forall f\in S\setminus S_0,\ \forall i\in K(f). \label{eq:separable_constraints}
    \end{alignat}
        and 
        \item given an optimal solution $(\Delta,x,y)$, an optimal weight function is obtained by setting \[ w^*(e)=x_i \quad \text{for every } e\in K_i, \qquad\text{and}\qquad w^*(f)=y_f \quad \text{for every } f\in S\setminus S_0. \]
    \end{enumerate}
\end{lemma}
\begin{proof}
    This is exactly the characterization above, with $\Delta=w-w^*$ used to write the objective as $\phi(\Delta)$.
\end{proof}

\begin{remark}\label{obs:im_only_strict}
    The \textsc{IM-Only} problem requires the weight function $w^*$ to be integral. A weight function $w^*$ is feasible for \textsc{IM-Only} if and only if the bases contained in $S_0$ are the only maximum-weight bases of $M$. Over the integers, this strict condition corresponds to the shifted inequality $x_i \ge y_f + 1$. Consequently, \textsc{IM-Only} is equivalent to the integer version of the program in Lemma~\ref{lem:separable_convex_prog}, with the additional constraints $x \in \mathbb{Z}^\ell$, $y \in \mathbb{Z}^{S \setminus S_0}$, and replacing~\eqref{eq:separable_constraints} by $x_i \ge y_f + 1$.
\end{remark}

We use this separable formulation as inspiration to derive combinatorial algorithms for the $\ell_\infty$- and $\ell_1$-norms.

\subsection{IM-All and IM-Only under \texorpdfstring{$\ell_\infty$}{l-infinity}-norm}

\subsubsection{Structural Properties}

For the $\ell_\infty$-norm, Lemma~\ref{lem:separable_convex_prog} gives the objective
$\max \big\{ \max_{i \in [\ell], e \in K_i} |w(e)-x_i|,\ \max_{f \in S \setminus S_0} |w(f)-y_f| \big\}$,
subject to $x_i \ge y_f$ for all $f \in S \setminus S_0$ and $i \in K(f)$.
\begin{lemma}\label{lem:hm}
    Let $w_\textsc{Feas}$ be a feasible solution for the \textsc{IM-All} or \textsc{IM-Only} instance $(M=(S,\cI),S_0,w,\allowbreak \|\cdot\|_\infty)$, and let $\delta=\|w-w_\textsc{Feas}\|_\infty$. Then there exists a feasible solution $w'_\textsc{Feas}$ satisfying the following:
    \begin{enumerate}[label=(\alph*)]\itemsep0em
        \item $\|w-w'_\textsc{Feas}\|_\infty=\delta$, \label{prop:a}
        \item for every component $K_i$ of $M|S_0$, there exists $e\in K_i$ with $w'_\textsc{Feas}(e)-w(e)=\delta$, and \label{prop:b}
        \item $w'_\textsc{Feas}(f)-w(f)=-\delta$ for all $f\in S\setminus S_0$. \label{prop:c}
    \end{enumerate}
\end{lemma}

\begin{proof}
    By Lemma~\ref{lem:separable_convex_prog}, $w_\textsc{Feas}$ corresponds to an assignment $(x,y)$ satisfying the feasibility constraints. Define $x'_i=\min_{e\in K_i}\{w(e)+\delta\}$ for every $i\in[\ell]$, and $y'_f=w(f)-\delta$ for every $f\in S\setminus S_0$.

    Since $\|w-w_\textsc{Feas}\|_\infty=\delta$, we have $w(e)-\delta\le x_i\le w(e)+\delta$ for every $e\in K_i$. Hence $x'_i\ge x_i$. Also, $x'_i\le w(e)+\delta$ for every $e\in K_i$, so $|w(e)-x'_i|\le\delta$. If $e_i\in K_i$ attains the minimum defining $x'_i$, then $x'_i-w(e_i)=\delta$, proving condition~\ref{prop:b}.

    Similarly, $y'_f\le y_f$ and $y'_f-w(f)=-\delta$ for every $f\in S\setminus S_0$, proving condition~\ref{prop:c}. Thus the $\ell_\infty$-distance from $w$ is exactly $\delta$.

    Finally, feasibility is preserved because each $x_i$ is only increased and each $y_f$ is only decreased. Thus every constraint $x_i\ge y_f$ remains valid. For \textsc{IM-Only}, the same argument preserves the shifted constraints $x_i\ge y_f+1$. Therefore $(x',y')$ defines a feasible weight function $w'_\textsc{Feas}$ satisfying the three conditions.
\end{proof}

\subsubsection{Combinatorial Algorithm for \textsc{IM-All} under \texorpdfstring{$\ell_\infty$}{l-infty}-norm}
\label{sec:algall}

Our algorithm for \textsc{IM-All} relies on
Proposition~\ref{prop:charact_optimum_basis} and
Lemma~\ref{lem:connected_components_decomp}. These results imply that, for any
basis $B_0\subseteq S_0$, a feasible weight function $w'$ must satisfy
$w'(y)=w'(x)$ for all $y\in S_0\setminus B_0$ and
$x\in C(B_0,y)$, and $w'(y)\le w'(x)$ for all
$y\in S\setminus S_0$ and $x\in C(B_0,y)$. The high-level idea of the algorithm
is as follows.

\paragraph{Phase 1: Homogenization.}
First, we determine the connected components $K_1,\dots,K_\ell$ of $M|S_0$.
For each component $K_i$, we replace the weights of all elements in $K_i$ by a
common value. The common values are chosen so that the largest change inside
$S_0$ is as small as possible. We denote the resulting weight function by
$w_H$.

\paragraph{Phase 2: Maximization.}
We compute the optimal value $\Delta$ of the \textsc{IM-Exists} instance
$(M,S_0,w_H,\|\cdot\|_\infty)$. We then define $w^*$ by setting
$w^*(s)=w_H(s)+\Delta$ for all $s\in S_0$, and
$w^*(s)=w_H(s)-\Delta$ for all $s\in S\setminus S_0$.

The formal description of the algorithm is given as Algorithm~\ref{alg:IM-All}; see Figure~\ref{fig:example_IM-All} for an illustration.

\begin{algorithm}[ht!]
\caption{Combinatorial algorithm for \textsc{IM-All} under the $\ell_\infty$-norm}\label{alg:IM-All}
\DontPrintSemicolon
\KwIn{An instance $(M=(S,\mathcal{I}),S_0,w,\|\cdot\|_\infty)$ of \textsc{IM-All}.}
\KwOut{An optimal weight function $w^*$.}
    \lIf{$M$ has no basis contained in $S_0$}{
        Set $w^* \coloneqq w$.        
    }
    \Else{
    \tcp*[l]{Phase 1: Homogenization}
        Determine the connected components $K_1,\dots,K_\ell$ of $M|S_0$.\;
        \For{$i \in [\ell]$}{
        Set $\delta_i^{\max}  \coloneqq \max \{ w(e) \colon e \in K_i \}$ and $\delta_i^{\min}  \coloneqq \min \{ w(e) \colon e \in K_i \} $.\;
        Set $m_i  \coloneqq (\delta_i^{\max} + \delta_i^{\min})/2$.\;
}
Set $\varrho \coloneqq \max_{i \in [\ell]} \{ \delta_i^{\max} - m_i\}$.\label{step:rho}\;
Set $\operatorname{shift}_i \coloneqq \varrho - \left( \delta_i^{\max} - m_i \right)$ for $i \in [\ell]$.\;
Set $w_H(s)\coloneqq m_i + \operatorname{shift}_i$ if $s \in K_i$ for $i \in [\ell]$ and $w_H(s)\coloneqq w(s) - \varrho$ if $s \in S \setminus S_0$ .\;
    \tcp*[l]{Phase 2: Maximization}
        Let $\Delta$ be the optimum value of the \textsc{IM-Exists} instance $(M=(S, \cI), S_0, w_H, \| \cdot \|_\infty)$.\;
        Set $w^*(s)\coloneqq w_H(s)+\Delta$ if $s\in S_0$, and $w^*(s)\coloneqq w_H(s)-\Delta$ if $s\in S\setminus S_0$. 
        \;
    }
    \Return{$w^*$.}\;
\end{algorithm}

\begin{thm}\label{thm:IM-All_algorithm}
    Algorithm~\ref{alg:IM-All} determines an optimal weight function $w^*$ for the \textsc{IM-All} instance $(M=(S,\cI),S_0,w,\|\cdot\|_\infty)$ using  $O(n\log n)$ elementary operations and $O(nr)$ independence-oracle calls.
\end{thm}
\begin{proof}
    The algorithm outputs a weight function $w^*$ that satisfies the conditions
    of Proposition~\ref{prop:charact_optimum_basis} and
    Lemma~\ref{lem:connected_components_decomp}, and is therefore feasible.

    Let $\delta^*=\|w-w^*\|_\infty$. Since Phase 1 is the best possible
    homogenization of the weights on the components of $M|S_0$, every feasible
    solution has distance at least $\varrho$ from $w$.

    To prove optimality, let $w^*_{\textsc{All}}$ be an optimal solution
    satisfying conditions~\ref{prop:a}--\ref{prop:c} of Lemma~\ref{lem:hm}, and
    let $\delta=\|w-w^*_{\textsc{All}}\|_\infty$. Since feasible solutions are
    constant on each component $K_i$ of $M|S_0$, condition~\ref{prop:b} gives
    \[
    w^*_{\textsc{All}}(s)=\delta_i^{\min}+\delta
    \qquad \forall s\in K_i.
    \]
    On the other hand, Phase 1 sets
    \[
    w_H(s)=\delta_i^{\min}+\varrho
    \qquad \forall s\in K_i.
    \]
    Therefore
    \[
    w_H(s)=w^*_{\textsc{All}}(s)-\delta+\varrho
    \qquad \forall s\in S_0.
    \]
    Also, by condition~\ref{prop:c},
    \[
    w^*_{\textsc{All}}(f)=w(f)-\delta
    \qquad \forall f\in S\setminus S_0,
    \]
    while Phase 1 sets $w_H(f)=w(f)-\varrho$. Hence
    \[
    w_H(f)=w^*_{\textsc{All}}(f)+\delta-\varrho
    \qquad \forall f\in S\setminus S_0.
    \]

    Thus $w^*_{\textsc{All}}$ is obtained from $w_H$ by increasing all weights
    in $S_0$ by $\delta-\varrho$ and decreasing all weights in
    $S\setminus S_0$ by $\delta-\varrho$. Since $w^*_{\textsc{All}}$ is
    feasible, the shift $\delta-\varrho$ is feasible for the
    \textsc{IM-Exists} instance solved in Phase 2. By the definition of
    $\Delta$, we have
    \[
    \Delta\le \delta-\varrho.
    \]
    The output of the algorithm has distance $\varrho+\Delta$ from $w$. Hence
    \[
    \delta^*=\varrho+\Delta\le \delta.
    \]
    Since $\delta$ is the optimum value, $w^*$ is optimal.

   \textit{Running time.} 
   Detecting whether $M$ has a basis contained in $S_0$ can be done by running the Greedy algorithm giving preferences to elements in $S_0$ using $O(n\log n)$ elementary operations and $O(n)$ independence-oracle calls.
   As mentioned earlier in Section~\ref{sec:preliminaries}, the connected     components of $M|S_0$ can be determined using $O(nr)$ independence-oracle calls~\cite{krogdahl1977dependence} by computing the
    connected components of the hypergraph of fundamental circuits for elements
    in $S_0\setminus B_0$, with respect to an arbitrary basis $B_0$ of
    $M|S_0$.
   Once the connected components are known, the weight function $w_H$ can
    be constructed in $O(n)$ time by scanning the elements of $S_0$ to
    compute the minimum and maximum weight in each component and then
    assigning the new weights.
    Finally, computing the value of $\Delta$ amounts to solving an \textsc{IM-Exists} instance, which requires $O(|S|\log |S|)$ elementary operations and $O(|S| + r\log r)$ independence-oracle calls by Theorem~\ref{thm:IM-Exists-combAlgo}.
\end{proof}

\begin{figure}[ht!]
\centering
    \begin{subfigure}[b]{0.33\textwidth}
        \centering
        \begin{tikzpicture}[scale=0.9,
            myEdge/.style={line width = 1.5pt},     
            state/.style={circle,  minimum size=1em, draw, line width = 1.2pt}]   
            \node (a) at (0,2) [state] {$a$};
            \node (b) at (4,2) [state] {$b$};
            \node (c) at (1,1) [state] {$c$};
            \node (d) at (3,1) [state] {$d$};
            \node (e) at (0,0) [state] {$e$};
            \node (f) at (4,0) [state] {$f$}; 
            \draw [myEdge] (a) to node[midway, above] {7} (b) ; 
            \draw [myEdge] (e) to node[midway, below] {6} (f) ;
            \foreach \u \v \w \pos in {a/c/{\rlap{0}\phantom{3.5..}}/right, c/d/6/above, d/b/6/left, a/e/0/left, e/c/1/right, d/f/3/left, b/f/{\rlap{1}\phantom{3.5..}}/right}{
                \draw [MyOrange!30, line width = 5pt] (\u) to (\v) ; 
                \draw [myEdge] (\u) to node[midway,\pos] {\w} (\v) ; 
            }     
        \end{tikzpicture} 
        \caption{Original instance.}
        \label{fig:example_IM-All_original}
    \end{subfigure}%
    \hfill
    \begin{subfigure}[b]{0.33\textwidth}
        \centering
        \begin{tikzpicture}[scale=0.9,
            myEdge/.style={line width = 1.5pt},
            myEdgeZigzag/.style={line width = 1.5pt, decoration = {zigzag, segment length = 4pt, amplitude = 1pt},decorate},
            state/.style={circle,  minimum size=1em, draw, line width = 1.2pt}]   
            \node (a) at (0,2) [state] {$a$};
            \node (b) at (4,2) [state] {$b$};
            \node (c) at (1,1) [state] {$c$};
            \node (d) at (3,1) [state] {$d$};
            \node (e) at (0,0) [state] {$e$};
            \node (f) at (4,0) [state] {$f$}; 
            \draw [myEdge] (a) to node[midway, above] {4.5} (b) ; 
            \draw [myEdge] (e) to node[midway, below] {3.5} (f) ;
            \foreach \u \v \w \pos in { a/c/2.5/right, c/d/8.5/above, d/b/3.5/left,  e/c/2.5/right, d/f/3.5/left}{
                \draw [MyOrange!30, line width = 5pt] (\u) to (\v) ; 
                \draw [myEdgeZigzag] (\u) to node[midway,\pos] {\w} (\v) ; 
            } 
            \foreach \u \v \w \pos in {a/e/2.5/left,b/f/3.5/right}{
                \draw [MyOrange!30, line width = 5pt] (\u) to (\v) ; 
                \draw [myEdge] (\u) to node[midway,\pos] {\w} (\v) ; 
            }
        \end{tikzpicture} 
        \caption{End of Phase 1.}
        \label{fig:example_IM-All_phase1Completed}
    \end{subfigure}%
    \hfill
    \begin{subfigure}[b]{0.33\textwidth}
        \centering
        \begin{tikzpicture}[scale=0.9,
            myEdge/.style={line width = 1.5pt},
            myEdgeZigzag/.style={line width = 1.5pt, decoration = {zigzag, segment length = 4pt, amplitude = 1pt},decorate},
            state/.style={circle,  minimum size=1em, draw, line width = 1.2pt}]   
            \node (a) at (0,2) [state] {$a$};
            \node (b) at (4,2) [state] {$b$};
            \node (c) at (1,1) [state] {$c$};
            \node (d) at (3,1) [state] {$d$};
            \node (e) at (0,0) [state] {$e$};
            \node (f) at (4,0) [state] {$f$}; 
            \draw [myEdge] (a) to node[midway, above] {3.5} (b) ; 
            \draw [myEdge] (e) to node[midway, below] {2.5} (f) ;
            \foreach \u \v \w \pos in { a/c/3.5/right, c/d/9.5/above, d/b/4.5/left,  e/c/3.5/right, d/f/4.5/left, a/e/3.5/left, b/f/4.5/right}{
                \draw [MyOrange!30, line width = 5pt] (\u) to (\v) ; 
                \draw [myEdge] (\u) to node[midway,\pos] {\w} (\v) ; 
            } 
        \end{tikzpicture}
        \caption{End of Phase 2.}
        \label{fig:example_IM-All_end}
    \end{subfigure}%
    \caption{Illustration of Algorithm~\ref{alg:IM-All} on the graphic matroid in Figure~\ref{fig:example_IM-All_original}; see Section~\ref{sec:preliminaries} for a recall on connected components for the graphic matroid. The set $S_0 = \{ac, ae, ce, cd, db, df, bf\}$ is highlighted in orange. The components of $M|S_0$ are $K_1 = \{ac, ae, ce\}$, $K_2 = \{cd\}$, and $K_3 = \{db, df, bf\}$. In Phase 1, we get $m_1 = 0.5$, $m_2 = 6$, $m_3 = 3.5$, yielding $\varrho = 2.5$ and shifts $2$, $2.5$, and $0$, respectively. The resulting weights $w_H$ and basis $B_0 =\{ac, ce, cd, db, df\}$ are shown in Figure~\ref{fig:example_IM-All_phase1Completed}. Phase 2 yields $\Delta = 1$, and the final weight $w^*$ with optimal value $\|w-w^*\|_\infty=3.5$ is shown in Figure~\ref{fig:example_IM-All_end}.
    } 
    \label{fig:example_IM-All}
\end{figure}

\begin{remark}\label{rem:IM-All_integral}
    For the \textsc{Integral-IM-All} problem under the $\ell_\infty$-norm, where $w\in\bZ^S$ and $w^*$ is required to be integer-valued, an analogous approach works: the optimum value is $\lceil\delta^*\rceil$, where $\delta^*$ is the optimum value of the corresponding real-valued instance.
\end{remark}

\subsubsection{Algorithm for \textsc{IM-Only} under \texorpdfstring{$\ell_\infty$}{l-infty}-norm}
\label{sec:algonly}

We show that \textsc{IM-Only} can be solved by first computing an optimal weight function for the corresponding \textsc{Integral-IM-All} instance (see Remark~\ref{rem:IM-All_integral}), and then slightly modifying it, if needed, so that the bases in $S_0$ are exactly the set of maximizers. 
See an illustration in Figure~\ref{fig:example_IM-Only-infty}.

\begin{thm}\label{thm:IM-Only_algorithm}
The integral \textsc{IM-Only} problem under the $\ell_\infty$-norm can be solved using $O(n\log n)$ elementary operations and $O(nr)$ independence-oracle calls. More precisely, let $w^*_{\textsc{All}}$ be an optimal integral solution of the corresponding \textsc{IM-All} instance satisfying Lemma~\ref{lem:hm}. If $w^*_{\textsc{All}}$ is feasible for \textsc{IM-Only}, then it is optimal. Otherwise, $w^*_{\textsc{All}}+\chi_{S_0}-\chi_{S\setminus S_0}$ is an optimal solution for \textsc{IM-Only}.
\end{thm}

\begin{proof}
    Let $w^*_{\textsc{Only}}$ be an optimal weight function for the
    \textsc{IM-Only} instance, and let
    $\delta_{\textsc{Only}}=\|w-w^*_{\textsc{Only}}\|_\infty$. Since
    $w^*_{\textsc{Only}}$ is also feasible for \textsc{Integral-IM-All}, we may assume
    that it satisfies conditions~\ref{prop:a}--\ref{prop:c} of
    Lemma~\ref{lem:hm}. Let
    $\delta_{\textsc{All}}=\|w-w^*_{\textsc{All}}\|_\infty$.

    Suppose first that $\delta_{\textsc{Only}}=\delta_{\textsc{All}}$. The
    conditions of Lemma~\ref{lem:hm}, together with the fact that feasible
    solutions are constant on the components of $M|S_0$, determine the weights
    uniquely from the value of $\delta$. Hence
    $w^*_{\textsc{Only}}=w^*_{\textsc{All}}$. Therefore
    $w^*_{\textsc{All}}$ is feasible for \textsc{IM-Only}, and it is optimal.

    We may now assume that
    $\delta_{\textsc{Only}}>\delta_{\textsc{All}}$. Since the weights are
    integer-valued, this gives
    \[
    \delta_{\textsc{Only}}\ge \delta_{\textsc{All}}+1.
    \]
    Define
    \[
    w'(s)=
    \begin{cases}
    w^*_{\textsc{All}}(s)+1, & \text{if } s\in S_0,\\
    w^*_{\textsc{All}}(s)-1, & \text{if } s\in S\setminus S_0.
    \end{cases}
    \]

    Every basis contained in $S_0$ has maximum $w^*_{\textsc{All}}$-weight, and
    all such bases gain the same amount under $w'$. Hence they still have the
    same $w'$-weight. Let $B_0\subseteq S_0$ be a basis, and let $B$ be a basis
    with $B\not\subseteq S_0$. Then
    \[
    w'(B_0)=w^*_{\textsc{All}}(B_0)+|B_0|.
    \]
    Also,
    \[
    w'(B)=w^*_{\textsc{All}}(B)+|B\cap S_0|-|B\setminus S_0|.
    \]
    Since $B\not\subseteq S_0$ and $|B|=|B_0|$, we have
    $|B\setminus S_0|\ge 1$, and therefore
    \[
    |B\cap S_0|-|B\setminus S_0|\le |B_0|-2.
    \]
    Thus
    \[
    w'(B)
    \le w^*_{\textsc{All}}(B)+|B_0|-2
    < w^*_{\textsc{All}}(B_0)+|B_0|
    = w'(B_0).
    \]
    Therefore the bases contained in $S_0$ are exactly the maximum-weight bases
    under $w'$, so $w'$ is feasible for \textsc{IM-Only}.

    Finally, since $w^*_{\textsc{All}}$ satisfies
    conditions~\ref{prop:a}--\ref{prop:c} of Lemma~\ref{lem:hm}, the change
    from $w^*_{\textsc{All}}$ to $w'$ increases the $\ell_\infty$ distance from
    $w$ by exactly $1$. Hence
    \[
    \|w-w'\|_\infty=\delta_{\textsc{All}}+1\le \delta_{\textsc{Only}}.
    \]
    Since $w'$ is feasible for \textsc{IM-Only}, it is optimal.

    \textit{Running time.} By Theorem~\ref{thm:IM-All_algorithm}, an (integral) solution $w^*_{\textsc{All}}$ of \textsc{IM-All} can be obtained using $O(n\log n)$ elementary operations and $O(nr)$ independence-oracle calls.
    By Remark~\ref{rem:feas_IM-ONLY}, we can verify feasibility of $w^*_{\textsc{All}}$ by running at most three rounds of the Greedy algorithm, which takes $O(n\log n)$ operations and $O(n)$ independence-oracle calls. If $w^*_{\textsc{All}}$ is infeasible for \textsc{IM-Only}, we construct $w'=w^*_{\textsc{All}}+\chi_{S_0}-\chi_{S\setminus S_0}$ in $O(n)$ time.
\end{proof}

\begin{figure}[ht!]
\centering
        \begin{tikzpicture}[scale=0.9,
            myEdge/.style={line width = 1.5pt},
            myEdgeZigzag/.style={line width = 1.5pt, decoration = {zigzag, segment length = 4pt, amplitude = 1pt},decorate},
            state/.style={circle,  minimum size=1em, draw, line width = 1.2pt}]   
            \node (a) at (0,2) [state] {$a$};
            \node (b) at (4,2) [state] {$b$};
            \node (c) at (1,1) [state] {$c$};
            \node (d) at (3,1) [state] {$d$};
            \node (e) at (0,0) [state] {$e$};
            \node (f) at (4,0) [state] {$f$}; 
            \draw [myEdge] (a) to node[midway, above] {3} (b) ; 
            \draw [myEdge] (e) to node[midway, below] {2} (f) ;
            \foreach \u \v \w \pos in { a/c/4/right, c/d/10/above, d/b/5/left,  e/c/4/right, d/f/5/left, a/e/4/left, b/f/5/right}{
                \draw [MyOrange!30, line width = 5pt] (\u) to (\v) ; 
                \draw [myEdge] (\u) to node[midway,\pos] {\w} (\v) ; 
            } 
        \end{tikzpicture}
    \caption{Illustration of the algorithm for \textsc{IM-Only} under the $\ell_\infty$-norm (Theorem~\ref{thm:IM-Only_algorithm}) applied to the graphic matroid in Figure~\ref{fig:example_IM-All_original}. Applying Remark~\ref{rem:IM-All_integral} to the solution shown in Figure~\ref{fig:example_IM-All_end}, we obtain the weight function depicted in this figure, which is optimal and has value $\|w-w^*\|_\infty=4$.
    } 
    \label{fig:example_IM-Only-infty}
\end{figure}
\subsection{IM-All and IM-Only under \texorpdfstring{$\ell_1$}{l-1}-norm}
We provide purely combinatorial algorithms for both problems by exploiting the convex formulation given in Lemma~\ref{lem:separable_convex_prog}. In particular, we show that the two problems are closely related and can be solved using essentially the same algorithm, requiring only a minor modification.

For the $\ell_1$ objective, Lemma~\ref{lem:separable_convex_prog} gives the
following formulation of \textsc{IM-All}. Recall that, for each $f\in S\setminus S_0$, the set $K(f)\subseteq[\ell]$ consists of the indices $i$ such that $(C(B_0,f)\setminus\{f\})\cap K_i\neq\emptyset$.
\begin{align*}
\min\quad
& \sum_{i=1}^{\ell}\sum_{e\in K_i}|w(e)-x_i|
  + \sum_{f\in S\setminus S_0}|w(f)-y_f| \\
\textnormal{s.t.}\quad
& x_i\ge y_f
  \qquad \forall f\in S\setminus S_0,\ \forall i\in K(f).
\end{align*}
If $w$ is integral, then, by Remark~\ref{obs:im_only_strict}, the corresponding
formulation for \textsc{IM-Only} is obtained by replacing the inequalities
$x_i\ge y_f$ with $x_i\ge y_f+1$.

For $C\in\{0,1\}$, put $z_f=y_f+C$. Here $C=0$ corresponds to \textsc{IM-All},
and $C=1$ corresponds to the fractional relaxation of \textsc{IM-Only}. The
formulation becomes
\begin{align}
\min\quad
& \sum_{i=1}^{\ell}\sum_{e\in K_i}|w(e)-x_i|
  + \sum_{f\in S\setminus S_0}|w(f)+C-z_f| \notag \\
\textnormal{s.t.}\quad
& z_f\le x_i
  \qquad \forall f\in S\setminus S_0,\ \forall i\in K(f).
\label{eq:l1_shifted_formulation}
\end{align}
After solving \eqref{eq:l1_shifted_formulation}, we recover $y_f=z_f-C$. For each value of $C$, define
\[
A_C=
\{w(e):e\in S_0\}
\cup
\{w(f)+C\colon f\in S\setminus S_0\}.
\]

\begin{lemma}\label{lem:l1_values_in_AC}
Let $w\in\mathbb R^S$ and let $C\in\{0,1\}$. Then there exists an
optimal solution to \eqref{eq:l1_shifted_formulation} satisfying
\begin{enumerate}[label=(\alph*)]\itemsep0em
    \item $x_i\in A_C$ for every $i\in[\ell]$,
    \item $z_f\in A_C$ for every $f\in S\setminus S_0$.
\end{enumerate}
Moreover, if $w$ is integer-valued, then
\eqref{eq:l1_shifted_formulation} has an integral optimal solution.
\end{lemma}

\begin{proof}
Write $A_C=\{a_1,\ldots,a_m\}$, with $a_1<\cdots<a_m$.

We first reduce to solutions whose variables lie in $[a_1,a_m]$. If a variable
has value below $a_1$, replace it by $a_1$; if it has value above $a_m$, replace
it by $a_m$. This preserves all inequalities $z_f\le x_i$. It also decreases
the objective, since every term has the form $|b-t|$ with $b\in A_C$.

Let $(x^*,z^*)$ be an optimal solution with all variables in $[a_1,a_m]$. We
add bounds around the values of this solution. If the value $v^*$ of a variable
$v$ belongs to $A_C$, add the equality $v=v^*$. Otherwise, let
$a_r<v^*<a_{r+1}$ and add the bounds
\[
a_r\le v\le a_{r+1}.
\]
We do this for every variable $v$, where $v$ is either some $x_i$ or some
$z_f$.

On the resulting feasible region, the objective is linear. Indeed, as functions
of $t$, both
\[
\sum_{e\in K_i}|w(e)-t|
\qquad\text{and}\qquad
|w(f)+C-t|
\]
are linear on every interval between two consecutive values of $A_C$. Thus the
restricted problem is a linear program. It contains $(x^*,z^*)$, and its
objective agrees with the original objective on the restricted feasible region.
Hence it has the same optimum value as \eqref{eq:l1_shifted_formulation}.

The constraints of this linear program are the inequalities $z_f\le x_i$ and
the bounds added above. The matrix is totally unimodular: each row of the form
$z_f\le x_i$ has one $+1$ and one $-1$, and adding variable bounds preserves
total unimodularity. Let $(\bar x,\bar z)$ be an optimal extreme point of the
restricted linear program.

We now show that the values of this extreme point belong to $A_C$. At an
extreme point, the tight constraints determine all variables. A tight
inequality $z_f\le x_i$ only gives the equality $z_f=x_i$; it does not fix a
number. Numbers can only come from tight bounds, and all bounds added above
have values in $A_C$. Therefore every variable value in $(\bar x,\bar z)$
belongs to $A_C$.

If $w$ is integer-valued, then all bounds are integral. Total unimodularity
implies that $(\bar x,\bar z)$ is integral. Since $C\in\{0,1\}$, the recovered
values $y_f=\bar z_f-C$ are also integral. This proves the lemma.

\end{proof}
\begin{remark}
    The integrality of the optimal solutions established in Lemma~\ref{lem:l1_values_in_AC} relies on the $\ell_1$-norm objective. This property does not hold in general for the $\ell_\infty$-norm. For instance, consider the \textsc{IM-All} problem on the uniform matroid of rank~$1$ over $S=\{s_0,s_1\}$, with target subset $S_0=\{s_0\}$ and initial weights $w(s_0)=0$ and $w(s_1)=1$. The optimal value under the $\ell_\infty$-norm is $1/2$, uniquely attained by the strictly fractional weight function $w^*(s_0)=w^*(s_1)=1/2$.
\end{remark}

Instead of solving \eqref{eq:l1_shifted_formulation} directly, we present Algorithm~\ref{alg:IM-All-Only-L1}, which solves the problem via a single minimum $s$--$t$ cut computation. The high-level idea of the algorithm is as follows. First, we determine the connected components of the restriction $M|S_0$ and construct the graph from Lemma~\ref{lem:l1_min_closure}. We compute a minimum $s$--$t$ cut and reconstruct the solution by setting $w^*(e)\coloneqq x_i$ for every $e\in K_i$, $i\in[\ell]$, and $w^*(f)\coloneqq z_f-C$ for every $f\in S\setminus S_0$. An illustration of the algorithm is shown in Figure~\ref{fig:example_IM-All-L1}.

\begin{algorithm}[ht!]
    \caption{Combinatorial algorithm for \textsc{IM-All} and \textsc{IM-Only} under the $\ell_1$-norm}\label{alg:IM-All-Only-L1}
    \DontPrintSemicolon
    \KwIn{An instance $(M=(S,\mathcal{I}),S_0,w,\|\cdot\|_1)$ of \textsc{IM-All} or \textsc{IM-Only}.}
    \KwOut{An optimal weight function $w^*$.}
    Set $C\coloneqq 0$ for \textsc{IM-All} and $C\coloneqq 1$ for \textsc{IM-Only}.\;    
    \lIf{$M$ has no basis contained in $S_0$}{
        Set $w^* \coloneqq w$.
    }
    \Else{
        \tcp*[l]{Phase 1: Decomposition}
        Determine the connected components $K_1,\dots,K_\ell$ of $M|_{S_0}$.\;
    
        Construct the graph $G$ corresponding to
        \eqref{eq:l1_shifted_formulation} with parameter $C$,
        as described in Lemma~\ref{lem:l1_min_closure}.\;
    
        \tcp*[l]{Phase 2: Minimum cut}
        Construct $\widehat{G}$ as described in Lemma~\ref{lem:l1_min_closure}.\;
        Compute a minimum $s$-$t$ cut in $\widehat{G}$.\;
    
        Recover the optimal values of the variables
        $x_1,\dots,x_\ell$ and $z_f$ for all
        $f\in S\setminus S_0$.\;
    
        \tcp*[l]{Phase 3: Reconstruction}
        Set $w^*(e)\coloneqq x_i$ for every $e\in K_i$, $i\in[\ell]$.\\
        Set $w^*(f)\coloneqq z_f-C$ for every
        $f\in S\setminus S_0$. 
    }
    \Return{$w^*$.}\;
\end{algorithm}

\begin{figure}[htp!]
\centering
    \begin{subfigure}[b]{0.33\textwidth}
        \centering
        \begin{tikzpicture}[scale=0.9,
            myEdge/.style={line width = 1.5pt},     
            state/.style={circle,  minimum size=1em, draw, line width = 1.2pt}]   
            \node (a) at (0,2) [state] {$a$};
            \node (b) at (4,2) [state] {$b$};
            \node (c) at (1,1) [state] {$c$};
            \node (d) at (3,1) [state] {$d$};
            \node (e) at (0,0) [state] {$e$};
            \node (f) at (4,0) [state] {$f$}; 
            \draw [myEdge] (a) to node[midway, above] {7} (b) ; 
            \draw [myEdge] (e) to node[midway, below] {6} (f) ;
            \foreach \u \v \w \pos in {a/c/{\rlap{0}\phantom{3.5..}}/right, c/d/6/above, d/b/6/left, a/e/0/left, e/c/1/right, d/f/3/left, b/f/{\rlap{1}\phantom{3.5..}}/right}{
                \draw [MyOrange!30, line width = 5pt] (\u) to (\v) ; 
                \draw [myEdge] (\u) to node[midway,\pos] {\w} (\v) ; 
            }     
        \end{tikzpicture} 
        \caption{Original instance.}
        \label{fig:example_IM-All_original-L1}
    \end{subfigure}%
    \hfill
    \begin{subfigure}[b]{0.6\textwidth}
        \centering
        \resizebox{\textwidth}{!}{
        \begin{tikzpicture}[
          every node/.style={font=\small},
          var/.style={circle,draw,thick,minimum size=8mm,inner sep=0pt},
          colx1/.style={var,draw=blue!60!black,fill=blue!8},
          colx2/.style={var,draw=teal!60!black,fill=teal!8},
          colx3/.style={var,draw=orange!70!black,fill=orange!10},
          colz/.style={var,draw=violet!60!black,fill=violet!8},
          colw/.style={var,draw=purple!60!black,fill=purple!8},
          crossArc/.style={-{Stealth[length=2mm]},black,thin},
          myArc/.style={-{Stealth[length=3mm]},thick},
          ]
        \def\cxone{0}
        \def\cxtwo{2.4}
        \def\cxthree{4.8}
        \def\czab{9.2}
        \def\czef{11.8}
        \def\cs{-3}
        \def\ct{14.8}
        \def\yfive{6}
        \def\yfour{4}
        \def\ythree{2}
        \def\ytwo{0}
        \def\ymid{3}
        \node[colx1] (x15) at (\cxone,\yfive)  {$(x_1,5)$};
        \node[left = 0.03ex of x15] {\scriptsize $\Delta=3$};
        \node[colx1] (x14) at (\cxone,\yfour)  {$(x_1,4)$};
        \node[left = 0.03ex of x14] {\scriptsize $\Delta=9$};
        \node[colx1] (x13) at (\cxone,\ythree) {$(x_1,3)$};
        \node[left = 0.03ex of x13] {\scriptsize $\Delta=6$};
        \node[colx1] (x12) at (\cxone,\ytwo)   {$(x_1,2)$};
        \node[left = 0.03ex of x12] {\scriptsize $\Delta=1$};
        \node[colx2] (x25) at (\cxtwo,\yfive)  {$(x_2,5)$};
        \node[left = 0.03ex of x25] {\scriptsize $\Delta=1$};
        \node[colx2] (x24) at (\cxtwo,\yfour)  {$(x_2,4)$};
        \node[left = 0.03ex of x24] {\scriptsize $\Delta=-3$};
        \node[colx2] (x23) at (\cxtwo,\ythree) {$(x_2,3)$};
        \node[left = 0.03ex of x23] {\scriptsize $\Delta=-2$};
        \node[colx2] (x22) at (\cxtwo,\ytwo)   {$(x_2,2)$};
        \node[left = 0.03ex of x22] {\scriptsize $\Delta=-1$};
        \node[colx3] (x35) at (\cxthree,\yfive)  {$(x_3,5)$};
        \node[left = 0.03ex of x35] {\scriptsize $\Delta=3$};
        \node[colx3] (x34) at (\cxthree,\yfour)  {$(x_3,4)$};
        \node[left = 0.03ex of x34] {\scriptsize $\Delta=3$};
        \node[colx3] (x33) at (\cxthree,\ythree) {$(x_3,3)$};
        \node[left = 0.03ex of x33] {\scriptsize $\Delta=-2$};
        \node[colx3] (x32) at (\cxthree,\ytwo)   {$(x_3,2)$};
        \node[left = 0.03ex of x32] {\scriptsize $\Delta=-3$};
        
        \node[colz] (z5) at (\czab,\yfive)  {$(z_{ab},5)$};
        \node[right  = 0.03ex of z5] {\scriptsize $\Delta=-1$};
        
        \node[colz] (z4) at (\czab,\yfour)  {$(z_{ab},4)$};
        \node[right  = 0.03ex of z4] {\scriptsize $\Delta=-3$};
        
        \node[colz] (z3) at (\czab,\ythree) {$(z_{ab},3)$};
        \node[right  = 0.03ex of z3] {\scriptsize $\Delta=-2$};
        
        \node[colz] (z2) at (\czab,\ytwo)   {$(z_{ab},2)$};
        \node[right  = 0.03ex of z2] {\scriptsize $\Delta=-1$};
        
        \node[colw] (w5) at (\czef,\yfive)  {$(z_{ef},5)$};
        \node[right =  0.03ex of w5] {\scriptsize $\Delta=1$};
        
        \node[colw] (w4) at (\czef,\yfour)  {$(z_{ef},4)$};
        \node[right =  0.03ex of w4] {\scriptsize $\Delta=-3$};
        
        \node[colw] (w3) at (\czef,\ythree) {$(z_{ef},3)$};
        \node[right =  0.03ex of w3] {\scriptsize $\Delta=-2$};
        
        \node[colw] (w2) at (\czef,\ytwo)   {$(z_{ef},2)$};
        \node[right =  0.03ex of w2] {\scriptsize $\Delta=-1$};
        
        \foreach \col in {x1,x2,x3,z,w}{
          \draw[myArc] (\col5)--(\col4);
          \draw[myArc] (\col4)--(\col3);
          \draw[myArc] (\col3)--(\col2);
        }        
        \foreach \row in {5,4,3,2}{
          \draw[crossArc] (z\row) to[bend left=18]  (x1\row);
          \draw[crossArc] (z\row) to[bend left=14]  (x2\row);
          \draw[crossArc] (z\row) to[bend left=10]  (x3\row);
          \draw[crossArc] (w\row) to[bend right=22] (x1\row);
          \draw[crossArc] (w\row) to[bend right=18] (x2\row);
          \draw[crossArc] (w\row) to[bend right=14] (x3\row);
        }
        \end{tikzpicture}
        }
        \caption{Phase 1: Construction of $G$.}
        \label{fig:phase1-G}
    \end{subfigure}%
    \hfill
    \begin{subfigure}[b]{0.6\textwidth}
        \centering
        \resizebox{\textwidth}{!}{
        \begin{tikzpicture}[
          every node/.style={font=\small},
          var/.style={circle,draw,thick,minimum size=8mm,inner sep=0pt},
          st/.style={var,draw=black,font=\bfseries},
          colx1/.style={var,draw=blue!60!black,fill=blue!8},
          colx2/.style={var,draw=teal!60!black,fill=teal!8},
          colx3/.style={var,draw=orange!70!black,fill=orange!10},
          colz/.style={var,draw=violet!60!black,fill=violet!8},
          colw/.style={var,draw=purple!60!black,fill=purple!8},
          capac/.style={midway,fill=white,inner sep=1pt,font=\scriptsize},
          Mlabel/.style={midway,fill=white,inner sep=1pt,font=\scriptsize\itshape},
          crossArc/.style={-{Stealth[length=2mm]},gray,thin},
          myArc/.style={-{Stealth[length=3mm]},thick},
          ]        
        \def\cxone{0}
        \def\cxtwo{2.4}
        \def\cxthree{4.8}
        \def\czab{9.2}
        \def\czef{11.8}
        \def\cs{-3}
        \def\ct{14.8}
        \def\yfive{6}
        \def\yfour{4}
        \def\ythree{2}
        \def\ytwo{0}
        \def\ymid{3}        
        \node[st] (s) at (\cs,\ymid) {$s$};
        \node[st] (t) at (\ct,\ymid) {$t$};
        \foreach \row/\y in {5/\yfive,4/\yfour,3/\ythree,2/\ytwo}{
          \node[colx1] (x1\row) at (\cxone,\y)   {$(x_1,\row)$};
          \node[colx2] (x2\row) at (\cxtwo,\y)   {$(x_2,\row)$};
          \node[colx3] (x3\row) at (\cxthree,\y) {$(x_3,\row)$};
          \node[colz]  (z\row)  at (\czab,\y)    {$(z_{ab},\row)$};
          \node[colw]  (w\row)  at (\czef,\y)    {$(z_{ef},\row)$};
        }
        \foreach \col in {x1,x2,x3,z,w}{
          \draw[myArc] (\col5)-- node[capac, pos=.1] {$M$}(\col4);
          \draw[myArc] (\col4)-- node[capac, pos=.1] {$M$}(\col3);
          \draw[myArc] (\col3)-- node[capac, pos=.1] {$M$}(\col2);
        }
        \foreach \row in {5,4,3,2}{
          \draw[crossArc] (z\row) to[bend left=18]  (x1\row);
          \draw[crossArc] (z\row) to[bend left=14]  (x2\row);
          \draw[crossArc] (z\row) to[bend left=10]  (x3\row);
          \draw[crossArc] (w\row) to[bend right=22] (x1\row);
          \draw[crossArc] (w\row) to[bend right=18] (x2\row);
          \draw[crossArc] (w\row) to[bend right=14] (x3\row);
        }
        \draw[myArc] (s) to[bend right=70] node[capac, pos=.9] {1} (x22);
        \draw[myArc] (s) to[bend left=8] node[capac, pos=.9] {2} (x23);
        \draw[myArc] (s) to[bend left=35] node[capac, pos=.9] {3} (x24);
        \draw[myArc] (s) to[bend right=75] node[capac, pos=.9] {3} (x32);
        \draw[myArc] (s) to[bend right=40] node[capac, pos=.9] {2} (x33);
        \draw[myArc] (s) to[bend right=60] node[capac, pos=.9] {1} (z2);
        \draw[myArc] (s) to[bend right=30] node[capac, pos=.9] {2} (z3);
        \draw[myArc] (s) to[bend left=25] node[capac, pos=.9] {3} (z4);
        \draw[myArc] (s) to[bend left=60] node[capac, pos=.9] {1} (z5);
        \draw[myArc] (s) to[bend right=50] node[capac, pos=.9] {1} (w2);
        \draw[myArc] (s) to[bend left=8] node[capac, pos=.9] {2} (w3);
        \draw[myArc] (s) to[bend left=70] node[capac, pos=.9] {3} (w4);
        \draw[myArc] (x12.south east) to[bend right=70] node[capac, pos=.20] {1} (t);
        \draw[myArc] (x13) to[bend right=80] node[capac, pos=.07] {6} (t);
        \draw[myArc] (x14) to[bend left=80] node[capac, pos=.1] {9} (t);
        \draw[myArc] (x15) to[bend left=60] node[capac, pos=.05] {3} (t);
        \draw[myArc] (x25) to[bend left=60] node[capac, pos=.1] {1} (t);
        \draw[myArc] (x34) to[bend right=10] node[capac, pos=.05] {3} (t);
        \draw[myArc] (x35) to[bend left=60] node[capac, pos=.05] {3} (t);
        \draw[myArc] (w5) to[bend right=0] node[capac, pos=.1] {1} (t);
        \end{tikzpicture}
        }
        \caption{Phase 2: Construction of $\widehat{G}$.}
        \label{fig:example_IM-All_end-L1}
    \end{subfigure}
    \hfill
    \begin{subfigure}[b]{0.33\textwidth}
        \centering
        \begin{tikzpicture}[scale=0.9,
            myEdge/.style={line width = 1.5pt},     
            state/.style={circle,  minimum size=1em, draw, line width = 1.2pt}]   
            \node (a) at (0,2) [state] {$a$};
            \node (b) at (4,2) [state] {$b$};
            \node (c) at (1,1) [state] {$c$};
            \node (d) at (3,1) [state] {$d$};
            \node (e) at (0,0) [state] {$e$};
            \node (f) at (4,0) [state] {$f$}; 
            \draw [myEdge] (a) to node[midway, above] {0} (b) ; 
            \draw [myEdge] (e) to node[midway, below] {0} (f) ;
            \foreach \u \v \w \pos in {a/c/{\rlap{0}\phantom{3.5..}}/right, c/d/6/above, d/b/3/left, a/e/0/left, e/c/0/right, d/f/3/left, b/f/{\rlap{3}\phantom{3.5..}}/right}{
                \draw [MyOrange!30, line width = 5pt] (\u) to (\v) ; 
                \draw [myEdge] (\u) to node[midway,\pos] {\w} (\v) ; 
            }     
        \end{tikzpicture} 
        \caption{End of Phase 3.}
        \label{fig:example_IM-All_final-L1}
    \end{subfigure}%
    \caption{Illustration of Algorithm~\ref{alg:IM-All-Only-L1} on the graphic matroid in Figure~\ref{fig:example_IM-All_original-L1} for $C=0$; see Section~\ref{sec:preliminaries} for a recall on connected components for the graphic matroid. The set $S_0 = \{ac, ae, ce, cd, db, df, bf\}$ is highlighted in orange. 
    In Phase 1, we find that the components of $M|S_0$ are $K_1 = \{ac, ae, ce\}$, $K_2 = \{cd\}$, and $K_3 = \{db, df, bf\}$. $A_0=  \{ 0,1,3,6 \} \cup \{6, 7 \}$ and hence $a_1=0,a_2=1,a_3=3,a_4=6,a_5=7$. As there are 3 connected components in the graphic matroid and 2 elements in $S\setminus S_0$, the variables are $x_1,x_2, x_3, z_{ab},z_{ef}$. Since $C(B_0,ab)-ab = \{ ac,cd, db \}$ and $C(B_0,ef)-ef = \{ ec,cd,df \}$, it follows that $K(ab)=\{1,2,3\}$ and $K(ef)=\{1,2,3\}$, yielding the constraints $z_{ab}\leq x_i$ and $z_{ef}\leq x_i$ for $i=1,2,3$. We calculate the node weights: for variable $x_1$, the weights of $K_1$ are $\{0,0,1\}$ and thus $\phi_{x_1}(t)= 2|t|+|1-t|$; whose successive differences yield the node weights $\Delta(x_1,2)=\phi_{x_1}(2) - \phi_{x_1}(1)=1, \Delta(x_1,3)= 6, \Delta(x_1,4)= 9, \Delta(x_1,5)= 3$. Similarly for the rest of the nodes, as shown in Figure~\ref{fig:phase1-G}.
    In Phase 2, we obtain the arc-capacitated graph $\widehat{G}$, where $M>51$ and all gray arcs also have capacity $M$. A minimum $s$-$t$ cut is obtained by placing the nodes $U^{\operatorname{cut}}=\{(x_2,2),(x_2,3),(x_2,4),(x_3,2),(x_3,3)\}$ together with $s$ on the source side, and all remaining nodes together with $t$ on the sink side. 
    The set $U^{\operatorname{cut}}$ is precisely the minimum-weight closed set of $G$.    
    The cut has value $13$, corresponding to a closure weight of $-11$. 
    Therefore, the recovered solution is $(x_1,x_2,x_3,z_{ab},z_{ef})=(a_1,a_4,a_3,a_1,a_1)$. 
    Finally, Phase 3 constructs the optimal weight function for the \textsc{IM-All} instance under the $\ell_1$-norm with $\| w-w^* \| = 19$, shown in Figure~\ref{fig:example_IM-All_final-L1}.
    } 
    \label{fig:example_IM-All-L1}
\end{figure}

The formulation above is a separable convex minimization problem with order constraints of the form $z_f\le x_i$. For integral weights, an optimum can be chosen from the finite set of values appearing in the input. This allows us to replace the convex program by a minimum-weight closure problem, and hence by a single minimum $s$-$t$ cut computation. Thus the $\ell_1$ algorithms for \textsc{IM-All} and \textsc{IM-Only} remain fully combinatorial.

Before proving the correctness of the algorithm, we show that we can indeed use the $s$-$t$ minimum cut problem.
\begin{lemma}\label{lem:l1_min_closure}
Let $C\in\{0,1\}$. Then formulation \eqref{eq:l1_shifted_formulation} can be solved by a single $s$-$t$ minimum cut computation on a directed graph with $O(n^2)$ nodes and $O(n^3)$ arcs. Moreover, such a graph can be constructed in $O(n^3)$ time.
\end{lemma}

\begin{proof}
    By Lemma~\ref{lem:l1_values_in_AC}, it is enough to consider solutions whose
    variables take values in $A_C=\{a_1,\ldots,a_m\}$, with $a_1<\cdots<a_m$.
    Since $A_C$ is defined from the weights of the elements of $S$, we have
    $m\le |S|$.
    
    We use the threshold-variable idea from Ahuja, Hochbaum, and
    Orlin~\cite{AhujaHochbaumOrlin2004}, specialized to the values in $A_C$.
    Since our constraints are only $z_f\le x_i$, the construction gives a
    minimum-weight closure instance on a graph $G$; we later reduce this graph to a minimum $s$-$t$ cut instance on a graph $\widehat{G}$. We give the details.
    
    For every variable $v$ in \eqref{eq:l1_shifted_formulation}, and every
    $j=2,\ldots,m$, create a node $(v,j)$. The node $(v,j)$ represents the condition
    $v\ge a_j$. There are $O(|S|)$ variables and $m\le |S|$, so this gives $O(|S|^2)$
    nodes.
    
    We add two types of arcs. First, for every variable $v$ and every
    $j=2,\ldots,m-1$, add
    \[
    (v,j+1)\to (v,j).
    \]
    This enforces that $v\ge a_{j+1}$ implies $v\ge a_j$. Second, for every
    inequality $z_f\le x_i$ and every $j=2,\ldots,m$, add
    \[
    (z_f,j)\to (x_i,j).
    \]
    This enforces that $z_f\ge a_j$ implies $x_i\ge a_j$.
    
    We now assign the node weights. For a variable $v$, define
    \[
    \phi_v(t)=
    \begin{cases}
    \sum_{e\in K_i}|w(e)-t|, & \text{if } v=x_i,\\
    |w(f)+C-t|, & \text{if } v=z_f.
    \end{cases}
    \]
    The node $(v,j)$ receives weight
    \[
    \Delta_{v,j}=\phi_v(a_j)-\phi_v(a_{j-1}).
    \]

    This concludes the description of the graph $G$. 
        
    A closed set selects, for each variable $v$, either no node or a prefix $(v,2),(v,3),\ldots,(v,j)$.
    
    We interpret no selected node as $v=a_1$, and this selected prefix as $v=a_j$.
    The weight selected in the column of $v$ is
    \[
    \sum_{h=2}^j \Delta_{v,h}
    =
    \phi_v(a_j)-\phi_v(a_1).
    \]
    Thus the total weight of the closed set is the objective value of the encoded
    solution, minus the constant $\sum_v\phi_v(a_1)$.
    
    The arcs encode exactly the constraints. The vertical arcs force each column to
    select a prefix. The arcs $(z_f,j)\to(x_i,j)$ force $z_f\le x_i$. Conversely,
    any feasible assignment with values in $A_C$ selects the nodes $(v,j)$ with
    $v\ge a_j$, and these nodes form a closed set. Hence feasible assignments with
    values in $A_C$ correspond exactly to closed sets.
    
    It remains to solve the resulting minimum-weight closure problem on $G$.
    This graph is node-weighted, where each node $(v,j)$ has weight $\Delta_{v,j}$. We reduce it to a minimum $s$-$t$ cut instance using the standard construction. We begin by constructing a graph $\widehat{G}$ as a copy of $G$. Introduce a source vertex $s$ and a sink vertex $t$.
    For each node $u$ with weight
    $\Delta_u<0$, add an arc $s\to u$ of capacity $-\Delta_u$. For each node $u$
    with weight $\Delta_u>0$, add an arc $u\to t$ of capacity $\Delta_u$. 
    Finally, assign every arc enforcing the closure constraints a capacity strictly larger than $\sum_u |\Delta_u|$. This concludes the construction of $\widehat{G}$.

    A finite cut leaves a closed set on the source side. Its capacity differs from
    the weight of that closed set by a constant. Hence a minimum cut gives a
    minimum-weight closed set, and therefore an optimal solution of
    \eqref{eq:l1_shifted_formulation}.

    Summarizing, let $U^{\operatorname{cut}}\cup\{s\}$ be the source side of a minimum
    $s$--$t$ cut in $\widehat{G}$. Then $U^{\operatorname{cut}}$ is the corresponding
    minimum-weight closed set. For each variable $v$, let $j$ be the largest index
    such that $(v,j)\in U^{\operatorname{cut}}$, and set $v=a_j$ (recall that
    $j\ge2$). If no node corresponding to $v$ belongs to $U^{\operatorname{cut}}$,
    set $v=a_1$. This yields an optimal solution to
    \eqref{eq:l1_shifted_formulation}. Consequently, the corresponding optimal
    weight function is obtained by assigning $w^*(e)=x_i$ for every
    $e\in K_i$, $i\in[\ell]$, and $w^*(f)=z_f-C$ for every
    $f\in S\setminus S_0$.
    
    \textit{Construction time}. The graph $G$ has $O(|S|^2)$ nodes. The vertical arcs contribute $O(|S|^2)$ arcs. There
    are at most $O(|S|^2)$ inequalities of the form $z_f\le x_i$, and each one gives
    one arc for each value of $A_C$. Since $|A_C|\le |S|$, these arcs contribute
    $O(|S|^3)$ arcs. The reduction from closure to minimum cut adds at most one arc
    incident to $s$ or $t$ per node. Thus the final graph $\widehat{G}$ has $O(|S|^2)$ nodes and $O(|S|^3)$ arcs. The same bound covers the construction time.
\end{proof}

We present the main result. When $w$ is integer-valued, Lemma~\ref{lem:l1_values_in_AC} implies that the problems admit an optimal solution with integer variable values (and hence an integer-valued optimal weight function).
\begin{thm}\label{thm:algo-IM-All-Only-L1}
    Algorithm~\ref{alg:IM-All-Only-L1} determines an optimal weight function
    for \textsc{IM-All} instances $(M=(S,\cI),S_0,w,\|\cdot\|_1)$ and for
    \textsc{IM-Only} instances of the same form with $w\in\mathbb Z^S$
    using $O(n^3 + \varphi)$ operations and $O(nr)$ independence-oracle calls,
    where $\varphi$ is the time required to compute one minimum $s$-$t$ cut on a graph with $O(n^2)$ vertices and $O(n^3)$ arcs.
\end{thm}
\begin{proof}
    The correctness of the algorithm follows from
    Lemmas~\ref{lem:separable_convex_prog},
    \ref{lem:l1_values_in_AC}, and
    \ref{lem:l1_min_closure}. Indeed,
    Lemma~\ref{lem:separable_convex_prog} gives an equivalent convex
    formulation of the problems,
    Lemma~\ref{lem:l1_values_in_AC} shows that it suffices to consider
    solutions whose variables take values in $A_C$, and
    Lemma~\ref{lem:l1_min_closure} reduces the resulting optimization problem
    to a single $s$-$t$ minimum cut computation. 
    
    \textit{Running time.}
    Detecting whether $M$ has a basis contained in $S_0$ can be done by running the greedy algorithm, giving priority to elements in $S_0$, using $O(n\log n)$ elementary operations and $O(n)$ independence-oracle calls. If such a basis exists, the same computation returns a basis $B_0\subseteq S_0$. As mentioned earlier in Section~\ref{sec:preliminaries}, the connected components of $M|S_0$ can be determined using $O(nr)$ independence-oracle calls~\cite{krogdahl1977dependence}. For each $f\in S\setminus S_0$, Remark~\ref{rem:findingFundCircuit} computes $C(B_0,f)$ using $O(r)$ independence-oracle calls, and $K(f)$ is obtained by recording the components that intersect $C(B_0,f)\setminus\{f\}$. Thus all sets $K(f)$ can be computed using $O(r|S\setminus S_0|)=O(nr)$ independence-oracle calls. By Lemma~\ref{lem:l1_min_closure}, the auxiliary graph can be constructed in $O(n^3)$ time. The optimal solution is then obtained by computing a minimum $s$--$t$ cut in this graph.
\end{proof}

\section{An auxiliary problem: \textsc{IM-Outside}}
\label{sec:relaxed-not-IM-Exists}

The negated variants require a second basic operation: making at least one maximum-weight basis lie outside $S_0$. We isolate this task as an auxiliary problem, \textsc{IM-Outside}. This problem is a relaxation of \textsc{IM-Not-Exists}: bases contained in $S_0$ are still allowed to be optimal, but at least one optimal basis must use an element outside $S_0$. This relaxation is exactly the separation step needed later for \textsc{IM-Not-Exists} and \textsc{IM-Not-Only}. Unlike \textsc{IM-Not-Exists}, it does not require strict separation, and therefore we state it over real-valued weight functions. It also arises naturally in our setting and, to the best of our knowledge, has not been studied before.

\problemdef{Inverse Matroid Outside (IM-Outside)}
    {A matroid $M=(S,\mathcal{I})$, a subset $S_0\subseteq S$ such that some basis of $M$ is not contained in $S_0$, a weight function $w\in\bR^S$, and an objective function $\|\cdot\|$ defined on $\bR^S$.}
    {Find a weight function $\widehat{w} \in \bR^S$ such that 
    \begin{enumerate}[label=(\alph*),topsep=1px,partopsep=2px]\itemsep0em
        \item \label{it:feas_cond_relaxed} there exists a basis not contained in $S_0$ that is of maximum $\widehat{w}$-weight, and
        \item \label{it:opt_cond_relaxed} $\|w-\widehat{w}\|$ is minimized.
    \end{enumerate}}

A weight function is \emph{feasible} if condition~\ref{it:feas_cond_relaxed} holds. Note that \textsc{IM-Outside} is always feasible: set $w'$ such that $w'(s)=0$ for all $s \in S_0$ and $w'(s)=1$ for all $s \in S \setminus S_0$.

\begin{remark}\label{rem:relaxedProb_feasibility}
    Given a weight function $w'$, we can verify in polynomial time whether $w'$ is feasible for \textsc{IM-Outside}. To do so, run the greedy algorithm while breaking ties in favor of elements in $S\setminus S_0$. Then $w'$ is feasible if and only if the resulting basis is not contained in $S_0$.
\end{remark}

\subsection{Structural Properties}
We first characterize the structure of an optimal solution under the $\ell_\infty$-norm.
Although the lemma follows from the correspondence between \textsc{IM-Outside} and \textsc{IM-Exists} on the dual matroid, we give a direct proof, which is independent of this reduction and more transparent.

\begin{lemma}\label{lem:relaxedProb_Nicer_special_weight_2}
    Let $\widehat{w}$ be a feasible solution for the \textsc{IM-Outside} instance $(M=(S,\cI),S_0,w,\|\cdot\|_\infty)$, and let $\delta=\|w-\widehat{w}\|_\infty$. Then there exists a feasible solution $w'$ satisfying the following:
    \begin{enumerate}[label=(\alph*)]\itemsep0em
        \item $\|w-w'\|_\infty=\delta$, and 
         \item $w'=w + \delta\cdot \left( \chi_{ S \setminus S_0} -\chi_{S_0} \right)$.
    \end{enumerate}    
     Furthermore, if $w$ is not feasible and $B_0\subseteq S_0$ is a maximum $w$-weight basis, then the optimal value is $\delta^* = \min\{ (w(e)-w(f))/2 \colon f\in S\setminus S_0,\ e\in C(B_0,f)\setminus\{f\}\}$. In particular, $\delta^*$ is half-integral if $w\in\mathbb Z^S$.
\end{lemma}
\begin{proof}
    If $w$ is feasible, the lemma holds trivially. So assume that $w$ is not feasible, i.e., every maximum $w$-weight basis of $M$ is contained in $S_0$. Let $B_0\subseteq S_0$ be a basis of maximum $w$-weight; then $B_0$ is a maximum-weight basis of $M$. By Proposition~\ref{prop:charact_optimum_basis}, for every pair $f\notin B_0$ and $e\in C(B_0,f)\setminus\{f\}$, we have $w(e)\geq w(f)$.

    Let $\delta^*=\min\{(w(e)-w(f))/2\colon f\in S\setminus S_0,\ e\in C(B_0,f)\setminus\{f\}\}$, and for $\lambda\geq 0$ write $w_\lambda\coloneqq w+\lambda\cdot(\chi_{S\setminus S_0}-\chi_{S_0})$. Let $(e',f')\in\argmin\{(w(e)-w(f))/2\colon f\in S\setminus S_0,\ e\in C(B_0,f)\setminus\{f\}\}$, so $w(e')-w(f')=2\delta^*$.
    
   We claim that $w_\lambda$ is feasible for every $\lambda\geq\delta^*$. First, for any basis $B\subseteq S_0$, all $r$ elements of $B$ lie in $S_0$ and are shifted by $-\lambda$, so $w_\lambda(B)=w(B)-\lambda r$; this is a constant shift of $w(B)$ and since $B_0$ maximizes $w(B)$ among bases $B\subseteq S_0$, it also maximizes $w_\lambda(B)$ among them:
    \[
    w_\lambda(B_0)\geq w_\lambda(B)\qquad\text{for every basis }B\subseteq S_0. \tag{1}
    \]
    Second, since $e'\in C(B_0,f')\setminus\{f'\}$, the set $B_0-e'+f'$ is a basis, and
    \[
    w_\lambda(B_0-e'+f')-w_\lambda(B_0)=w_\lambda(f')-w_\lambda(e')=(w(f')+\lambda)-(w(e')-\lambda)=2\lambda-\big(w(e')-w(f')\big)=2\lambda-2\delta^*,
    \]
    which is $\geq 0$ whenever $\lambda\geq\delta^*$. Hence
    \[
    w_\lambda(B_0-e'+f')\geq w_\lambda(B_0)\qquad\text{for every }\lambda\geq\delta^*. \tag{2}
    \]
    Combining (1) and (2), $w_\lambda(B_0-e'+f')\geq w_\lambda(B)$ for every basis $B\subseteq S_0$. Since $B_0-e'+f'\not\subseteq S_0$ (as $f'\in S\setminus S_0$), it follows that a basis of maximum $w_\lambda$-weight over all bases of $M$ cannot be contained in $S_0$: either $B_0-e'+f'$ itself attains this maximum, or some other basis exceeds the weight $w_\lambda(B_0-e'+f')$ and therefore also exceeds the weight of every $S_0$-basis, hence is itself not contained in $S_0$. Thus $w_\lambda$ is feasible for every $\lambda\geq\delta^*$.

    We claim that every feasible weight function is at $\ell_\infty$-distance at least $\delta^*$ from $w$. Suppose, for a contradiction, that there exists a feasible weight function $w''$ with
    $ \|w-w''\|_\infty<\delta^*. $
    Let $B'$ be a maximum $w''$-weight basis that is not contained in $S_0$. Choose an element $f\in B'\setminus S_0$. By Proposition~\ref{prop:symmetric-exchange}, there exists $e\in B_0\setminus B'$ such that both $B_0-e+f$ and $B'-f+e$ are bases.
    Since $B'$ is of maximum $w''$-weight,
    $ w''(B')\geq w''(B'-f+e), $
    and therefore
    $ w''(f)\geq w''(e). $
    On the other hand, by the definition of $\delta^*$ and given that $e\in C(B_0,f)\setminus\{f\}$,
    $ w(e)-w(f)\geq 2\delta^*. $
    Moreover, since $\|w-w''\|_\infty<\delta^*$, we have
    $ w''(e)>w(e)-\delta^* $
    and
    $ w''(f)<w(f)+\delta^*. $
    Combining these inequalities gives
    $ w''(e)>w(e)-\delta^*\geq w(f)+\delta^*>w''(f), $
    contradicting $w''(f)\geq w''(e)$.

    Hence every feasible solution has $\ell_\infty$-distance at least $\delta^*$ from $w$; in particular, applying this to the feasible solution $\widehat w$ of the statement, $\delta=\|w-\widehat w\|_\infty\geq \delta^*$.

    Setting $w'\coloneqq w_\delta=w+\delta\cdot(\chi_{S\setminus S_0}-\chi_{S_0})$, since $\delta\geq\delta^*$, $w'$ is feasible by the first part, so it satisfies (b) by construction, and, since $w'-w$ equals $\delta$ on $S\setminus S_0$ and $-\delta$ on $S_0$, $\|w-w'\|_\infty=\delta$, so it satisfies (a) as well.

    Finally, since $w_{\delta^*}$ is feasible with $\|w-w_{\delta^*}\|_\infty=\delta^*$, and every feasible solution is at $\ell_\infty$-distance at least $\delta^*$ from $w$, the optimal value equals $\delta^*$. If $w\in\bZ^S$, every term $w(e)-w(f)$ is an integer, so each term $(w(e)-w(f))/2$ in the definition of $\delta^*$ is a half-integer, and hence $\delta^*$ is half-integral.
  \end{proof}

We obtain a similar statement under the $\ell_1$-norm.
\begin{lemma}\label{lem:relaxed-nice-form-L1}
Let $\widehat{w}$ be a feasible solution for the \textsc{IM-Outside}
instance $(M=(S,\cI),S_0,w,\|\cdot\|_1)$, and let
$\delta=\|w-\widehat{w}\|_1$. Suppose that $w$ has a maximum-weight
basis $B_0\subseteq S_0$, and let $B_{\max}$ be a maximum
$\widehat{w}$-weight basis not contained in $S_0$.

Then there exists a feasible solution $w'$ satisfying the following:
\begin{enumerate}[label=(\alph*)]\itemsep0em
    \item $\|w-w'\|_1=\delta$, and 
     \item $w'=w + \delta\cdot \chi_f$ for some element $f \in B_{\max}\setminus S_0$.
\end{enumerate}
Furthermore, $f$ is contained in a basis of maximum $w'$-weight. 
Moreover, if no maximum $\widehat w$-weight basis is contained in $S_0$, then the same choice of $f$ satisfies $f\notin\cl(S_0[w'\ge w'(f)])$.
\end{lemma}
\begin{proof}
If $\delta=0$, the statement clearly holds. Suppose $\delta>0$.
By Proposition~\ref{prop:exchange}, there exists a bijection  $\varphi\colon B_0\setminus B_{\max}\to B_{\max}\setminus B_0$ such that $B_0-e+\varphi(e)$ is a basis for every $e\in B_0\setminus B_{\max}$. Therefore,
     \[ w(B_0)- w(B_{\max}) = \sum_{e\in B_0\setminus B_{\max}} w(e) - w(\varphi(e)).\]
    As $B_0$ has maximum $w$-weight and $\varphi(e) \in C(B_0,\varphi(e))$, Proposition~\ref{prop:charact_optimum_basis} implies $w(e) - w(\varphi(e))\geq 0$ for every exchanged pair. Hence, for every pair $e\in B_0\setminus B_{\max}, \varphi(e) \in  B_{\max}\setminus B_0$, we have $w(e) - w(\varphi(e)) \leq  w(B_0)- w(B_{\max})$. 
    On the other hand, we show that $w(B_0)- w(B_{\max})\leq \delta$. This follows from feasibility of $\widehat{w}$, since $\widehat{w}(B_0) \leq \widehat{w}(B_{\max})$ implies $0\leq \widehat{w}(B_{\max}) - \widehat{w}(B_0)$. Adding $w(B_0) - w(B_{\max})$ to both sides yields
    \[
       w(B_0) - w(B_{\max}) \leq  \sum_{s \in B_{\max}} (\widehat{w}(s) - w(s)) +  \sum_{s \in B_{0}} (w(s) - \widehat{w}(s)) \leq \sum_{s\in S} |\widehat{w}(s) - w(s)| = \delta.
     \]
    Choose an $e\in B_0\setminus B_{\max}$ for which $\varphi(e)\in S\setminus S_0$; such an $e$ exists since $B_{\max}\not\subseteq S_0$. By the above, $w(e)-w(\varphi(e))\le \delta$. 
    Let $f=\varphi(e)$, and consider $w'=w+\delta\cdot\chi_f$. We show that $w'$ is feasible.
    Let $i$ be the iteration at which the Greedy algorithm (with respect to $w$, breaking ties in favor of $B_0$) selects $e$. Since $B_0-e+f$ is a basis, we have $f\notin \cl(B_0-e)$. Therefore, before iteration $i$, the set selected by the Greedy algorithm is contained in $B_0-e$, and thus $f$ is not in its closure. This set spans $S[w>w(e)]$, so $f\notin\cl(S[w>w(e)])$. Moreover,
    \[
    w'(f)=w(f)+\delta\ge w(e)=w'(e).
    \]
   If no maximum $\widehat w$-weight basis is contained in $S_0$, then the
inequality above is strict. Set $A=S_0[w'\ge w'(f)]$. Since $w'$ agrees
with $w$ on $S_0$ and $B_0$ is a maximum $w$-weight basis of $M|S_0$,
the set $B_0\cap A$ spans $A$. Also, $e\notin A$ and $B_0-e+f$ is a basis,
so $f\notin\cl(B_0-e)$. Hence $
    f\notin\cl(A)=\cl(S_0[w'\ge w'(f)]).$

Run the Greedy algorithm with respect to $w'$, breaking ties in favor of $f$. Before $f$ is considered, only elements in $S[w'>w'(f)]$ have been processed. Since $w'$ differs from $w$ only at $f$ and $w'(f)\geq w(e)$, this set is contained in $S[w>w(e)]$. Therefore, the elements selected before $f$ do not span $f$, and the Greedy algorithm selects $f$. 
 Let $B'$ be the basis returned by the algorithm. Since $f\in B'$ and $f\notin S_0$, the basis $B'$ is not contained in $S_0$. Therefore, $w'$ is feasible. Finally, $\|w-w'\|_1=\delta$
    which completes the proof.
\end{proof}
\subsection{Algorithm for \textsc{IM-Outside} under \texorpdfstring{$\ell_\infty$}{l-infty}-norm}

Lemma~\ref{lem:relaxedProb_Nicer_special_weight_2} gives an algorithm, presented as Algorithm~\ref{algo:relaxed-NOT-IM-Exists}. At a high level, the algorithm  finds a
pair $f \in S \setminus S_0$, $e \in C(B_0, f)\setminus\{f\}$, where $B_0$ is a
maximum $w$-weight basis inside $S_0$, minimizing $w(e) - w(f)$. An illustration is shown in Figure~\ref{fig:example_IM-OUTSIDE-LINFTY}.

\begin{algorithm}[th!]
\caption{Combinatorial algorithm for \textsc{IM-Outside} under the $\ell_\infty$-norm}\label{algo:relaxed-NOT-IM-Exists}
\DontPrintSemicolon
\KwIn{An instance $( M=(S, \cI),  S_0, w, \|\cdot\|_\infty )$ of \textsc{IM-Outside}.}
\KwOut{An optimal weight function $\widehat{w}$.}
   \lIf {$w$ is feasible} {
        Set $\widehat{w}\coloneqq w$.
    }
    \Else{
        Let $B_0\subseteq S_0$ be a basis of maximum $w$-weight. \;
        Let $\delta= \min\{ (w(e)-w(f))/2 \colon f\in S\setminus S_0,\ e\in C(B_0,f)\setminus\{f\}\}$. \;
        Set $\widehat{w}\coloneqq w + \delta\cdot \left( \chi_{ S \setminus S_0} -\chi_{S_0} \right)$. \;
        }
    \Return{$\widehat{w}$}\;
\end{algorithm}

\begin{figure}[ht!]
\centering
    \begin{subfigure}[b]{0.33\textwidth}
        \centering
        \begin{tikzpicture}[scale=0.9,
            myEdge/.style={line width = 1.5pt},
            myEdgeZigzag/.style={line width = 1.5pt, decoration = {zigzag, segment length = 4pt, amplitude = 1pt},decorate},
            state/.style={circle,  minimum size=1em, draw, line width = 1.2pt}]   
            \node (a) at (0,2) [state] {$a$};
            \node (b) at (4,2) [state] {$b$};
            \node (c) at (1,1) [state] {$c$};
            \node (d) at (3,1) [state] {$d$};
            \node (e) at (0,0) [state] {$e$};
            \node (f) at (4,0) [state] {$f$}; 
            \draw [myEdge] (a) to node[midway, above] {3} (b) ; 
            \draw [myEdge] (e) to node[midway, below] {2} (f) ;
            \foreach \u \v \w \pos in { a/c/4/right, c/d/10/above, d/b/5/left,  e/c/4/right, d/f/5/left}{
                \draw [MyOrange!30, line width = 5pt] (\u) to (\v) ; 
                \draw [myEdgeZigzag] (\u) to node[midway,\pos] {\w} (\v) ; 
            } 
            \foreach \u \v \w \pos in {a/e/4/left,b/f/5/right}{
                \draw [MyOrange!30, line width = 5pt] (\u) to (\v) ; 
                \draw [myEdge] (\u) to node[midway,\pos] {\w} (\v) ; 
            }
        \end{tikzpicture}
        \caption{Original instance.}
        \label{fig:example_IM-outside_original}
    \end{subfigure}%
    \begin{subfigure}[b]{0.33\textwidth}
        \centering
        \begin{tikzpicture}[scale=0.9,
            myEdge/.style={line width = 1.5pt},
            myEdgeZigzag/.style={line width = 1.5pt, decoration = {zigzag, segment length = 4pt, amplitude = 1pt},decorate},
            state/.style={circle,  minimum size=1em, draw, line width = 1.2pt}]   
            \node (a) at (0,2) [state] {$a$};
            \node (b) at (4,2) [state] {$b$};
            \node (c) at (1,1) [state] {$c$};
            \node (d) at (3,1) [state] {$d$};
            \node (e) at (0,0) [state] {$e$};
            \node (f) at (4,0) [state] {$f$}; 
            \draw [myEdge] (a) to node[midway, above] {3.5} (b) ; 
            \draw [myEdge] (e) to node[midway, below] {2.5} (f) ;
            \foreach \u \v \w \pos in { a/c/3.5/right, c/d/9.5/above, d/b/4.5/left,  e/c/3.5/right, d/f/4.5/left, a/e/3.5/left, b/f/4.5/right}{
                \draw [MyOrange!30, line width = 5pt] (\u) to (\v) ; 
                \draw [myEdge] (\u) to node[midway,\pos] {\w} (\v) ; 
            } 
        \end{tikzpicture}
        \caption{Optimal weight function.}
        \label{fig:IM-OUSIDE-FINAL}
    \end{subfigure}%
    \caption{Illustration of Algorithm~\ref{algo:relaxed-NOT-IM-Exists} on the graphic matroid in Figure~\ref{fig:example_IM-outside_original}. The basis $B_0$ is shown in zigzagged lines. Thus $\delta=1/2$, where the minimum is attained by the pair $ab$ and $ac$. The optimal weight function $\widehat{w}$ is shown in Figure~\ref{fig:IM-OUSIDE-FINAL}, and a basis of maximum $\widehat{w}$-weight not contained in $S_0$ is the spanning tree $B_0 - ac + ab$.
    } 
    \label{fig:example_IM-OUTSIDE-LINFTY}
\end{figure}

\begin{thm}\label{thm:algo-IM-Outside-infty}
    Algorithm~\ref{algo:relaxed-NOT-IM-Exists} determines an optimal weight function $\widehat{w}$ for the \textsc{IM-Outside} instance $( M=(S, \cI),  S_0, w, \|\cdot\|_\infty )$ with $w\in\bR^S$ 
    using $O(n\log n)$ operations and $O(n+r \cdot |S\setminus S_0|)$ independence-oracle calls.
\end{thm}
\begin{proof}
    The correctness of the algorithm follows from Lemma~\ref{lem:relaxedProb_Nicer_special_weight_2}.

    \textit{Running time.} First, checking feasibility of $w$ is done using the Greedy algorithm in $O(|S|\log |S|)$ time and $O(|S|)$ independence-oracle calls by Remark~\ref{rem:relaxedProb_feasibility}. If $w$ is not feasible, the last feasibility check returned $B_0\subseteq S_0$ of maximum weight. For each $f\in S\setminus S_0$, Remark~\ref{rem:findingFundCircuit} computes $C(B_0,f)$ using $O(|B_0|)$ independence-oracle calls; repeating this for all $f\in S\setminus S_0$ and taking the minimum of $(w(e)-w(f))/2$ over all $e\in C(B_0,f)\setminus\{f\}$ yields $\delta$ using $O(|B_0|\cdot|S\setminus S_0|)$ independence-oracle calls. Together with the initial feasibility check, the algorithm uses $O(n+r\cdot|S\setminus S_0|)$ independence-oracle calls.
\end{proof}

Lemma~\ref{lem:relaxedProb_Nicer_special_weight_2} has another interesting consequence. Namely, if integrality constraints are added on the weight function $\widehat{w}$ in \textsc{IM-Outside}, then an optimal solution can be obtained from an optimal fractional solution in a straightforward manner.
An illustration is shown in Figure~\ref{fig:example_IM-OUTSIDE-INTEGRAL}.

\begin{cor}\label{cor:rine}
    For the \textsc{Integral-IM-Outside} problem under the $\ell_\infty$-norm, where $w\in\bZ^S$ and $\widehat{w}$ is required to be integer-valued, an analogous approach works: it suffices to replace  $\delta$ by its ceiling.
\end{cor}
\begin{figure}[ht!]
\centering
    \begin{subfigure}[b]{0.33\textwidth}
        \centering
        \begin{tikzpicture}[scale=0.9,
            myEdge/.style={line width = 1.5pt},
            myEdgeZigzag/.style={line width = 1.5pt, decoration = {zigzag, segment length = 4pt, amplitude = 1pt},decorate},
            state/.style={circle,  minimum size=1em, draw, line width = 1.2pt}]   
            \node (a) at (0,2) [state] {$a$};
            \node (b) at (4,2) [state] {$b$};
            \node (c) at (1,1) [state] {$c$};
            \node (d) at (3,1) [state] {$d$};
            \node (e) at (0,0) [state] {$e$};
            \node (f) at (4,0) [state] {$f$}; 
            \draw [myEdge] (a) to node[midway, above] {4} (b) ; 
            \draw [myEdge] (e) to node[midway, below] {3} (f) ;
            \foreach \u \v \w \pos in { a/c/3/right, c/d/9/above, d/b/4/left,  e/c/3/right, d/f/4/left, a/e/3/left, b/f/4/right}{
                \draw [MyOrange!30, line width = 5pt] (\u) to (\v) ; 
                \draw [myEdge] (\u) to node[midway,\pos] {\w} (\v) ; 
            } 
        \end{tikzpicture}
    \end{subfigure}%
    \caption{Illustration of Corollary~\ref{cor:rine} on the graphic matroid in Figure~\ref{fig:example_IM-outside_original}, where the optimal value is $1$ and the optimal weight function $\widehat{w}$ is shown, and a basis of maximum $\widehat{w}$-weight not contained in $S_0$ is the spanning tree $B_0 - ac + ab$.
    } 
    \label{fig:example_IM-OUTSIDE-INTEGRAL}
\end{figure}
\subsection{Algorithm for \textsc{IM-Outside} under \texorpdfstring{$\ell_1$}{l-1}-norm}

The proof of Lemma~\ref{lem:relaxed-nice-form-L1} yields a straightforward algorithm, shown in Algorithm~\ref{algo:IM-OUTSIDE-L1}: 
find a maximum $w$-weight basis $B_0\subseteq S_0$, then find the pair $e\in B_0$, 
$f\in S\setminus S_0$ with $e\in C(B_0,f)\setminus\{f\}$ minimizing $w(e)-w(f)$. An illustration is shown in Figure~\ref{fig:example_IM-OUTISDE-L1}.
\begin{algorithm}[th!]
\caption{Algorithm for \textsc{IM-Outside} under the $\ell_1$-norm}\label{algo:IM-OUTSIDE-L1}
\DontPrintSemicolon
\KwIn{An instance $( M=(S, \cI),  S_0, w, \|\cdot\|_1)$ of \textsc{IM-Outside}.}
\KwOut{An optimal weight for \textsc{IM-Outside}.}
    \lIf {$w$ is feasible} {
        Set $\widehat{w}\coloneqq w$.
    }
    \Else{
        Let $B_0\subseteq S_0$ be a basis of maximum $w$-weight. \;
        Let $(e^*,f^*)= \argmin\{ w(e)-w(f) \colon f\in S\setminus S_0,\ e\in C(B_0,f)\setminus\{f\}\}$. \;
        Set $\delta\coloneqq w(e^*)-w(f^*)$. \;
        Set $\widehat{w}\coloneqq w+\delta\cdot\chi_{f^*}$.
        }
    \Return{$\widehat{w}$}\;
\end{algorithm}

\begin{figure}[ht!]
\centering
    \begin{subfigure}[b]{0.33\textwidth}
        \centering
        \begin{tikzpicture}[scale=0.9,
            myEdge/.style={line width = 1.5pt},
            myEdgeZigzag/.style={line width = 1.5pt, decoration = {zigzag, segment length = 4pt, amplitude = 1pt},decorate},
            state/.style={circle,  minimum size=1em, draw, line width = 1.2pt}]   
            \node (a) at (0,2) [state] {$a$};
            \node (b) at (4,2) [state] {$b$};
            \node (c) at (1,1) [state] {$c$};
            \node (d) at (3,1) [state] {$d$};
            \node (e) at (0,0) [state] {$e$};
            \node (f) at (4,0) [state] {$f$}; 
            \draw [myEdge] (a) to node[midway, above] {3} (b) ; 
            \draw [myEdge] (e) to node[midway, below] {2} (f) ;
            \foreach \u \v \w \pos in { a/c/4/right, c/d/10/above, d/b/5/left,  e/c/4/right, d/f/5/left}{
                \draw [MyOrange!30, line width = 5pt] (\u) to (\v) ; 
                \draw [myEdgeZigzag] (\u) to node[midway,\pos] {\w} (\v) ; 
            } 
            \foreach \u \v \w \pos in {a/e/4/left,b/f/5/right}{
                \draw [MyOrange!30, line width = 5pt] (\u) to (\v) ; 
                \draw [myEdge] (\u) to node[midway,\pos] {\w} (\v) ; 
            }
        \end{tikzpicture}
        \caption{Original instance.}
        \label{fig:example_IM-outside_original-L1}
    \end{subfigure}%
    \begin{subfigure}[b]{0.33\textwidth}
        \centering
        \begin{tikzpicture}[scale=0.9,
            myEdge/.style={line width = 1.5pt},
            myEdgeZigzag/.style={line width = 1.5pt, decoration = {zigzag, segment length = 4pt, amplitude = 1pt},decorate},
            state/.style={circle,  minimum size=1em, draw, line width = 1.2pt}]   
            \node (a) at (0,2) [state] {$a$};
            \node (b) at (4,2) [state] {$b$};
            \node (c) at (1,1) [state] {$c$};
            \node (d) at (3,1) [state] {$d$};
            \node (e) at (0,0) [state] {$e$};
            \node (f) at (4,0) [state] {$f$}; 
            \draw [myEdge] (a) to node[midway, above] {4} (b) ; 
            \draw [myEdge] (e) to node[midway, below] {2} (f) ;
            \foreach \u \v \w \pos in { a/c/4/right, c/d/10/above, d/b/5/left,  e/c/4/right, d/f/5/left, a/e/4/left, b/f/5/right}{
                \draw [MyOrange!30, line width = 5pt] (\u) to (\v) ; 
                \draw [myEdge] (\u) to node[midway,\pos] {\w} (\v) ; 
            } 
        \end{tikzpicture}
        \caption{Optimal weight function.}
        \label{fig:IM-OUSIDE-FINAL-L1}
    \end{subfigure}%
    \caption{Illustration of Algorithm~\ref{algo:IM-OUTSIDE-L1} on the graphic matroid in Figure~\ref{fig:example_IM-outside_original-L1}. The basis $B_0$ is shown in zigzagged lines. Thus $\delta=1$, where the minimum is attained by the pair $ab$ and $ac$. The optimal weight function $\widehat{w}$ is shown in Figure~\ref{fig:IM-OUSIDE-FINAL-L1}, and a basis of maximum $\widehat{w}$-weight not contained in $S_0$ is the spanning tree $B_0 - ac + ab$.
    } 
    \label{fig:example_IM-OUTISDE-L1}
\end{figure}

\begin{thm}\label{thm:algo-IM-Outside-L1}
     Algorithm~\ref{algo:IM-OUTSIDE-L1} determines an optimal weight function $\widehat{w}$ for the \textsc{IM-Outside} instance $( M=(S, \cI),  S_0, w, \|\cdot\|_1 )$ with $w\in\bR^S$
     using $O(n\log n)$ operations and $O(n+r\cdot|S\setminus S_0|)$ independence-oracle calls.
     Furthermore, if $w$ is an integer-valued weight function, then the optimal solution is also integer-valued.
\end{thm}
\begin{proof}
    The correctness of the algorithm follows from Lemma~\ref{lem:relaxed-nice-form-L1}.
    The integrality claim follows immediately.
   
   \textit{Running time.} First, checking feasibility of $w$ is done using the Greedy algorithm in $O(|S|\log |S|)$ time and $O(|S|)$ independence-oracle calls by Remark~\ref{rem:relaxedProb_feasibility}.
   If $w$ is not feasible, the last feasibility check returned $B_0\subseteq S_0$ of maximum weight. 
   For each $f\in S\setminus S_0$, Remark~\ref{rem:findingFundCircuit} computes $C(B_0,f)$ using $O(|B_0|)$ independence-oracle calls; repeating this for all $f\in S\setminus S_0$ and taking the minimum of $w(e)-w(f)$ over all $e\in C(B_0,f)\setminus\{f\}$ yields $\delta$ using $O(|B_0|\cdot|S\setminus S_0|)$ independence-oracle calls. Together with the initial feasibility check, the algorithm uses $O(n+r\cdot|S\setminus S_0|)$ independence-oracle calls.
\end{proof}

\section{\textsc{Inverse Matroid Not Exists}}
\label{sec:not-IM-Exists}

We now turn to the first negated variant, \textsc{IM-Not-Exists}, the negation of \textsc{IM-Exists}. The goal is to perturb the weights so that no basis contained in $S_0$ is maximum-weight. This is stronger than \textsc{IM-Outside}: it is not enough to create one optimal basis outside $S_0$; all optimal bases contained in $S_0$ must be eliminated.

Related one-sided versions of partial inverse matroid optimization were considered in~\cite{li2016algorithm,zhang2016partialmatroid}. In those works, the prescribed set is an independent set that has to be contained in a maximum-weight basis, and the weights are allowed to move only in one direction, either only by increases or only by decreases. Our setting is different in two ways: the set $S_0$ is an arbitrary subset of the ground set, and we allow both increases and decreases. This makes the negated variants more flexible, but it also introduces the strict-separation issue already encountered for \textsc{IM-Only}: eliminating every optimal basis contained in $S_0$ requires a strict inequality that, over real-valued weights, need not be attained by any feasible solution.

We resolve this by restricting attention to integer-valued weights, under which the infimum defining $\delta^*$ is always attained. Over the reals, this need not hold: let $M$ be the uniform matroid of rank one on $\{a,b\}$, let $S_0=\{a\}$, and let $w(a)=w(b)=0$. There is no feasible solution with deviation zero. However, for every $\delta>0$, the weight function defined by $w^*(a)=0$ and $w^*(b)=\delta$ is feasible and has deviation $\|w-w^*\|_1=\delta$ and $\|w-w^*\|_\infty=\delta$.

\problemdef{Inverse Matroid Not Exists (\textsc{IM-Not-Exists})}
    {A matroid $M=(S,\mathcal{I})$, a subset $S_0\subseteq S$ such that some basis of $M$ is not contained in $S_0$, a weight function $w\in\bZ^S$, and an objective function $\|\cdot\|$ defined on $\bR^S$.}
    {Find a weight function $w^* \in \bZ^S$ such that 
    \begin{enumerate}[label=(\alph*),topsep=1px,partopsep=2px]\itemsep0em
        \item \label{it:feas_cond_not-IM-Exists} none of the bases contained in $S_0$ is of maximum $w^*$-weight, and
        \item \label{it:opt_cond_not-IM-Exists} $\| w - w^* \|$ is minimized.
    \end{enumerate}}

A weight function is \emph{feasible} if condition~\ref{it:feas_cond_not-IM-Exists} holds. Note that \textsc{IM-Not-Exists} is always feasible: set $w'$ such that $w'(s)=0$ for all $s \in S_0$ and $w'(s)=1$ for all $s \in S \setminus S_0$. Here we point out that the condition on $S_0$ allows us to avoid infeasible instances. In this case, we observe the following. We may first delete loops from $S\setminus S_0$ without changing the optimum: loops belong to no basis, so their weights do not affect the maximum-weight bases, and in an optimal solution their weights are left unchanged. If this leaves $S\setminus S_0=\emptyset$, then every basis is contained in $S_0$, and the instance is infeasible.

\begin{remark}\label{rem:IM-NOT-EXISTS-feasib}
    Given a weight function $w'$, we can verify in polynomial time whether $w'$ is feasible for \textsc{IM-Not-Exists}. To do so, run the greedy algorithm while breaking ties in favor of elements in $S_0$. Then $w'$ is feasible if and only if the resulting basis is not contained in $S_0$.
\end{remark}
\subsection{IM-Not-Exists under \texorpdfstring{$\ell_\infty$}{l-infinity}-norm}
There are three cases. The first two are immediate. In the remaining case, we reduce the problem to \textsc{IM-Outside}. We start with the immediate cases.
If $S_0$ does not contain a maximum $w$-weight basis, then $w$ is already feasible. The second case is stated below.

\begin{lemma}\label{lem:Not-IM-Exists_easy_case}
    Let $(M=(S,\cI),S_0,w,\|\cdot\|_\infty)$ be an instance of \textsc{IM-Not-Exists} and assume that $M$ has a maximum $w$-weight basis contained in $S_0$, and also one not contained in $S_0$. Then $w^*=w-\chi_{S_0}+\chi_{S\setminus S_0}$ is an optimal solution for the problem.
\end{lemma}
\begin{proof}
    By the assumptions of the lemma and the integrality constraints, we know that the optimum value of any feasible solution is at least $1$. Since $\|w-w^*\|_\infty=1$, it suffices to show that $w^*$ is feasible. To see this, let $B_0\subseteq S_0$ be an arbitrary basis of $M$, and let $B_{\max}$ be a maximum $w$-weight basis of $M$ not contained in $S_0$. Then 
    \begin{align*}
        w^*(B_0)
        & = w(B_0) - |B_0| \\
        & \leq w(B_{\max}) - |B_{\max}| \\
        & < w(B_{\max}) - |B_{\max}\cap S_0| + |B_{\max}\setminus S_0|\\
        & = w^*(B_{\max}),
    \end{align*}
This concludes the proof.
\end{proof}

The remaining case can be solved using the following observation, which relates to \textsc{Integral-IM-Outside}.

\begin{lemma}\label{lem:isudufb}
    Let $\widehat{\delta}\in\bZ_+$ be the optimum value for a \textsc{Integral-IM-Outside} instance $( M=(S, \cI),  S_0, w, \allowbreak\|\cdot\|_\infty )$. Then either $w+\widehat{\delta}\cdot(\chi_{S\setminus S_0}-\chi_{S_0})$ or $w+(\widehat{\delta}+1)\cdot(\chi_{S\setminus S_0}-\chi_{S_0})$ is an optimal solution to the \textsc{IM-Not-Exists} problem on the same instance.
\end{lemma}
\begin{proof}
        For ease of discussion, let $w'\coloneqq w+\widehat{\delta}\cdot(\chi_{S\setminus S_0}-\chi_{S_0})$. By Corollary~\ref{cor:rine}, $w'$ is an optimal solution for \textsc{IM-Outside}. If this weight function is feasible for \textsc{IM-Not-Exists}, then it is also optimal, because \textsc{IM-Outside} is a relaxation of \textsc{IM-Not-Exists}. Otherwise, there exists a basis $B_0 \subseteq S_0$ which has maximum $w'$-weight. Note that, by definition of $\widehat{\delta}$, there exists a basis $B\not \subseteq S_0$ that has maximum $w'$-weight. Applying Lemma~\ref{lem:Not-IM-Exists_easy_case} for the instance $( M=(S, \cI),  S_0, w', \|\cdot\|_\infty )$, the statement follows.
      \end{proof}

Our algorithm is presented as Algorithm~\ref{algo:NOT-IM-Exists}. The high-level idea is to distinguish among the three possible cases: (1) the original weight function is already optimal; (2) the optimal solution to \textsc{Integral-IM-Outside} is also optimal for the current problem; or (3) the desired solution is obtained by taking an optimal solution to \textsc{Integral-IM-Outside}, decreasing the weights of the elements in $S_0$ by one unit, and increasing the weights of all remaining elements by one unit. An illustration is shown in Figure~\ref{fig:example_IM-NOT-EXISTS-linfty}.

\begin{algorithm}[th!]
\caption{Algorithm for \textsc{IM-Not-Exists} under the $\ell_\infty$-norm}\label{algo:NOT-IM-Exists}
\DontPrintSemicolon
\KwIn{An instance $( M=(S, \cI),  S_0, w, \|\cdot\|_\infty )$ of \textsc{IM-Not-Exists}.}
\KwOut{An optimal weight for \textsc{IM-Not-Exists}.}
    \lIf {$w$ is feasible} {
        Set $w^*\coloneqq w$.
    }
    \Else{
        \vspace{2pt}
        Let $\widehat{\delta}$ be the optimal value of \textsc{Integral-IM-Outside} for instance $( M=(S, \cI),  S_0, w, \|\cdot\|_\infty )$.\;
        \lIf{$w+\widehat{\delta}\cdot(\chi_{S\setminus S_0}-\chi_{S_0})$ is feasible}
            {
               Set $w^*\coloneqq w+\widehat{\delta}\cdot(\chi_{S\setminus S_0}-\chi_{S_0})$.
            }
        \lElse{
            Set $w^*\coloneqq w+(\widehat{\delta}+1)\cdot(\chi_{S\setminus S_0}-\chi_{S_0})$.
        }
    }   
    \Return{$w^*$}
\end{algorithm}

\begin{figure}[ht!]
\centering
    \begin{subfigure}[b]{0.33\textwidth}
        \centering
        \begin{tikzpicture}[scale=0.9,
            myEdge/.style={line width = 1.5pt},
            myEdgeZigzag/.style={line width = 1.5pt, decoration = {zigzag, segment length = 4pt, amplitude = 1pt},decorate},
            state/.style={circle,  minimum size=1em, draw, line width = 1.2pt}]   
            \node (a) at (0,2) [state] {$a$};
            \node (b) at (4,2) [state] {$b$};
            \node (c) at (1,1) [state] {$c$};
            \node (d) at (3,1) [state] {$d$};
            \node (e) at (0,0) [state] {$e$};
            \node (f) at (4,0) [state] {$f$}; 
            \draw [myEdge] (a) to node[midway, above] {3} (b) ; 
            \draw [myEdge] (e) to node[midway, below] {2} (f) ;
            \foreach \u \v \w \pos in { a/c/4/right, c/d/10/above, d/b/5/left,  e/c/4/right, d/f/5/left}{
                \draw [MyOrange!30, line width = 5pt] (\u) to (\v) ; 
                \draw [myEdge] (\u) to node[midway,\pos] {\w} (\v) ; 
            } 
            \foreach \u \v \w \pos in {a/e/4/left,b/f/5/right}{
                \draw [MyOrange!30, line width = 5pt] (\u) to (\v) ; 
                \draw [myEdge] (\u) to node[midway,\pos] {\w} (\v) ; 
            }
        \end{tikzpicture}
        \caption{Original instance.}
        \label{fig:example_IM-NOT-EXISTS-ORIGINAL}
    \end{subfigure}%
    \begin{subfigure}[b]{0.33\textwidth}
        \centering
        \begin{tikzpicture}[scale=0.9,
            myEdge/.style={line width = 1.5pt},
            myEdgeZigzag/.style={line width = 1.5pt, decoration = {zigzag, segment length = 4pt, amplitude = 1pt},decorate},
            state/.style={circle,  minimum size=1em, draw, line width = 1.2pt}]   
            \node (a) at (0,2) [state] {$a$};
            \node (b) at (4,2) [state] {$b$};
            \node (c) at (1,1) [state] {$c$};
            \node (d) at (3,1) [state] {$d$};
            \node (e) at (0,0) [state] {$e$};
            \node (f) at (4,0) [state] {$f$}; 
            \draw [myEdge] (a) to node[midway, above] {5} (b) ; 
            \draw [myEdge] (e) to node[midway, below] {4} (f) ;
            \foreach \u \v \w \pos in { a/c/2/right, c/d/8/above, d/b/3/left,  e/c/2/right, d/f/3/left, a/e/2/left, b/f/3/right}{
                \draw [MyOrange!30, line width = 5pt] (\u) to (\v) ; 
                \draw [myEdge] (\u) to node[midway,\pos] {\w} (\v) ; 
            } 
        \end{tikzpicture}
        \caption{Optimal weight function.}
        \label{fig:IM-NOT-EXISTS-FINAL-linfty}
    \end{subfigure}%
    \caption{Illustration of Algorithm~\ref{algo:NOT-IM-Exists} on the graphic matroid in Figure~\ref{fig:example_IM-NOT-EXISTS-ORIGINAL}.
    The optimal solution for \textsc{Integral-IM-Outside} shown in Figure~\ref{fig:example_IM-OUTSIDE-INTEGRAL} is not feasible for \textsc{IM-Not-Exists}. The optimal weight function $w^*$ is shown in Figure~\ref{fig:IM-NOT-EXISTS-FINAL-linfty} with optimal value $2$. Under $w^*$, the edge $ab$ is contained in every maximum-weight basis and cannot be exchanged for any other edge. Therefore, no basis contained in $S_0$ has maximum weight, as required for \textsc{IM-Not-Exists}.
    }
    \label{fig:example_IM-NOT-EXISTS-linfty}
\end{figure}

\begin{thm}\label{thm:NOT-IM-Exists-algorithm}
    Algorithm~\ref{algo:NOT-IM-Exists} determines an optimal weight function $w^*$ for the \textsc{IM-Not-Exists} instance  $( M=(S, \cI),  S_0, w, \|\cdot\|_\infty )$ 
    using $O(n\log n)$ operations and $O(n+r\cdot|S\setminus S_0|)$ independence-oracle calls. 
\end{thm}
\begin{proof}
    The correctness of the algorithm follows from Lemma~\ref{lem:isudufb}. Since one can find a solution of \textsc{IM-Outside} with integrality constraints in polynomial time by Corollary~\ref{cor:rine}, the theorem follows.

    \textit{Running time.} By Remark~\ref{rem:IM-NOT-EXISTS-feasib}, testing feasibility of $w$ and $w+\widehat{\delta}\cdot(\chi_{S\setminus S_0}-\chi_{S_0})$ is done using $O(|S|\log |S|)$ operations and $O(|S|)$ independence-oracle calls.
    By Theorem~\ref{thm:algo-IM-Outside-infty}, we can find $\widehat{\delta}$ using $O(|S|\log |S|)$ operations and $O(|B_0|\cdot|S\setminus S_0|)$ independence-oracle calls. Therefore, the algorithm uses
    $O(n+r\cdot|S\setminus S_0|)$ independence-oracle calls in total.
  \end{proof}

\subsection{IM-Not-Exists under \texorpdfstring{$\ell_1$}{l1}-norm}

We use the following notation in this section. Let $\overline{S_0}\coloneqq S\setminus S_0$ be the set of outside elements. For $A\subseteq S$, a weight function $x$, and a threshold $t$, define
\[
A[x\ge t]\coloneqq \{e\in A\colon x(e)\ge t\}.
\]

\begin{remark} \label{rem:IM-Not-Exists-CLOSURE}
If $\cl(S_0)\neq S$, then no basis of $M$ is contained in $S_0$, since a basis spanning $S$ would require $\cl(S_0)=S$. Hence $w^*=w$ is feasible with deviation zero, so we assume henceforth that $\cl(S_0)=S$.
We may also delete loops from $\overline{S_0}$ without changing the optimum: loops belong to no basis, so their weights do not affect the maximum-weight bases, and in an optimal solution their weights are left unchanged.
\end{remark}

The following structural theorem shows that, under the assumptions of Remark~\ref{rem:IM-Not-Exists-CLOSURE}, there is an optimal solution that changes the weight of at most one outside element. 
We prove the theorem at the end of the section. Before that, we present the resulting algorithm and its running time.

For each $f\in \overline{S_0}$, define the threshold 
\[
T_f\coloneqq \min\{T\in\bZ\colon T\ge w(f)\text{ and }
f\notin \cl(S_0[w\ge T])\};
\]
Informally, $T_f$ is the smallest integer weight assignable to $f$ such that $f$ lies outside the closure of the high-weight elements of $S_0$.
This value is well-defined: for all sufficiently large $T$, we have
$S_0[w\ge T]=\emptyset$, and $f\notin\cl(\emptyset)$ because $f$ is
not a loop (loops were deleted from $\overline{S_0}$ in Remark~\ref{rem:IM-Not-Exists-CLOSURE}).

\begin{thm}\label{thm:imnotexists}
Assume that $\cl(S_0)=S$ and that $\overline{S_0}$ has no loops. Let $B'$ be a maximum-weight basis of $M|\overline{S_0}$. Then
\[
\OPT=\min \{ T_f-w(f) \colon f\in B' \}.
\]
Moreover, if $\OPT=0$, then $w^*=w$ is optimal. If $\OPT > 0$, there exists $f^* \in B'$ for which $w^*=w+\OPT\cdot \chi_{f^*}$ is an optimal weight function.
\end{thm}

Theorem~\ref{thm:imnotexists} reduces the problem to choosing one outside element whose weight should possibly be increased. For each outside element $f$, the threshold $T_f$ is the smallest integer value to which $f$ has to be raised in order to escape the closure of the high-weight elements of $S_0$. The theorem shows that it is enough to compute these thresholds for the elements of a maximum-weight basis $B'$ of $M|\overline{S_0}$. This leads to Algorithm~\ref{alg:notexistim}. An illustration is shown in Figure~\ref{fig:example_IM-NOT-EXISTS-L1}.

\begin{algorithm}[ht!]
\caption{Algorithm for \textsc{IM-Not-Exists} under the $\ell_1$-norm}\label{alg:notexistim}
\DontPrintSemicolon
\KwIn{An instance $( M=(S, \cI),  S_0, w, \|\cdot\|_1)$ of \textsc{IM-Not-Exists}.}
\KwOut{An optimal integer weight for \textsc{IM-Not-Exists}.
}
\lIf{$\cl(S_0)\neq S$}{\Return{$w^*\coloneqq w$}}
Delete loops from $\overline{S_0}$.\;
Compute a maximum-weight basis $B_{0}$ of $M|S_0$ and a maximum-weight basis $B'$ of $M|\overline{S_0}$.\;
\ForEach{$f\in B'$}{
    Let $L_f$ be the sorted list $\{w(f)\}\cup\{w(e)+1\colon  e\in B_{0},\ w(e)\ge w(f)\}$.\;
    Binary search the smallest $T\in L_f$ such that $B_{0}[w\ge T]\cup\{f\}\in\cI$.\;
    \tcp{Equivalently, $f\notin\cl(B_{0}[w\ge T])=\cl(S_0[w\ge T])$.}
    Set $T_f\coloneqq T$.\;
}
Let $f^*\in\argmin_{f\in B'} (T_f-w(f))$ and set $\OPT\coloneqq T_{f^*}-w(f^*)$.\;
Set $w^*=w+\OPT\cdot \chi_{f^*}$.\;
\Return{$w^*$}\;
\end{algorithm}

\begin{figure}[ht!]
\centering
    \begin{subfigure}[b]{0.33\textwidth}
        \centering
        \begin{tikzpicture}[scale=0.9,
            myEdge/.style={line width = 1.5pt},
            myEdgeZigzag/.style={line width = 1.5pt, decoration = {zigzag, segment length = 4pt, amplitude = 1pt},decorate},
            state/.style={circle,  minimum size=1em, draw, line width = 1.2pt}]   
            \node (a) at (0,2) [state] {$a$};
            \node (b) at (4,2) [state] {$b$};
            \node (c) at (1,1) [state] {$c$};
            \node (d) at (3,1) [state] {$d$};
            \node (e) at (0,0) [state] {$e$};
            \node (f) at (4,0) [state] {$f$}; 
            \draw [myEdge] (a) to node[midway, above] {3} (b) ; 
            \draw [myEdge] (e) to node[midway, below] {2} (f) ;
            \foreach \u \v \w \pos in { a/c/4/right, c/d/10/above, d/b/5/left,  e/c/4/right, d/f/5/left}{
                \draw [MyOrange!30, line width = 5pt] (\u) to (\v) ; 
                \draw [myEdgeZigzag] (\u) to node[midway,\pos] {\w} (\v) ; 
            } 
            \foreach \u \v \w \pos in {a/e/4/left,b/f/5/right}{
                \draw [MyOrange!30, line width = 5pt] (\u) to (\v) ; 
                \draw [myEdge] (\u) to node[midway,\pos] {\w} (\v) ; 
            }
        \end{tikzpicture}
        \caption{Original instance.}
        \label{fig:example_IM-not-exists_original-L1}
    \end{subfigure}%
    \begin{subfigure}[b]{0.33\textwidth}
        \centering
        \begin{tikzpicture}[scale=0.9,
            myEdge/.style={line width = 1.5pt},
            myEdgeZigzag/.style={line width = 1.5pt, decoration = {zigzag, segment length = 4pt, amplitude = 1pt},decorate},
            state/.style={circle,  minimum size=1em, draw, line width = 1.2pt}]   
            \node (a) at (0,2) [state] {$a$};
            \node (b) at (4,2) [state] {$b$};
            \node (c) at (1,1) [state] {$c$};
            \node (d) at (3,1) [state] {$d$};
            \node (e) at (0,0) [state] {$e$};
            \node (f) at (4,0) [state] {$f$}; 
            \draw [myEdge] (a) to node[midway, above] {5} (b) ; 
            \draw [myEdge] (e) to node[midway, below] {2} (f) ;
            \foreach \u \v \w \pos in { a/c/4/right, c/d/10/above, d/b/5/left,  e/c/4/right, d/f/5/left}{
                \draw [MyOrange!30, line width = 5pt] (\u) to (\v) ; 
                \draw [myEdge] (\u) to node[midway,\pos] {\w} (\v) ; 
            } 
            \foreach \u \v \w \pos in {a/e/4/left,b/f/5/right}{
                \draw [MyOrange!30, line width = 5pt] (\u) to (\v) ; 
                \draw [myEdge] (\u) to node[midway,\pos] {\w} (\v) ; 
            }
        \end{tikzpicture}
        \caption{Optimal weight function.}
        \label{fig:IM-not-exists-FINAL-L1}
    \end{subfigure}%
    \caption{Illustration of Algorithm~\ref{alg:notexistim} on the graphic matroid in Figure~\ref{fig:example_IM-not-exists_original-L1}. The basis $B_0$ is shown in zigzagged lines, and $B'=S\setminus S_0$. We have $L_{ab}=\{3\}\cup\{5,6,11\}$ and $L_{ef}=\{2\}\cup\{5,6,11\}$, and thus $T_{ab}=5$ and $T_{ef}=5$. Hence $f^*$ is edge $ab$, implying that $\delta^*=2$; the optimal weight function $w^*$ is shown in Figure~\ref{fig:IM-not-exists-FINAL-L1}. Under $w^*$, the edge $ab$ is always preferred over $ac$; therefore, no basis contained in $S_0$ has maximum weight, as required for \textsc{IM-Not-Exists}.
    } 
    \label{fig:example_IM-NOT-EXISTS-L1}
\end{figure}

\begin{thm}\label{thm:algorithmINMOTEXISTS}
Algorithm~\ref{alg:notexistim} determines an optimal weight function $w^*$ for the \textsc{IM-Not-Exists} instance  $( M=(S, \cI),  S_0, w, \|\cdot\|_1)$
using $O(n\log n)$ operations and
$O(n+r\log r)$ independence-oracle calls.
\end{thm}

\begin{proof}
First handle the degenerate cases. If $\cl(S_0)\neq S$, the algorithm returns $w^*=w$, which is optimal because $S_0$ contains no basis of $M$. The test can be performed by running the greedy algorithm on the sorted ground set. The algorithm then deletes loops from $\overline{S_0}$. Observe that, after deleting loops, it is not possible to obtain $\overline{S_0}=\emptyset$, since there exists a basis not contained in $S_0$ by definition of the problem.

Assume now that $\cl(S_0)=S$ and that $\overline{S_0}$ is nonempty and has no loops.
The algorithm computes a maximum-weight basis $B_{0}$ of $M|S_0$
and a maximum-weight basis $B'$ of $M|\overline{S_0}$ by the greedy algorithm.
For every threshold $T$, the set $B_{0}[w\ge T]$ spans
$S_0[w\ge T]$. Hence
\[
\cl(B_{0}[w\ge T])=\cl(S_0[w\ge T]).
\]
Therefore, for each $f\in B'$, the value $T_f$ can be computed by
checking the smallest threshold $T$ such that
$B_{0}[w\ge T]\cup\{f\}$ is independent. It is enough to consider the candidate thresholds
\[
\{w(f)\}\cup\{w(e)+1\colon e\in B_{0},\ w(e)\ge w(f)\}.
\]
Indeed, the condition $f\notin\cl(S_0[w\ge T])$ is equivalent to
$f\notin\cl(B_{0}[w\ge T])$, so its truth value can change only when
$B_{0}[w\ge T]$ changes. Since the weights are integral and the
inequality is non-strict, an element $e\in B_{0}$ remains in
$B_{0}[w\ge T]$ at $T=w(e)$ and is removed only at the next integer
threshold, namely $T=w(e)+1$. Hence the relevant thresholds above
$w(f)$ are exactly those listed above.

By Theorem~\ref{thm:imnotexists}, the optimum value is
\[
\OPT=\min_{f\in B'}(T_f-w(f)).
\]
Thus the element $f^*$ chosen by the algorithm attains the optimum. If
$\OPT=0$, the algorithm returns $w^*=w$, which is optimal. If
$\OPT>0$, Theorem~\ref{thm:imnotexists} guarantees that setting
$w^*(f^*)=T_{f^*}$ and leaving all other weights unchanged is an optimal
solution. Therefore Algorithm~\ref{alg:notexistim} correctly solves
\textsc{IM-Not-Exists}.

\textit{Running time.} The feasibility checks (testing $\cl(S_0)=S$ and deleting loops from
$\overline{S_0}$) use $O(|S|)$ independence-oracle calls and $O(|S| \log |S|)$ time for sorting the elements. Computing $B_0$ and
$B'$ by the greedy algorithm uses $O(n)$ independence-oracle calls after
sorting the elements by nonincreasing weight.
There are at most $r+1$ candidate thresholds for each $f\in B'$, and the
condition $f\notin\cl(S_0[w\ge T])$ is monotone in $T$. Thus binary search
computes each $T_f$ using $O(\log r)$ independence-oracle calls. Since
$|B'|\le r$, this gives $O(r\log r)$ oracle calls for this step.
In total, the algorithm uses $O(n\log n)$ comparisons and
$O(n+r\log r)$ independence-oracle calls.
\end{proof}

\subsubsection{Proof of Theorem~\ref{thm:imnotexists}}

The proof has two steps: first we show that
$\OPT=\min_{f\in \overline{S_0}}(T_f-w(f))$, and then we show that this minimum can
be restricted to any maximum-weight basis of $M|\overline{S_0}$.

We first recall a standard optimality criterion for maximum-weight bases  (Corollary~\ref{cor:charact_optimum_basis_Closure}): for $w^*\in\bZ^S$, a basis $B$ is $w^*$-maximum if and only if $e\in \cl(B[w^*\ge w^*(e)])$ for every $e\in S\setminus B$. We then apply it to bases contained in $S_0$ in order to characterize feasible weight functions for \textsc{IM-Not-Exists}.

As a direct consequence of Lemma~\ref{lem:IM-Exists_closure_for_feas_weight}, for $w^*\in\bZ^S$, there exists a $w^*$-maximum basis contained in $S_0$
 if and only if $f\in \cl(S_0[w^*\ge w^*(f)])$ for every $f\in \overline{S_0}$. Hence
 $w^*$ is feasible for \textsc{IM-Not-Exists} if and only if there is
 $f\in \overline{S_0}$ such that $f\notin \cl(S_0[w^*\ge w^*(f)])$.

Now recall that for all $f\in \overline{S_0}$ we defined 
$T_f\coloneqq \min\{T\in\bZ\colon T\ge w(f)\text{ and }
f\notin \cl(S_0[w\ge T])\}.$ We show that $\delta^*$ can be obtained using this value.
    \begin{lemma}\label{lem:normal}
    Let $(M=(S,\cI),S_0,w,\|\cdot\|_1)$ be an instance of \textsc{IM-Not-Exists} and assume that $w$ is not feasible. Then
    \[
    \OPT = \min\{ T_f - w(f) \colon f \in \overline{S_0}\}.
    \]
     Moreover, if $\OPT > 0$, there exists $f^* \in \overline{S_0}$ for which $w'=w+\OPT\cdot \chi_{f^*}$ is an optimal weight function.
\end{lemma}

\begin{proof}
Let $\OPT$ be the optimal value of \textsc{IM-Not-Exists} and let $w^*$ be an
optimal weight function. Without loss of generality, we can assume
$w^*(f)\geq w(f)$ for $f\in \overline{S_0}$ and $w^*(e)\leq w(e)$ for
$e\in S_0$.

Since $w^*$ is feasible for \textsc{IM-Not-Exists}, no maximum $w^*$-weight
basis is contained in $S_0$. In particular, there exists a maximum
$w^*$-weight basis not contained in $S_0$, so $w^*$ is feasible for the
\textsc{IM-Outside} instance $(M=(S,\cI),S_0,w,\|\cdot\|_1)$ with the same
deviation $\|w-w^*\|_1=\OPT$. Applying Lemma~\ref{lem:relaxed-nice-form-L1},
we obtain a feasible solution $w'=w+\OPT\cdot\chi_f$ for some $f\notin S_0$,
with $\|w-w'\|_1=\OPT$, such that $f\notin\cl(S_0[w'\ge w'(f)])$.

Let $f'=\argmin\{T_g-w(g)\colon g\in\overline{S_0}\}$. We show that
$\OPT\ge T_{f'}-w(f')$. Since $w'=w+\OPT\chi_f$ and $f\notin S_0$, we have
$S_0[w'\ge w'(f)]=S_0[w\ge w(f)+\OPT]$. Thus
$f\notin\cl(S_0[w\ge w(f)+\OPT])$, and hence $T_f\le w(f)+\OPT$.
Therefore $T_{f'}-w(f')\le T_f-w(f)\le \OPT$, where the first inequality
follows from the choice of $f'$.

For the reverse inequality, consider
$\widetilde w=w+(T_{f'}-w(f'))\cdot\chi_{f'}$. We have
$S_0[\widetilde w\geq \widetilde w(f')]=S_0[w\geq T_{f'}]$, and by
definition of $T_{f'}$ we have $f'\notin\cl(S_0[w\geq T_{f'}])$. Hence
$\widetilde w$ is feasible by the closure condition, with deviation
$T_{f'}-w(f')=\min_{f\in\overline{S_0}}(T_f-w(f))$. Thus
$\OPT\le \min_{f\in\overline{S_0}}(T_f-w(f))$.
This proves the reverse inequality and hence the lemma.
\end{proof}

\begin{lemma}\label{lem:basis}
Let $B'$ be a maximum-weight basis of $M|\overline{S_0}$. Then
\[
\min_{f\in \overline{S_0}}(T_f-w(f))=\min_{f\in B'}(T_f-w(f)).
\]
\end{lemma}

\begin{proof}
Let $f\in \overline{S_0}$ minimize $T_f-w(f)$. If $f\in B'$, we are done. Otherwise,
Corollary~\ref{cor:charact_optimum_basis_Closure} applied in $M|\overline{S_0}$ gives
$f\in\cl(B'[w\ge w(f)])$. Let $T=T_f$. Since
$f\notin\cl(S_0[w\ge T])$, not every element of $B'[w\ge w(f)]$ belongs
to $\cl(S_0[w\ge T])$; otherwise
$f\in\cl(B'[w\ge w(f)])\subseteq\cl(S_0[w\ge T])$, a contradiction.
Thus some $b\in B'[w\ge w(f)]$ satisfies
$b\notin\cl(S_0[w\ge T])$.

If $w(b)\le T$, then $T$ is feasible for $b$, and
$T_b-w(b)\le T-w(b)\le T-w(f)=T_f-w(f)$. If $w(b)>T$, then
$S_0[w\ge w(b)]\subseteq S_0[w\ge T]$, so
$b\notin\cl(S_0[w\ge w(b)])$, and therefore $T_b=w(b)$. Hence again
$T_b-w(b)=0\le T_f-w(f)$. Thus some $b\in B'$ is no worse than $f$, and
equality follows.
\end{proof}

\begin{proof}[Proof of Theorem~\ref{thm:imnotexists}]
The equality $\OPT=\min_{f\in \overline{S_0}}(T_f-w(f))$ and the normal form follow
from Lemma~\ref{lem:normal}. By Lemma~\ref{lem:basis}, we get
$\OPT=\min_{f\in B'}(T_f-w(f))$. If $\OPT=0$, then the deviation-zero
solution is necessarily $w^*=w$. If $\OPT>0$, Lemma~\ref{lem:normal}
gives an optimal solution changing a single outside element, and
Lemma~\ref{lem:basis} allows this element to be chosen in $B'$.
\end{proof}

\section{\textsc{Inverse Matroid Not All} and \textsc{Not Only}}\label{sec:not-IM-All}

The last two variants negate the conditions defining \textsc{IM-All} and
\textsc{IM-Only}. The first asks that at least one basis contained in $S_0$
fail to be maximum-weight. The second allows this failure, or alternatively a
maximum-weight basis outside $S_0$.
Thus, our algorithms combine the one-basis separation idea from \textsc{IM-Not-All} with the outside-basis subproblem of Section~\ref{sec:relaxed-not-IM-Exists}.

In the negated version of \textsc{IM-All}, the goal is to modify the weight function so that at least one basis contained in $S_0$ is not of maximum weight. As before, to avoid dealing with $\varepsilon$-perturbations, we consider integer-valued weight functions.

\problemdef{Inverse Matroid Not All (IM-Not-All)}
    {A matroid $M = (S, \mathcal{I})$, a subset $S_0 \subseteq S$ that contains at least one basis and does not contain all bases, a weight function $w \in \mathbb{Z}^S$, and an objective function $\|\cdot\|$ defined on $\mathbb{R}^S$.}
    {Find a weight function $w^* \in \bZ^S$ such that 
    \begin{enumerate}[label=(\alph*),topsep=1px,partopsep=2px]\itemsep0em
        \item{\label{it:feas_cond_not-IM-All}} there exists a basis contained in $S_0$ that is not of maximum $w^*$-weight, and
        \item \label{it:opt_cond_not-IM-All} $\|w - w^*\|$ is minimized.
    \end{enumerate}} 

We say that a weight function is \emph{feasible} if condition~\ref{it:feas_cond_not-IM-All} holds. Note that \textsc{IM-Not-All} is always feasible: set $w'$ such that $w'(s)=0$ for all $s \in S_0$ and $w'(s)=1$ for all $s \in S \setminus S_0$.

\begin{remark}\label{rem:IM-NOT-ALL-feas}
    Given a weight function $w'$, we can verify in polynomial time whether $w'$ is feasible for \textsc{IM-Not-All}. To do so, find a minimum $w'$-weight basis $B_0$ of $M|S_0$ and a maximum $w'$-weight basis $A$ of $M$. Thus, $w'$ is feasible if and only if $w'(B_0) < w'(A)$.
\end{remark}

We also consider the negated version of \textsc{IM-Only}, which is formally defined as follows. This problem is also stated with integer-valued weight functions. 

\problemdef{Inverse Matroid Not Only (IM-Not-Only)}
    {A matroid $M=(S,\mathcal{I})$, a subset $S_0 \subseteq S$,
    that contains at least one basis and does not contain all bases,
    a weight function $w\in\bZ^S$, and an objective function $\|\cdot\|$ defined on $\bR^S$.}
    {Find a weight function $w^* \in \bZ^S$ such that 
    \begin{enumerate}[label=(\alph*),topsep=1px,partopsep=2px]\itemsep0em
        \item{\label{it:feas_cond_not-IM-Only}} either
            \begin{enumerate}[label=(a.\arabic*), topsep=1px,partopsep=2px,itemsep=0pt]
                \item{\label{it:feas_cond_not-IM-Only-1}} there exists a basis contained in $S_0$ that is not of maximum $w^*$-weight, or
                \item{\label{it:feas_cond_not-IM-Only-2}}  there exists a basis not contained in $S_0$ that has maximum $w^*$-weight, 
            \end{enumerate}
        \item[] and        
        \item \label{it:opt_cond_not-IM-Only} $\| w - w^*\|$ is minimized.
    \end{enumerate}}
    
We say that a weight function is \emph{feasible} if condition~\ref{it:feas_cond_not-IM-Only} holds. Note that a feasible weight function always exists: set $w'$ such that $w'(s)=0$ for all $s \in S_0$ and $w'(s)=1$ for all $s \in S \setminus S_0$.

\begin{remark}\label{rem:IM-Not-Only_feasible}
    Given a weight function $w'$, we can verify in polynomial time whether $w'$ is feasible for \textsc{IM-Not-Only}. This follows from noticing that condition~\ref{it:feas_cond_not-IM-Only-1} is exactly condition~\ref{it:feas_cond_not-IM-All} of \textsc{IM-Not-All}, and condition~\ref{it:feas_cond_not-IM-Only-2} is condition~\ref{it:feas_cond_relaxed} of \textsc{IM-Outside} when the input set is $S_0$.
\end{remark}

\subsection{Structural Properties}
\label{sec:prepnot}
For the $\ell_\infty$-norm, similarly to the case of \textsc{IM-All}, we will use the fact that \textsc{IM-Not-All} admits an optimal solution with a very restricted structure.

\begin{lemma}\label{lem:not}
Let $w_\textsc{Feas}$ be a feasible solution for the \textsc{IM-Not-All} 
instance $(M=(S,\cI),S_0,w,\|\cdot\|_\infty)$, and let $\delta=\|w-w_\textsc{Feas}\|_\infty$. Then there exists a feasible solution $w'_\textsc{Feas}$ satisfying the following:
\begin{enumerate}[label=(\alph*)]\itemsep0em
    \item $\|w-w'_\textsc{Feas}\|_\infty=\delta$, and \label{prop:aa}
    \item $w'_\textsc{Feas}=w + \delta\cdot (\chi_{S\setminus B_0}-\chi_{B_0})$ for some $B_0\subseteq S_0$.
\end{enumerate}
\end{lemma}
\begin{proof}
    Let $B_0\subseteq S_0$ be a basis that is not of maximum $w_\textsc{Feas}$-weight and let $B_{\max}$ be a basis of maximum $w_\textsc{Feas}$-weight -- such bases exist since $w_\textsc{Feas}$ is feasible. For $w'_\textsc{Feas}=w+\delta\cdot (\chi_{S\setminus B_0}-\chi_{B_0})$, we get 
    \begin{align*}
        w'_{\textsc{Feas}}(B_0)
        &=
        w'_{\textsc{Feas}}(B_0\setminus B_{\max})+w'_{\textsc{Feas}}(B_0\cap B_{\max})\\
        &\leq 
        w_{\textsc{Feas}}(B_0\setminus B_{\max})+w'_{\textsc{Feas}}(B_0\cap B_{\max})\\
        &< 
        w_{\textsc{Feas}}(B_{\max}\setminus B_0)+w'_{\textsc{Feas}}(B_0\cap B_{\max})\\
        &\leq 
        w'_{\textsc{Feas}}(B_{\max}\setminus B_0)+w'_{\textsc{Feas}}(B_0\cap B_{\max})\\
        &=
        w'_{\textsc{Feas}}(B_{\max}),
    \end{align*}
    where the strict inequality holds by $w_{\textsc{Feas}}(B_0)<w_{\textsc{Feas}}(B_{\max})$. That is, $w'_\textsc{Feas}$ is a feasible solution that satisfies the conditions of the lemma.
  \end{proof}

A similar statement holds for the $\ell_1$-norm: it suffices to decrease the weight of a single element of $S_0$ to obtain a feasible weight function.
\begin{lemma}\label{lem:notL1}
Let $w_\textsc{Feas}$ be a feasible solution for the \textsc{IM-Not-All} 
instance $(M=(S,\cI),S_0,w,\|\cdot\|_1)$, and let $\delta=\|w-w_\textsc{Feas}\|_1$. Then there exists a feasible solution $w'_\textsc{Feas}$ satisfying the following:
\begin{enumerate}[label=(\alph*)]\itemsep0em
    \item $\|w-w'_\textsc{Feas}\|_1=\delta$, and \label{prop:aaL1}
    \item $w'_\textsc{Feas}=w - \delta\cdot \chi_e$ for some element $e$ of some $B_0\subseteq S_0$.
\end{enumerate}
\end{lemma}
\begin{proof}
    Let $B_0\subseteq S_0$ be a basis that is not of maximum $w_{\textsc{Feas}}$-weight. By Proposition~\ref{prop:charact_optimum_basis}, there exist $f_0\in S\setminus B_0$ and $e_0\in C(B_0,f_0)\setminus\{f_0\}$ such that $w_{\textsc{Feas}}(f_0)>w_{\textsc{Feas}}(e_0)$. Since the weights are integral, $w_{\textsc{Feas}}(f_0)- w_{\textsc{Feas}}(e_0)\geq 1$. 
    Using $\delta=\|w-w_{\textsc{Feas}}\|_1$, we obtain
    \begin{align*} w(e_0)-w(f_0)+1 &\leq \bigl(w(e_0)-w_{\textsc{Feas}}(e_0)\bigr) + \bigl(w_{\textsc{Feas}}(f_0)-w(f_0)\bigr)\\ &\leq |w(e_0)-w_{\textsc{Feas}}(e_0)| + |w(f_0)-w_{\textsc{Feas}}(f_0)| \leq\delta. 
    \end{align*}

    Choose $(e,f)$ minimizing $w(e)-w(f)+1$ over all pairs satisfying $f\in S\setminus B_0$ and $e\in C(B_0,f)\setminus\{f\}$. Since $(e_0,f_0)$ is one such pair, \[ w(e)-w(f)+1\leq\delta. \]

Define $w'_{\textsc{Feas}}=w-\delta\chi_e$. By the definition of the fundamental circuit, $B=B_0-e+f$ is a basis. Moreover, \[ w'_{\textsc{Feas}}(B)-w'_{\textsc{Feas}}(B_0) = w(f)-w(e)+\delta \geq 1. \] Thus $B$ is strictly heavier than $B_0$ under $w'_{\textsc{Feas}}$. Hence $B_0$ is not a maximum $w'_{\textsc{Feas}}$-weight basis, and $w'_{\textsc{Feas}}$ is feasible. Finally, $\|w-w'_{\textsc{Feas}}\|_1=\delta$. \end{proof}
\subsection{Algorithm for \textsc{IM-Not-All} under \texorpdfstring{$\ell_\infty$}{l-infinity}- and \texorpdfstring{$\ell_1$}{l-1}-norms}
\label{sec:algnotall}

Our algorithm is presented as Algorithm~\ref{algo:NOT-IM-All} for both the $\ell_\infty$- and $\ell_1$-norms. The high-level idea for both norms is to distinguish among three possible cases: (1) the optimal value is $0$; (2) the optimal value is $1$; or (3) the optimal value exceeds $1$ and must be computed by identifying an appropriate pair consisting of an element $f \in S \setminus B_0$ and an element $e \in C(B_0,f)\setminus\{f\}$, where $B_0$ is the unique maximum-weight basis in this case. An illustration is shown in Figure~\ref{fig:example_IM-NOT-ALL}.

\begin{algorithm}[th!]
\caption{Algorithm for \textsc{IM-Not-All} under the $\ell_\infty$- and $\ell_1$-norms}\label{algo:NOT-IM-All}
\DontPrintSemicolon
\KwIn{An instance $(M=(S,\cI),\, S_0,\, w,\, \|\cdot\|)$ of \textsc{IM-Not-All}.}
\KwOut{An optimal weight function $w^*$ for \textsc{IM-Not-All}.}

\If{$w$ is feasible}{
    \Return{$w$}\;
}
Let $B_0 \subseteq S_0$ be a basis of maximum $w$-weight and fix $e' \in B_0$ that has minimum $w$-weight.\;
\uIf{$\|\cdot\| = \|\cdot\|_\infty$}{
    Set $w' \coloneqq w - \chi_{B_0}$.\;
    \lIf{$w'$ is feasible}{\Return{$w'$}}
    Set $\delta \coloneqq \min\bigl\{\lceil(w(e)-w(f)+1)/2\rceil
        \colon f \in S \setminus B_0,\; e \in C(B_0,f)\setminus\{f\}\bigr\}$.\;
    \Return{$w + \delta\cdot(\chi_{S \setminus B_0} - \chi_{B_0})$}\;
}
\ElseIf{$\|\cdot\| = \|\cdot\|_1$}{
    Set $w' \coloneqq w - \chi_{e'}$.\;
    \lIf{$w'$ is feasible}{\Return{$w'$}}
    Set $(e^*,f^*) \in \argmin\bigl\{w(e)-w(f)+1
        \colon f \in S \setminus B_0,\; e \in C(B_0,f)\setminus\{f\}\bigr\}$ and set $\delta \coloneqq w(e^*)-w(f^*)+1 $.\;
    \Return{$w - \delta \cdot \chi_{e^*}$}\;
}
\end{algorithm}
\begin{figure}[ht!]
\centering
    \begin{subfigure}[b]{0.33\textwidth}
        \centering
        \begin{tikzpicture}[scale=0.9,
            myEdge/.style={line width = 1.5pt},
            myEdgeZigzag/.style={line width = 1.5pt, decoration = {zigzag, segment length = 4pt, amplitude = 1pt},decorate},
            state/.style={circle,  minimum size=1em, draw, line width = 1.2pt}]   
            \node (a) at (0,2) [state] {$a$};
            \node (b) at (4,2) [state] {$b$};
            \node (c) at (1,1) [state] {$c$};
            \node (d) at (3,1) [state] {$d$};
            \node (e) at (0,0) [state] {$e$};
            \node (f) at (4,0) [state] {$f$}; 
            \draw [myEdge] (a) to node[midway, above] {3} (b) ; 
            \draw [myEdge] (e) to node[midway, below] {2} (f) ;
            \foreach \u \v \w \pos in { a/c/4/right, c/d/10/above, d/b/5/left,  e/c/4/right, d/f/5/left}{
                \draw [MyOrange!30, line width = 5pt] (\u) to (\v) ; 
                \draw [myEdgeZigzag] (\u) to node[midway,\pos] {\w} (\v) ; 
            } 
            \foreach \u \v \w \pos in {a/e/4/left,b/f/5/right}{
                \draw [MyOrange!30, line width = 5pt] (\u) to (\v) ; 
                \draw [myEdge] (\u) to node[midway,\pos] {\w} (\v) ; 
            }
        \end{tikzpicture}
        \caption{Original instance.}
        \label{fig:example_IM-not-all}
    \end{subfigure}%
    \hfill
    \begin{subfigure}[b]{0.33\textwidth}
        \centering
        \begin{tikzpicture}[scale=0.9,
            myEdge/.style={line width = 1.5pt},
            myEdgeZigzag/.style={line width = 1.5pt, decoration = {zigzag, segment length = 4pt, amplitude = 1pt},decorate},
            state/.style={circle,  minimum size=1em, draw, line width = 1.2pt}]   
            \node (a) at (0,2) [state] {$a$};
            \node (b) at (4,2) [state] {$b$};
            \node (c) at (1,1) [state] {$c$};
            \node (d) at (3,1) [state] {$d$};
            \node (e) at (0,0) [state] {$e$};
            \node (f) at (4,0) [state] {$f$}; 
            \draw [myEdge] (a) to node[midway, above] {3} (b) ; 
            \draw [myEdge] (e) to node[midway, below] {2} (f) ;
            \foreach \u \v \w \pos in { a/c/3/right, c/d/9/above, d/b/4/left,  e/c/3/right, d/f/4/left}{
                \draw [MyOrange!30, line width = 5pt] (\u) to (\v) ; 
                \draw [myEdge] (\u) to node[midway,\pos] {\w} (\v) ; 
            } 
            \foreach \u \v \w \pos in {a/e/4/left,b/f/5/right}{
                \draw [MyOrange!30, line width = 5pt] (\u) to (\v) ; 
                \draw [myEdge] (\u) to node[midway,\pos] {\w} (\v) ; 
            }
        \end{tikzpicture}
        \caption{Optimal weight for the $\ell_\infty$-norm.}
        \label{fig:IM-not-all-FINAL-linfty}
    \end{subfigure}%
    \hfill
    \begin{subfigure}[b]{0.33\textwidth}
        \centering
        \begin{tikzpicture}[scale=0.9,
            myEdge/.style={line width = 1.5pt},
            myEdgeZigzag/.style={line width = 1.5pt, decoration = {zigzag, segment length = 4pt, amplitude = 1pt},decorate},
            state/.style={circle,  minimum size=1em, draw, line width = 1.2pt}]   
            \node (a) at (0,2) [state] {$a$};
            \node (b) at (4,2) [state] {$b$};
            \node (c) at (1,1) [state] {$c$};
            \node (d) at (3,1) [state] {$d$};
            \node (e) at (0,0) [state] {$e$};
            \node (f) at (4,0) [state] {$f$}; 
            \draw [myEdge] (a) to node[midway, above] {3} (b) ; 
            \draw [myEdge] (e) to node[midway, below] {2} (f) ;
            \foreach \u \v \w \pos in { a/c/3/right, c/d/10/above, d/b/5/left,  e/c/4/right, d/f/5/left}{
                \draw [MyOrange!30, line width = 5pt] (\u) to (\v) ; 
                \draw [myEdge] (\u) to node[midway,\pos] {\w} (\v) ; 
            } 
            \foreach \u \v \w \pos in {a/e/4/left,b/f/5/right}{
                \draw [MyOrange!30, line width = 5pt] (\u) to (\v) ; 
                \draw [myEdge] (\u) to node[midway,\pos] {\w} (\v) ; 
            }
        \end{tikzpicture}
        \caption{Optimal weight for the $\ell_1$-norm.}
        \label{fig:IM-not-all-FINAL-L1}
    \end{subfigure}%
    \caption{Illustration of Algorithm~\ref{algo:NOT-IM-All} on the graphic matroid in Figure~\ref{fig:example_IM-not-all}. The basis $B_0$ is shown in zigzagged lines.
    Under both norms, the optimal value is $1$, and an optimal weight function is shown in Figures~\ref{fig:IM-not-all-FINAL-linfty} and~\ref{fig:IM-not-all-FINAL-L1} for the $\ell_\infty$- and $\ell_1$-norms, respectively.
    }
    \label{fig:example_IM-NOT-ALL}
\end{figure}

\begin{thm}\label{thm:NOT-IM-All-algorithm}
    Let $\|\cdot\|=\|\cdot\|_\infty$ or $\|\cdot\|=\|\cdot\|_1$. 
    Algorithm~\ref{algo:NOT-IM-All} determines an optimal weight function $w^*$ for the \textsc{IM-Not-All} instance  $( M=(S, \cI),  S_0, w, \|\cdot\|)$
    using $O(n\log n+r(n-r))$ elementary operations and $O(n+r(n-r))$ independence-oracle calls.
\end{thm}
\begin{proof}
    We have three cases.
    
    \textbf{Case 1.}  If $S_0$ contains a basis that is not of maximum $w$-weight, then $w$ is feasible and we are done.
    
    For the other two cases, all bases in $S_0$ have maximum weight. Let $B_0\subseteq S_0$ be an arbitrary basis contained in $S_0$. 
    
    \textbf{Case 2.} Assume first that there exists a basis $B_{\max}$ of maximum $w$-weight different from $B_0$. Note that $B_{\max}$ is allowed to be contained in $S_0$. 
    For both the $\ell_\infty$- and $\ell_1$-norms, we define a feasible weight function $w'$ with value $1$.
    First, we claim that $w'=w-\chi_{B_0}$ is a feasible solution for the $\ell_\infty$-norm. Indeed, $w'(B_0)=w(B_0)-|B_0| < w(B_{\max})-|B_{\max}\cap B_0|=w'(B_{\max})$, hence $w'$ is feasible. Similarly, for the $\ell_1$-norm, by Proposition~\ref{prop:symmetric-exchange}, choose $e\in B_0\setminus B_{\max}$ and $f\in B_{\max}\setminus B_0$ such that both $B_0-e+f$ and $B_{\max}-f+e$ are bases; then $w(e)=w(f)$. Set $w'\coloneqq w-\chi_e$. Thus $w'(B_0)=w(B_0)-1<w(B_{\max})=w'(B_{\max})$, so $w'$ is feasible; moreover, $w(e)-w(f)+1=1$ occurs in the final minimization, so the algorithm returns an optimal solution.

    \textbf{Case 3.} Suppose now that $B_0$ is the unique maximum $w$-weight basis of the matroid.
    
    Consider the $\ell_\infty$-norm. 
    If $\delta^*$ denotes the optimum value of the problem, then, by Lemma~\ref{lem:not}, $w^*=w+\delta^*(\chi_{S\setminus B_0}-\chi_{B_0})$ is an optimal solution. 
    Let us denote by $\cB_{\not \subseteq S_0}$ the set of bases not contained in $S_0$. The idea is to make a basis $B\in \cB_{\not \subseteq S_0}$ to be a maximum-weight basis with bigger weight than $B_0$.

For any basis $B\in\cB_{\not\subseteq S_0}$, we have
$w^*(B)-w^*(B_0)=w(B)-w(B_0)+\delta^*|B_0\triangle B|$. Hence,  $w^*(B)>w^*(B_0)$
    holds if and only if
    \[
    \delta^*>
    \frac{w(B_0)-w(B)}{|B_0\triangle B|}.
    \]
    Since all weights are integral, this is equivalent to
    \[
    \delta^*
    \ge
    \left\lceil
    \frac{w(B_0)-w(B)+1}{|B_0\triangle B|}
    \right\rceil.
    \]
    Therefore, the minimum feasible value is
    \[
    \min\Bigl\{
    \Bigl\lceil
    \frac{w(B_0)-w(B)+1}{|B_0\triangle B|}
    \Bigr\rceil
    \Bigm|
    B\in\cB_{\not \subseteq S_0}
    \Bigr\}.
    \]
    
    It remains to show that the above minimum is attained on a basis differing from $B_0$ by a single exchange.
    Let
    \[
    \delta'=
    \min\Bigl\{
    \frac{w(e)-w(f)}{2}
    \Bigm|
    e\in B_0,\;
    f\in S\setminus B_0,\;
    B_0-e+f\in \cB_{\not \subseteq S_0}
    \Bigr\}.
    \]
    Since $B_0$ is the unique maximum $w$-weight basis, we have $w(e)>w(f)$ for every such exchange, and therefore $\delta'>0$.
    
    Now let $B\in \cB_{\not \subseteq S_0}$ be arbitrary.
    By Proposition~\ref{prop:exchange}, there exists a bijection
    $\phi:B_0\setminus B\rightarrow B\setminus B_0$
    such that $B_0-e+\phi(e)$ is a basis for every $e\in B_0\setminus B$.
    Hence, $w(e)-w(\phi(e))\ge 2\delta'$ for every $e\in B_0\setminus B$.
    Summing over all $e\in B_0\setminus B$, we obtain
    \[
       w(B_0)-w(B)
    =
    \sum_{e\in B_0\setminus B}
    \bigl(w(e)-w(\phi(e))\bigr) 
    \ge
    2\delta'|B_0\setminus B|
    =
    \delta'|B_0\triangle B|.
    \]
    Thus, let $k=|B_0\setminus B|\geq 1$. Then $|B_0\triangle B|=2k$ and $w(B_0)-w(B)\geq 2k\delta'$. The term $+1$ is needed to make $B$ strictly heavier than $B_0$. Since $2\delta'\in\mathbb Z$, it follows that
    \[
        \left\lceil\frac{w(B_0)-w(B)+1}{|B_0\triangle B|}\right\rceil
        \geq \left\lceil\delta'+\frac{1}{2k}\right\rceil
        = \left\lceil\delta'+\frac{1}{2}\right\rceil.
    \]
    Equality is attained by a basis $B_0-e+f$ attaining the minimum in the definition of $\delta'$. Therefore, the optimum value is computed by the algorithm.
    
    For the $\ell_1$-norm, if $\delta^*$ denotes the optimum value of the problem then, by Lemma~\ref{lem:notL1}, $w^*=w-\delta^*\cdot\chi_{e}$ is an optimal solution for some $e\in B_0$.
    Then, by Proposition~\ref{prop:charact_optimum_basis}, it suffices to consider bases of the form $B=B_0-e+f$, where $f\notin B_0$ and $e\in C(B_0,f)\setminus\{f\}$. In this case, note that $w(e) \geq w(f)$.
        For such a basis,
        $w^*(B_0)=w(B_0)-\delta^* $
        and
        $w^*(B)=w(B_0)-w(e)+w(f)$.
        It follows that 
        $w^*(B)>w^*(B_0)$ holds
        if and only if
        $w(B_0)-w(e)+w(f) > w(B_0)-\delta^*$,
        which is equivalent to
        $\delta^*>w(e)-w(f)$.
        Since all weights are integral, this is equivalent to
        $\delta^*\geq w(e)-w(f)+1$.        
        Conversely, if
        $\delta=w(e)-w(f)+1$,
        then
        $w^*(B)-w^*(B_0)=w(f)-w(e)+\delta=1$,
        and therefore $B$ is strictly heavier than $B_0$ under $w^*$, showing that $w^*$ is feasible.
        It follows that, for a fixed element $e\in B_0$, the minimum feasible value is
        $\min \{w(e)-w(f)+1 \colon f\in S\setminus B_0,\; e\in C(B_0,f)\setminus\{f\}\}$.
        Minimizing over all choices of $e\in B_0$ yields
        exactly the value computed by the algorithm.

        \textit{Running time.}
        We analyze the running time for the $\ell_1$-norm, as the analysis for the $\ell_\infty$-norm is identical. By Remark~\ref{rem:IM-NOT-ALL-feas}, testing whether a weight function is feasible requires at most two executions of the Greedy algorithm, taking $O(|S|\log |S|)$ time and $O(|S|)$ independence-oracle calls in total. Moreover, computing the maximum-weight basis $B_0\subseteq S_0$ requires one additional execution of the Greedy algorithm, with the same complexity.
        
        If neither $w$ nor $w'$ is feasible, the algorithm computes $(e^*,f^*)$. By Remark~\ref{rem:findingFundCircuit}, for each element $f\in S\setminus B_0$, the fundamental circuit $C(B_0,f)$ can be computed using $O(r)$ independence-oracle calls. Thus, computing all fundamental circuits requires $O(r\,|S\setminus B_0|)$ independence-oracle calls. While computing these circuits, the algorithm simultaneously evaluates $w(e)-w(f)+1$ for each $e\in C(B_0,f)\setminus\{f\}$ and updates the minimum value. Hence, this phase requires $O(r\,|S\setminus B_0|)$ elementary operations.        
  \end{proof}

\subsection{Algorithm for \textsc{IM-Not-Only} under the \texorpdfstring{$\ell_\infty$}{l-infinity}- and \texorpdfstring{$\ell_1$}{l1}-norms}
\label{sec:algnotonly}
For the $\ell_\infty$-norm, Remark~\ref{rem:IM-Not-Only_feasible} implies that \textsc{IM-Not-Only} can be solved by solving one \textsc{IM-Not-All} instance and one \textsc{IM-Outside} instance. More precisely, Algorithm~\ref{algo:NOT-IM-All} computes a weight function $w_1$ satisfying~\ref{it:feas_cond_not-IM-Only-1} that minimizes $\|w-w_1\|_\infty$, while Algorithm~\ref{algo:relaxed-NOT-IM-Exists} is adapted to compute a weight function $w_2$ satisfying~\ref{it:feas_cond_not-IM-Only-2} that minimizes $\|w-w_2\|_\infty$ according to Corollary~\ref{cor:rine}. By Remark~\ref{rem:IM-Not-Only_feasible}, an optimal solution for \textsc{IM-Not-Only} is the one among $w_1$ and $w_2$ having smaller $\ell_\infty$-distance from $w$. Since the latter algorithm runs in $O(n\log n+r|S\setminus S_0|)$ elementary operations and $O(n+r|S\setminus S_0|)$ independence-oracle calls, and $|S\setminus S_0|\le |S\setminus B_0|$, the overall running time is dominated by Algorithm~\ref{algo:NOT-IM-All}.

\begin{cor}\label{cor:NOT-IM-Only-infty}
    Given an \textsc{IM-Not-Only} instance $(M=(S,\cI),S_0,w,\|\cdot\|_\infty)$, an optimal weight function can be computed 
     using $O(n\log n+r(n-r))$ elementary operations and $O(n+r(n-r))$ independence-oracle calls.
\end{cor}

For the $\ell_1$-norm, Remark~\ref{rem:IM-Not-Only_feasible} again implies that \textsc{IM-Not-Only} can be solved by solving one \textsc{IM-Not-All} instance and one \textsc{IM-Outside} instance. More precisely, Algorithm~\ref{algo:NOT-IM-All} computes a weight function $w_1$ satisfying~\ref{it:feas_cond_not-IM-Only-1} that minimizes $\|w-w_1\|_1$, while Algorithm~\ref{algo:IM-OUTSIDE-L1} computes a weight function $w_2$ satisfying~\ref{it:feas_cond_not-IM-Only-2} that minimizes $\|w-w_2\|_1$. By Remark~\ref{rem:IM-Not-Only_feasible}, an optimal solution for \textsc{IM-Not-Only} is the one among $w_1$ and $w_2$ having smaller $\ell_1$-distance from $w$. Since Algorithm~\ref{algo:IM-OUTSIDE-L1} runs in $O(n\log n+r\,|S\setminus S_0|)$ elementary operations and $O(n+r\,|S\setminus S_0|)$ independence-oracle calls, the overall running time is dominated by Algorithm~\ref{algo:NOT-IM-All}.

\begin{cor}\label{cor:NOT-IM-Only-L1}
    Given an \textsc{IM-Not-Only} instance $(M=(S,\cI),S_0,w,\|\cdot\|_1)$, an optimal weight function can be computed 
    using $O(n\log n+r(n-r))$ elementary operations and $O(n+r(n-r))$ independence-oracle calls.
\end{cor}

\section{Discussion}
\label{sec:discussion}

We have given a complete classification of six subset-constrained inverse matroid variants under the $\ell_\infty$- and $\ell_1$-norms, using integral formulations whenever strict separation may prevent attainment over the reals. The results show a sharp norm-dependent behavior: the $\ell_\infty$ versions admit combinatorial polynomial-time algorithms throughout, whereas under the $\ell_1$-norm the positive existence variant becomes strongly $\NP$-hard even for graphic matroids. At the same time, the remaining $\ell_1$ variants retain enough matroidal structure to be solved by exchange, closure, and minimum-cut arguments.

It is worth noting that our results can be extended to the more general setting where lower and upper bounds are given on the coordinates of the modified weight function. However, the corresponding algorithms become more involved due to numerous case distinctions, which would detract from the main structural insights that are more easily presented in the unconstrained setting. Moreover, some of the problems remain easy to solve under the weighted $\ell_\infty$-norm $\| p \|_{\omega} = \max \{ |p(s)|\cdot \omega(s) \colon s \in S \}$, where $\omega\in\bR_+^S$ represents the ``cost'' of modifying each coordinate of $w$.

We conclude the paper by mentioning some open questions:
\begin{enumerate}\itemsep0em
    \item Under the $\ell_\infty$-norm, all six variants admit polynomial-time algorithms, whereas the decision version of \textsc{IM-Exists} under the $\ell_1$-norm is already strongly $\NP$-complete. It would be interesting to investigate the complexity landscape of the considered problems under alternative distance measures.
    \item We introduced six variants of the \textsc{Inverse Matroid} problem, all of which can also be formulated for the intersection of two matroids. Designing efficient algorithms for this setting would be particularly interesting due to its applications to bipartite matchings and arborescences.
    \item Subset-constrained variants of \textsc{Inverse Optimization} can be considered in general, although certain settings may lead to $\NP$-hard problems. An intriguing direction is to identify structural assumptions on $\mathcal{F}$ or the underlying optimization problem under which the proposed variants remain tractable.
    \item The hardness of \textsc{IM-Exists} under the $\ell_1$-norm already holds for graphic matroids. It would be interesting to identify restricted graphic or matroid classes for which this problem becomes tractable, or to determine whether useful approximation guarantees are possible.
\end{enumerate}


\paragraph{Acknowledgment.} Research supported by the Lend\"ulet Programme of the Hungarian Academy of Sciences -- grant number LP2021-1/2021, by the Ministry of Innovation and Technology of Hungary from the National Research, Development and Innovation Fund -- grants ADVANCED 150556, 153096, and ELTE TKP 2021-NKTA-62, by Dynasnet European Research Council Synergy project -- grant number ERC-2018-SYG 810115, by ANID-Chile -- grant FONDECYT 1231669, and by Centro de Modelamiento Matem\'atico (CMM) through the BASAL fund FB210005 (ANID-Chile).

\paragraph{Tool and computational resource disclosure.} During the preparation of this manuscript, the authors used ChatGPT (GPT-5.5, OpenAI) and Claude (Anthropic) to assist with editing, grammar, and formatting, and with verifying and streamlining proofs.
The authors verified and reviewed all AI-assisted edits and take full responsibility for all content within this manuscript.

%
%
\bibliographystyle{splncs04}
\bibliography{inv_over_sets}

@article{berczi2023infinity,
    title={Newton-type Algorithms for Inverse Optimization: Weighted Bottleneck {H}amming Distance and $\ell_\infty$-norm Objectives}, 
    author={B{\'e}rczi, Krist{\'o}f and Mendoza-Cadena,Lydia Mirabel and Varga,Kitti},
    journal = {Optimization Letters},
    year    = {2025},
}

@InProceedings{berczi2024spanLATIN,
    author={B{\'e}rczi, Krist{\'o}f and Mendoza-Cadena, Lydia Mirabel and Varga, Kitti},
    editor={Soto, Jos{\'e} A. and Wiese, Andreas},
    title={Newton-Type Algorithms for Inverse Optimization: Weighted Span Objective},
    booktitle={LATIN 2024: Theoretical Informatics},
    shorttitle={LATIN 2024},
    year={2024},
    publisher={Springer Nature Switzerland},
    address={Cham},
    pages={334--347},
    isbn={978-3-031-55601-2},
}

@book{frank2011book,
  title={Connections in combinatorial optimization},
  author={Frank, Andr{\'a}s},
  volume={38},
  year={2011},
  publisher={OUP Oxford}
}

@article{heuberger2004inverse,
 title={Inverse combinatorial optimization: A survey on problems, methods, and results},
 author={Clemens Heuberger},
 journal={Journal of Combinatorial Optimization},
 volume={8},
 number={3},
 pages={329--361},
 year={2004},
}

@incollection{demange2014introduction,
 title={An introduction to inverse combinatorial problems},
 author={Marc Demange and J{\'e}r{\^o}me Monnot},
 booktitle={Paradigms of Combinatorial Optimization: Problems and New Approaches},
 publisher={John Wiley \& Sons, Inc.},
 edition={Second},
 pages={547--586},
 year={2014},
}

@article{chan2023inverse,
 title={Inverse Optimization: Theory and Applications},
 author={Timothy C. Y. Chan and Rafid Mahmood and Ian Yihang Zhu},
 journal={Operations Research},
 year={2023},
}

@article{mao1999inverse,
  title={Inverse problems of matroid intersection},
  author={Cai, Mao-Cheng},
  journal={Journal of combinatorial optimization},
  volume={3},
  number={4},
  pages={465--474},
  year={1999},
  publisher={Springer}
}

@article{dell2003base,
  title={The base-matroid and inverse combinatorial optimization problems},
  author={Dell'Amico, Mauro and Maffioli, Francesco and Malucelli, Federico},
  journal={Discrete applied mathematics},
  volume={128},
  number={2-3},
  pages={337--353},
  year={2003},
  publisher={Elsevier}
}

@article{zhang2016partialmatroid,
  title={Algorithms for the partial inverse matroid problem in which weights can only be increased},
  author={Zhang, Zhao and Li, Shuangshuang and Lai, Hong-Jian and Du, Ding-Zhu},
  journal={Journal of Global Optimization},
  volume={65},
  number={4},
  pages={801--811},
  year={2016},
  publisher={Springer}
}

@article{aman2016inverse,
    title={Inverse matroid optimization problem under the weighted Hamming distances},
    author={Aman, Massoud and Hassanpour, Hassan and Tayyebi, Javad},
    journal={Bulletin of the Transilvania University of Bra\c{s}ov. Series III: Mathematics, Informatics, Physics},
    pages={85--98},
    year={2016},
    volume={9(58)},
    number={2},
    }

@inproceedings{ahmadian2018algorithms,
  title={Algorithms for inverse optimization problems},
  author={Ahmadian, Sara and Bhaskar, Umang and Sanit{\`a}, Laura and Swamy, Chaitanya},
  booktitle={26th Annual European Symposium on Algorithms (ESA 2018)},
  year={2018},
  organization={Schloss Dagstuhl-Leibniz-Zentrum fuer Informatik}
}

@INPROCEEDINGS{tayybi2021inverse,
  author={Tayyebi, Javad and Bigdeli, Hamid},
  booktitle={2021 52nd Annual Iranian Mathematics Conference (AIMC)}, 
  title={Inverse matroid optimization problem under Chebyshev distance}, 
  year={2021},
  volume={},
  number={},
  pages={59-61},
}

@book{richter2016inverse,
 title={Inverse Problems: Basics, Theory and Applications in Geophysics},
 author={Mathias Richter},
 publisher={Birkh{\"a}user},
 year={2016},
}

@book{oxley2011matroid,
    author = {Oxley, James},
    title = {Matroid Theory},
    publisher = {Oxford University Press},
    year = {2011},
    month = {02},
    isbn = {9780198566946},
}

@techreport{gassner2009,
   author = {Gassner, Elisabeth},
   institution = {TU Graz},
   title = {Selected partial inverse combinatorial optimization problems with forbidden elements},
   year = {2009},
}

@article{gassner2010,
     author = {Gassner, Elisabeth},
     title = {The partial inverse minimum cut problem with {L1-norm} is strongly {NP-hard}},
     journal = {RAIRO - Operations Research - Recherche Op\'erationnelle},
     pages = {241--249},
     publisher = {EDP-Sciences},
     volume = {44},
     number = {3},
     year = {2010},
}

@article{li2016algorithm,
   author = {Shuangshuang Li and Zhao Zhang and Hong-Jian Lai},
   journal = {Theoretical Computer Science},
   month = {8},
   pages = {119-124},
   publisher = {Elsevier B.V.},
   title = {Algorithm for constraint partial inverse matroid problem with weight increase forbidden},
   volume = {640},
   year = {2016},
}

@inproceedings{dong2022partial,
 title={Partial Inverse Min-Max Spanning Tree Problem Under the Weighted Bottleneck {Hamming} Distance},
 author={Qingzhen Dong and Xianyue Li and Yu Yang},
 booktitle={AAIM 2022: Algorithmic Aspects in Information and Management},
 series={lecture Notes in Computer Science},
 volume={13513},
 publisher={Springer},
 pages={351--362},
 year={2022},
}

@techreport{lai2003complexity,
    title={The complexity of preprocessing},
    author={Lai, Tsung-Chyan and Orlin, James B.},
    journal={Research Report of Sloan School of Mangement},
    year={2003},
    institution={MIT}
}

@article{li2019capacited,
    author = {Li, Xianyue and Shu, Xichao and Huang, Huijing and Bai, Jingjing} ,
    title = {Capacitated partial inverse maximum spanning tree under the weighted Hamming distance},
    journal = {Journal of Combinatorial Optimization},
    year = {2019},
}

@article{li2020approximation,
  title={Approximation algorithms for capacitated partial inverse maximum spanning tree problem},
  author={Li, Xianyue and Zhang, Zhao and Yang, Ruowang and Zhang, Heping and Du, Ding-Zhu},
  journal={Journal of Global Optimization},
  volume={77},
  number={2},
  pages={319--340},
  year={2020},
  publisher={Springer},
}

@article{li2022partial,
  title={Partial inverse maximum spanning tree problem under the Chebyshev norm},
  author={Li, Xianyue and Yang, Ruowang and Zhang, Heping and Zhang, Zhao},
  journal={Journal of Combinatorial Optimization},
  volume={44},
  number={5},
  pages={3331--3350},
  year={2022},
  publisher={Springer},
}

@InProceedings{li2021capacited,
    author="Li, Xianyue
    and Yang, Ruowang
    and Zhang, Heping
    and Zhang, Zhao",
    editor="Du, Ding-Zhu
    and Du, Donglei
    and Wu, Chenchen
    and Xu, Dachuan",
    title="Capacitated Partial Inverse Maximum Spanning Tree Under the Weighted $\ell_\infty$-norm",
    booktitle="Combinatorial Optimization and Applications",
    year="2021",
    publisher="Springer International Publishing",
    address="Cham",
    pages="389--399",
}

@article{Ahuja2001Inverse,
   author = {Ahuja, Ravindra K. and Orlin, James B.},
   number = {5},
   journal = {Operations Research},
   month = {10},
   pages = {771--783},
   title = {Inverse Optimization},
   volume = {49},
   year = {2001},
}

@article{zhang2002general,
   author = {Zhang, Jianzhong and Liu, Zhenhong},
   issn = {13826905},
   number = {2},
   journal = {Journal of Combinatorial Optimization},
   pages = {207--227},
   title = {A General Model of Some Inverse Combinatorial Optimization Problems and Its Solution Method Under  $\ell_\infty$ Norm},
   volume = {6},
   year = {2002},
}

@article{zhang1999further,
    title = {A further study on inverse linear programming problems},
    journal = {Journal of Computational and Applied Mathematics},
    volume = {106},
    number = {2},
    pages = {345-359},
    year = {1999},
    issn = {0377-0427},
    author = {Zhang, Jianzhong and Liu, Zhenhong}
}

@inproceedings{cohen2016invisible,
  title={The invisible hand of dynamic market pricing},
  author={Cohen-Addad, Vincent and Eden, Alon and Feldman, Michal and Fiat, Amos},
  booktitle={Proceedings of the 2016 ACM Conference on Economics and Computation},
  pages={383--400},
  year={2016}
}

@article{berczi2021market,
  title={Market pricing for matroid rank valuations},
  author={B{\'e}rczi, Krist{\'o}f and Kakimura, Naonori and Kobayashi, Yusuke},
  journal={SIAM Journal on Discrete Mathematics},
  volume={35},
  number={4},
  pages={2662--2678},
  year={2021},
  publisher={SIAM}
}

@inproceedings{hsu2016prices,
  title={Do prices coordinate markets?},
  author={Hsu, Justin and Morgenstern, Jamie and Rogers, Ryan and Roth, Aaron and Vohra, Rakesh},
  booktitle={Proceedings of the Forty-Eighth Annual ACM Symposium on Theory of Computing},
  pages={440--453},
  year={2016}
}

@book{Garey1979computers,
  author    = {Michael R. Garey and David S. Johnson},
  title     = {Computers and Intractability: A Guide to the Theory of {NP}-Completeness},
  publisher = {W. H. Freeman and Company},
  address   = {San Francisco, CA},
  year      = {1979},
}

@article{krogdahl1977dependence,
  title={The dependence graph for bases in matroids},
  author={Krogdahl, Stein},
  journal={Discrete Mathematics},
  volume={19},
  number={1},
  pages={47--59},
  year={1977},
  publisher={Elsevier}
}

@article{frank2022discrete,
  title={A discrete convex min-max formula for box-TDI polyhedra},
  author={Frank, Andr{\'a}s and Murota, Kazuo},
  journal={Mathematics of Operations Research},
  volume={47},
  number={2},
  pages={1026--1047},
  year={2022},
  publisher={INFORMS}
}

@book{Guan2025InverseBook,
  author = {Guan, Xiucui and Pardalos, Panos M. and Zhang,Binwu},
  publisher = {Springer},
  title = {Inverse Combinatorial Optimization Problems},
  year = {2025},
  address = {Switzerland},
  issn = {1931-6828},
  doi={10.1007/978-3-031-91175-0}
}

@article{AhujaHochbaumOrlin2004,
  author       = {Ravindra K. Ahuja and
                  Dorit S. Hochbaum and
                  James B. Orlin},
  title        = {A Cut-Based Algorithm for the Nonlinear Dual of the Minimum Cost Network
                  Flow Problem},
  journal      = {Algorithmica},
  volume       = {39},
  number       = {3},
  pages        = {189--208},
  year         = {2004},
  doi          = {10.1007/S00453-004-1085-2},
}

\end{document}